\documentclass[letterpaper,12pt,titlepage,oneside,final]{report}

\newcommand{\href}[1]{#1}

\usepackage{ifthen}
\newboolean{ElectronicVersion}
\setboolean{ElectronicVersion}{true} 

\usepackage{geometry}

\usepackage{amsmath} 
\usepackage{amssymb} 
\usepackage{amsthm} 
\usepackage{bbm} 

\usepackage{eepic} 

\newcommand{\inner}[2]{\langle #1 , #2\rangle}
\newcommand{\Inner}[2]{\left\langle #1 , #2\right\rangle}
\newcommand{\defeq}{\stackrel{\smash{\textnormal{\tiny def}}}{=}}
\newcommand{\trans}{{\scriptstyle\mathsf{T}}}
\newcommand{\cls}[1]{\mathrm{#1}}

\DeclareMathOperator{\spn}{span}

\DeclareMathOperator{\fid}{F}

\newcommand{\kprod}[3]{#1_{#2\dots#3}}

\newtheorem{thm}{Theorem}[chapter]
\newenvironment{theorem}{\begin{thm}}{\end{thm}}
\newenvironment{numberedtheorem}[1]
{
  \renewcommand{\thethm}{#1}
  \begin{thm}
}{
  \end{thm}
  \addtocounter{thm}{-1}
}

\newtheorem{lemma}[thm]{Lemma}

\newtheorem{proposition}[thm]{Proposition}
\newtheorem{corollary}{Corollary}[thm]
\newtheorem{fact}[thm]{Fact}
\theoremstyle{definition}

\newtheorem{defn}[thm]{Definition}
\newenvironment{definition}{\begin{defn}}{\qed\end{defn}}

\newtheorem{exampl}[thm]{Example}
\newenvironment{example}{\begin{exampl}}{\qed\end{exampl}}
\newtheorem{remark}{Remark}[thm]

\newenvironment{enumerateroman}
{
  
  \begin{enumerate}
}{
  \end{enumerate}
  
}

\newcommand{\pa}[1]{(#1)}
\newcommand{\Pa}[1]{\left(#1\right)}
\newcommand{\br}[1]{\pa{#1}}
\newcommand{\Br}[1]{\Pa{#1}}
\newcommand{\set}[1]{\{#1\}}
\newcommand{\Set}[1]{\left\{#1\right\}}

\def\Jamiolkowski{J}
\newcommand{\jam}[1]{\Jamiolkowski\pa{#1}}

\DeclareMathOperator{\vectorize}{vec}
\newcommand{\col}[1]{\vectorize\pa{#1}}
\newcommand{\row}[1]{\vectorize\pa{#1}^*}
\newcommand{\Col}[1]{\vectorize\Pa{#1}}
\newcommand{\Row}[1]{\vectorize\Pa{#1}^*}

\DeclareMathOperator{\trace}{Tr}
\newcommand{\ptr}[2]{\trace_{#1}\pa{#2}}
\newcommand{\Ptr}[2]{\trace_{#1}\Pa{#2}}
\newcommand{\tr}[1]{\ptr{}{#1}}
\newcommand{\Tr}[1]{\Ptr{}{#1}}

\DeclareMathOperator{\contract}{contract}
\newcommand{\con}[2]{\contract[#1]\pa{#2}}
\newcommand{\Con}[2]{\contract[#1]\Pa{#2}}

\DeclareMathOperator{\rank}{rank}

\newcommand{\tinyspace}{\mspace{1mu}}
\newcommand{\abs}[1]{|\tinyspace#1\tinyspace|}
\newcommand{\Abs}[1]{\left|\tinyspace#1\tinyspace\right|}
\newcommand{\norm}[1]{\lVert\tinyspace#1\tinyspace\rVert}
\newcommand{\Norm}[1]{\left\lVert\tinyspace#1\tinyspace\right\rVert}
\newcommand{\fnorm}[1]{\norm{#1}_{\mathrm{F}}}
\newcommand{\Fnorm}[1]{\Norm{#1}_{\mathrm{F}}}
\newcommand{\tnorm}[1]{\norm{#1}_{\trace}}
\newcommand{\Tnorm}[1]{\Norm{#1}_{\trace}}
\newcommand{\dnorm}[1]{\norm{#1}_{\diamond}}
\newcommand{\Dnorm}[1]{\Norm{#1}_{\diamond}}
\newcommand{\snorm}[2]{\norm{#1}_{\diamond{#2}}}
\newcommand{\Snorm}[2]{\Norm{#1}_{\diamond{#2}}}

\newcommand{\fontmapset}{\mathbf} 

\newcommand{\mset}[2]{\fontmapset{#1}\pa{#2}}

\newcommand{\lin}[1]{\mset{L}{#1}}

\newcommand{\her}[1]{\mset{H}{#1}}

\newcommand{\pos}[1]{\mset{H^+}{#1}}

\newcommand{\identity}{\mathbbm{1}}
\newcommand{\idsup}[1]{\identity_{#1}}

\newcommand{\Real}{\mathbb{R}}
\newcommand{\Complex}{\mathbb{C}}

\def\noisy{\Delta}

\def\ot{\otimes}

\newcommand{\sub}[1]{{\downarrow}{#1}}
\newcommand{\strategy}[1]{\mathbf{#1}}
\newcommand{\costrategy}[1]{\mathrm{co}\textrm{-}{\strategy{#1}}}

\newcommand{\st}{\strategy{S}}
\newcommand{\cst}{\costrategy{S}}
\newcommand{\subst}{\sub{\st}}
\newcommand{\csubst}{\sub{\cst}}

\def\cA{\mathcal{A}}
\def\cB{\mathcal{B}}
\def\cC{\mathcal{C}}
\def\cD{\mathcal{D}}
\def\cE{\mathcal{E}}
\def\cF{\mathcal{F}}

\def\cH{\mathcal{H}}

\def\cV{\mathcal{V}}
\def\cW{\mathcal{W}}
\def\cX{\mathcal{X}}
\def\cY{\mathcal{Y}}
\def\cZ{\mathcal{Z}}

\def\bA{\mathbf{A}}
\def\bB{\mathbf{B}}
\def\bC{\mathbf{C}}
\def\bD{\mathbf{D}}

\def\bK{\mathbf{K}}

\def\bQ{\mathbf{Q}}

\def\bS{\mathbf{S}}
\def\bT{\mathbf{T}}

\def\rL{\mathrm{L}}

\def\rR{\mathrm{R}}

\def\thmchar{Characterization of strategies}
\def\theoreminnerproduct{Interaction output probabilities}

\usepackage{graphicx} 
\usepackage{geometry} 

\ifthenelse{\boolean{ElectronicVersion}}{
    \usepackage[letterpaper=true,bookmarks,pagebackref,
	plainpages=false, 
        pdfpagelabels=true, 
        pdftitle=Quantum\ Strategies\ and\ Local\ Operations, 
        pdfauthor=Gustav\ Gutoski	 
        ]{hyperref}}{}

\let\origdoublepage\cleardoublepage
\newcommand{\clearemptydoublepage}{%
  \clearpage{\pagestyle{empty}\origdoublepage}}
\let\cleardoublepage\clearemptydoublepage

\begin{document}

\pagestyle{empty}
\pagenumbering{roman}

\begin{titlepage}
        \begin{center}
        \vspace*{1.0cm}

        \Huge
        {\bf Quantum Strategies and Local Operations}

        \vspace*{1.0cm}

        \normalsize
        by \\

        \vspace*{1.0cm}

        \Large
        Gustav Gutoski \\

        \vspace*{3.0cm}

        \normalsize
        A thesis \\
        presented to the University of Waterloo \\ 
        in fulfillment of the \\
        thesis requirement for the degree of \\
        Doctor of Philosophy \\
        in \\
        Computer Science \\

        \vspace*{2.0cm}

        Waterloo, Ontario, Canada, 2009 \\

        \vspace*{1.0cm}

        \copyright\ Gustav Gutoski 2009 \\
        \end{center}
\end{titlepage}

\pagestyle{plain}
\setcounter{page}{2}

\ifthenelse{\boolean{ElectronicVersion}}{
  \noindent
  I hereby declare that I am the sole author of this thesis.  This is a true copy of the thesis, including any required final revisions, as accepted by my examiners.

  \bigskip
  
  \noindent
  I understand that my thesis may be made electronically available to the public.}{
 \noindent
 I hereby declare that I am the sole author of this thesis.

 \smallskip

 \noindent
 I authorize the University of Waterloo to lend this thesis to other
 institutions or individuals for the purpose of scholarly research.

 \bigskip

 \noindent
 I further authorize the University of Waterloo to reproduce this
 thesis by photocopying or by other means, in total or in part,
 at the request of other institutions or individuals for the purpose
 of scholarly research.
}
\newpage


\begin{center}\textbf{Abstract}\end{center}

This thesis is divided into two parts.
In Part \ref{part:strategies} we introduce a new formalism for quantum strategies, which specify the actions of one party in any multi-party interaction involving the exchange of multiple quantum messages among the parties.
This formalism associates with each strategy a single positive semidefinite operator acting only upon the tensor product of the input and output message spaces for the strategy.
We establish three fundamental properties of this new representation for quantum strategies and we list several applications, including a quantum version of von Neumann's celebrated 1928 Min-Max Theorem for zero-sum games and an efficient algorithm for computing the value of such a game.

In Part \ref{part:LOSE} we establish several properties of a class of quantum operations that can be implemented locally with shared quantum entanglement or classical randomness.
In particular, we establish the existence of a ball of local operations with shared randomness lying within the space spanned by the no-signaling operations and centred at the completely noisy channel.
The existence of this ball is employed to prove that the weak membership problem for local operations with shared entanglement is strongly $\cls{NP}$-hard.
We also provide characterizations of local operations in terms of linear functionals that are positive and ``completely'' positive on a certain cone of Hermitian operators, under a natural notion of complete positivity appropriate to that cone.
We end the thesis with a discussion of the properties of no-signaling quantum operations.

%
%
%
%

\newpage


\begin{center}\textbf{Acknowledgements}\end{center}

First and foremost, I thank my supervisor John Watrous for his support.
He is an outstanding supervisor and I am very lucky to have been given the opportunity to benefit from his guidance for the past six years.
I also thank the rest of my committee---Andrew Childs, Richard Cleve, Alex Russell, and Levent Tun\c{c}el---for taking the time to read this thesis and for their helpful suggestions.
Finally, I thank my parents for a lifetime of encouragement and support.
\newpage



\tableofcontents
\newpage




\pagenumbering{arabic}


\addcontentsline{toc}{chapter}{\textbf{Errata}}

\section*{Errata} \label{sec:errata}

\emph{(March 14, 2012)}

The proofs of Lemma \ref{lemma:unit-ball-technical} and Theorem \ref{thm:dnorm-dual-herm} in Chapter \ref{ch:norms} assume that if $\Pi$ is a projection and $P$ is positive semidefinite then $\Pi P\Pi\preceq P$.
This assumption is false and hence the proofs that employ it are invalid.

Nevertheless, the important claims of Chapter \ref{ch:norms} such as
Proposition \ref{prop:unit-ball} (Unit ball of the strategy $r$-norms) and Theorem \ref{thm:sep} (Distinguishability of convex sets of strategies) are true.
Corrected proofs of these claims can be found in a recent publication of the author \cite{Gutoski12}.

\chapter{Introduction} \label{ch:intro}

This thesis investigates two distinct topics of interest within the discipline of quantum information theory: quantum strategies and local operations with shared entanglement.
The discussion on quantum strategies is contained in Part \ref{part:strategies}, while the discussion on local operations with shared entanglement is given in Part \ref{part:LOSE}.

This introductory chapter provides a broad overview of the results of the thesis in Section \ref{sec:overview}.
A review of relevant background material from linear algebra, convex analysis, and quantum information is provided in Section \ref{sec:prelim}.

\section{Overview} \label{sec:overview}

In this section we provide a summary of the main contributions of the present thesis.
Mathematics and quantum formalism are invoked only informally so as to facilitate a broad description of results without getting bogged down in detail.

Except where otherwise noted, the content of this thesis is drawn from existing literature as follows:
\begin{itemize}

\item
Part \ref{part:strategies}, excluding Chapter \ref{ch:norms}, first appeared in preliminary form in Ref.~\cite{GutoskiW07}.

\item
Chapter \ref{ch:norms} is otherwise unpublished and due solely to the present author.

\item
Part \ref{part:LOSE} first appeared in Ref.~\cite{Gutoski09}.

\end{itemize}

\subsection{Quantum strategies} \label{subsec:intro:strategies}

In Part \ref{part:strategies} of this thesis we propose a new mathematical formalism for quantum strategies, prove several fundamental properties of this formalism, and provide several applications.

Informally, a \emph{quantum strategy} is a complete specification of the actions of one party in any multi-party interaction involving the exchange of one or more quantum messages among the parties.
Due to the generality of this notion, the potential for application of this formalism is very broad.
Indeed, our formalism should in principle apply to any framework that incorporates the exchange of quantum information among multiple entities, such as quantum cryptography, computational complexity, communication complexity, and distributed computation.

In this introductory section concerning quantum strategies, it is convenient to avoid cluttering discussion with historical background and citations.
Instead, the necessary background and references for each topic are covered in detail as they appear in the main body of the thesis.

\subsubsection{Three properties of the new formalism}

Chapter \ref{ch:strategies} is devoted to formal definitions of quantum strategies and a discussion thereof.
Under our new formalism, such a strategy is represented by a positive semidefinite operator $S$, the dimensions of which depend only upon the size of the messages exchanged in the interaction and \emph{not} upon the size of any memory workspace maintained by the strategy between messages.
(This distinction is important, as it permits us to consider strategies that call for an arbitrarily large memory workspace.)
We prove in Chapter \ref{ch:properties} that the set of all positive semidefinite operators which are valid representations of quantum strategies is characterized by a simple and efficiently-verifiable collection of linear equality conditions.

In order to extract useful classical information from such an interaction, a strategy will often call for one or more quantum measurements throughout the interaction.
In this case, the strategy is instead represented by a set $\set{S_a}$ of positive semidefinite operators indexed by all the possible combinations of outcomes of the measurements.
These strategies are called \emph{measuring strategies} and satisfy $\sum_a S_a = S$ for some ordinary (non-measuring) strategy $S$.
(By comparison, an ordinary POVM-type quantum measurement $\set{P_a}$ satisfies $\sum_a P_a = I$.)

We also prove in Chapter \ref{ch:properties} that the relationship between measuring strategies and other strategies is analogous to that between ordinary quantum measurements and quantum states.
In particular, basic quantum formalism tells us that for any ordinary quantum measurement $\set{P_a}$ with outcomes indexed by $a$ and any quantum state $\rho$ it holds that the probability with which the measurement $\set{P_a}$ yields a particular outcome $a$ when applied to a quantum system in state $\rho$ is given by the inner product \[ \Pr[\textrm{$\set{P_a}$ yields outcome $a$ on $\rho$}] = \inner{P_a}{\rho} = \ptr{}{P_a\rho}.\]
Similarly, we show that the probability with which a measuring strategy $\set{S_a}$ yields outcome $a$ after an interaction with a compatible quantum strategy $T$ is given by \[ \Pr[\textrm{$\set{S_a}$ yields outcome $a$ when interacting with $T$}] = \inner{S_a}{T} = \ptr{}{S_aT}.\]

Finally, we establish a convenient formula for computing the maximum probability with which a given measuring strategy $\set{S_a}$ can be forced to produce a given outcome $a$.
This probability is given by the \emph{minimum} real number $\lambda$ for which there exists an ordinary (non-measuring) strategy $Q$ with the property that the operator $\lambda Q - S_a$ is still positive semidefinite.

\subsubsection{Applications}

These three properties of quantum strategies---their linear characterization, inner product relationship, and formula for maximum output probability---open the door to a variety of new applications, several of which are presented in Chapter \ref{ch:applications}.

First and foremost, these properties pave the way for the first fully general theory of two-player zero-sum quantum games.
In particular, we prove a quantum analogue of von Neumann's famous 1928 Min-Max Theorem for zero-sum games.
We also show that the value of such a game can be expressed as the value of a semidefinite optimization problem and can therefore be efficiently approximated to arbitrary precision by standard algorithms for semidefinite optimization.

We then apply this newfound algorithm to computational complexity theory, establishing that the fundamental class $\cls{EXP}$ of decision problems that admit deterministic exponential-time classical solutions coincides with the exotic class $\cls{QRG}$ of decision problems that admit a quantum interactive proof with two competing provers.
That is, \[ \cls{QRG} = \cls{EXP}. \]
This equivalence is a rare characterization of a fundamental classical complexity class by a purely quantum complexity class.
As problems in $\cls{EXP}$ also admit \emph{classical} interactive proofs with competing provers, we obtain as a corollary the fact that quantum interactive proofs with competing provers provably contain no additional expressive power beyond that of classical interactive proofs with competing provers.
By contrast, it is widely believed, but not proven, that polynomial-time quantum computers are strictly more powerful than polynomial-time classical computers.

Elsewhere within the domain of complexity theory, we employ our new formalism to prove that many-message quantum interactive proofs with one prover or with two competing provers may be repeated multiple times in parallel so as to decrease the probability of error without increasing the number of messages exchanged among the parties in the interaction.

Finally, the aforementioned properties of quantum strategies are applied to yield an alternate and simplified proof of Kitaev's bound for strong quantum coin-flipping protocols.
Coin-flipping is a fundamental primitive arising in the study of cryptography in the context of secure two-party computation.
A \emph{coin-flipping} protocol is an interaction between to mutually untrusting parties who wish to agree on a random bit (a coin flip) via remote communication.
A \emph{strong coin-flipping protocol} with bias $\varepsilon$ has the property that two honest parties always produce a perfectly random coin toss, yet a dishonest party who attempts to force a given outcome upon an honest party can succeed with probability no more than $1/2 + \varepsilon$.
Kitaev showed that any strong coin-flipping protocol in which the parties exchange and process quantum information must have bias at least $1/\sqrt{2} - 1/2\approx 0.207$.
Kitaev's original proof of this fact relied upon the powerful machinery of semidefinite optimization duality.
In our proof, the complication of semidefinite optimization duality is successfully encapsulated in the properties of quantum strategies; what remains is a simple calculation that fits easily into half of a page.

It is noteworthy that a fourth application of the formalism of quantum strategies---beyond quantum game theory, complexity theory, and coin-flipping---appears in Chapter \ref{ch:NPhard} of Part \ref{part:LOSE} of this thesis.
In particular, the inner product relationship for quantum measuring strategies is employed to establish the $\cls{NP}$-hardness of weak membership testing for local operations with shared entanglement.

In addition to our results, other authors have used the formalism of quantum strategies in other areas, as we now describe.

\subsubsection{Independent development of the new formalism}

Our formalism and some of its properties were independently re-discovered by Chiribella, D'Ariano, and Perinotti \cite{ChiribellaD+08a, ChiribellaD+09a}.
What we call a ``quantum strategy,'' they call a ``quantum comb.''
In their initial publication on the subject, these authors prove an analogue of our Theorem \ref{thm:char} (\thmchar).
Moreover, our Theorem \ref{theorem:inner-product} (\theoreminnerproduct) is established in Refs.~\cite{ChiribellaD+08b, ChiribellaD+08d}.

These and other authors have provided several additional applications of quantum strategies to such problems as optimization of quantum circuits architecture \cite{ChiribellaD+08a}, cloning and learning of unitary operations \cite{ChiribellaD+08c, BisioC+09a}, and an impossibility proof for quantum bit commitment \cite{ChiribellaD+09b}, among others \cite{ChiribellaD+08e, BisioC+09b}.

\subsubsection{Distance measures}

After establishing useful properties of quantum strategies in Chapter \ref{ch:properties} and then applying those properties in Chapter \ref{ch:applications}, we return in Chapter \ref{ch:norms} to basic formalism for quantum strategies.
We define a new norm that captures the distinguishability of quantum strategies in the same sense that the trace norm for operators captures the distinguishability of quantum states or the diamond norm for super-operators captures the distinguishability of quantum operations.
Whereas the trace norm $\tnorm{\rho-\sigma}$ for quantum states $\rho,\sigma$ is given by \[ \tnorm{\rho-\sigma} = \max \Set{\inner{P_0-P_1}{\rho-\sigma} : \set{P_0,P_1} \textrm{ is a quantum measurement} },\]
we define the \emph{strategy $r$-norm} $\snorm{Q-R}{r}$ for quantum strategies $Q,R$ by \[ \snorm{Q-R}{r} = \max \Set{\inner{S_0-S_1}{Q-R} : \set{S_0,S_1} \textrm{ is a measuring strategy} }.\]
Here the subscript $r$ denotes the number of rounds of messages in the protocol for which $Q,R$ are strategies.
In particular, each positive integer $r$ induces a different strategy norm.
Our choice of notation is inspired by the fact that this norm is shown to coincide with the diamond norm for the case $r=1$.

Our primary application of the strategy norm is a generalization of a result of Ref.~\cite{GutoskiW05}, which states that for any two convex sets $\bA_0,\bA_1$ of quantum states there exists a \emph{fixed} quantum measurement that distinguishes \emph{any} two states chosen from these sets with probability that varies according to the minimal trace norm distance between the sets $\bA_0$ and $\bA_1$.
In other words, this measurement can be used to distinguish \emph{any} choices of states from $\bA_0,\bA_1$ at least as well as any measurement could distinguish the two \emph{closest} states from those sets.

Accordingly, our new result states that for any two convex sets $\bS_0,\bS_1$ of $r$-round quantum strategies, there exists a fixed measuring strategy that distinguishes any choices of strategies from those sets with probability according to the minimal distance between the sets $\bS_0$ and $\bS_1$ as measured by the new strategy $r$-norm.

We conclude Chapter \ref{ch:norms} with a discussion of the dual of the diamond norm for super-operators and its relation to the strategy norms.
The dual of the diamond norm $\dnorm{\Phi}^*$ of a super-operator $\Phi$ is defined by \[ \dnorm{\Phi}^* = \max_{\snorm{\Psi}=1} \abs{\inner{\Psi}{\Phi}} \] for some appropriate notion of inner product between super-operators.

While the diamond norm plays a fundamental role in the theory of quantum information, its dual has never been studied.
Hence, we establish several basic facts about this norm.
For example, the maximum in the definition of $\dnorm{\Phi}^*$ is achieved by a Hermitian-preserving super-operator whenever $\Phi$ is Hermitian-preserving, and by a completely positive super-operator whenever $\Phi$ is completely positive.
These facts are employed to show that a variant of the strategy 1-norm coincides with the dual of the diamond norm.
Thus, the strategy $r$-norms are shown to generalize both the diamond norm \emph{and} its dual.

\subsection{Local operations with shared entanglement}

In Part \ref{part:LOSE} we prove several properties of local operations with shared entanglement or randomness.
Informally, a \emph{local operation} is a quantum operation that can be implemented by distinct parties, each acting only upon his or her own portion of an overall quantum system.

In a \emph{local operation with shared randomness (LOSR)}, the parties are permitted to share random bits---say, common knowledge of an integer sampled at random according to some fixed probability distribution.
The parties may use their knowledge of this shared randomness to correlate their distinct operations.

In a \emph{local operation with shared entanglement (LOSE)}, rather than randomness, the parties are permitted to share among them distinct portions of some quantum system, the overall state of which may be entangled across the different parties.
The existence of such a shared state permits the parties to achieve correlations among their local operations that could not otherwise be achieved with randomness alone.

By contrast with Section \ref{subsec:intro:strategies}, it is expedient to include a discussion of historical background and citations in this introductory section concerning local operations with shared entanglement.

\subsubsection{Background}

LOSE operations are of particular interest in the quantum information community in part because any physical operation jointly implemented by spatially separated parties who obey both quantum mechanics and relativistic causality must necessarily be of this form.
Moreover, it is notoriously difficult to say anything meaningful about this class of operations, despite its simple definition.

One source of such difficulty stems from the fact that there exist two-party LOSE operations with the fascinating property that they cannot be implemented with any finite amount of shared entanglement \cite{LeungT+08}.
This difficulty has manifested itself quite prominently in the study of two-player co-operative games: classical games characterize $\cls{NP}$,\footnote{
As noted in Ref.~\cite{KempeK+07}, this characterization follows from the
PCP Theorem \cite{AroraL+98,AroraS98}.
}
%
%
%
%
whereas quantum games with shared entanglement are not even known to be computable. 
Within the context of these games, interest has focused largely on special cases of LOSE operations \cite{KobayashiM03,CleveH+04,Wehner06,CleveS+08,CleveGJ07,KempeR+07}, but progress has been made recently in the general case \cite{KempeK+07,KempeK+08,LeungT+08,DohertyL+08,NavascuesP+08,ItoK+09}.
In the physics literature, LOSE operations are often discussed in the context of \emph{no-signaling} operations \cite{BeckmanG+01,EggelingSW02,PianiH+06}.

In the present thesis some light is shed on the general class of multi-party LOSE operations, as well as the sub-class of LOSR operations.
Several distinct results are established, many of which mirror some existing result pertaining to separable quantum states.
What follows is a brief description of each result together with its analogue from the literature on separable states where appropriate.

\subsubsection{Ball around the identity}

\begin{description}

\item \emph{Prior work on separable quantum states.}
If $A$ is a Hermitian operator acting on a $d$-dimensional bipartite space and whose Frobenius norm at most 1 then the perturbation $I\pm A$ of the identity represents an (unnormalized) bipartite separable quantum state.
In other words, there is a ball of (normalized) bipartite separable states with radius $\frac{1}{d}$ in Frobenius norm centred at the completely mixed state $\frac{1}{d}I$ \cite{Gurvits02,GurvitsB02}.

A similar ball exists in the multipartite case, but with smaller radius.
In particular, there is a ball of $m$-partite separable $d$-dimensional states with radius $\Omega\Pa{2^{-m/2}d^{-1}}$ in Frobenius norm centred at the completely mixed state \cite{GurvitsB03}.
Subsequent results offer some improvements on this radius \cite{Szarek05,GurvitsB05,Hildebrand05}.

\item \emph{Present work on local quantum operations.}
An analogous result is proven in Chapter \ref{ch:ball} for multi-party LOSE and LOSR operations.
Specifically, if $A$ is a Hermitian operator acting on an $n$-dimensional $m$-partite space and whose Frobenius norm scales as $O\Pa{2^{-m}n^{-3/2}}$ then $I\pm A$ is the Choi-Jamio\l kowski representation of an (unnormalized) $m$-party LOSR operation.
As the unnormalized completely noisy channel \[ \noisy:X\mapsto\tr{X}I \] is the unique quantum operation whose Choi-Jamio\l kowski representation equals the identity, it follows that there is a ball of $m$-party LOSR operations (and hence also of LOSE operations) with radius $\Omega\Pa{2^{-m}n^{-3/2}d^{-1}}$ in Frobenius norm centred at the completely noisy channel $\frac{1}{d}\noisy$.
(Here the normalization factor $d$ is the dimension of the output system.)

The perturbation $A$ must lie in the space spanned by Choi-Jamio\l kowski representations of the no-signaling operations.
Conceptual implications of this technicality are discussed in the concluding remarks of Chapter \ref{ch:ball}.
No-signaling operations are discussed in Chapter \ref{ch:no-sig} of the present thesis, as summarized below.

\item \emph{Comparison of proof techniques.}
Existence of this ball of LOSR operations is established via elementary linear algebra.
By contrast, existence of the ball of separable states was originally established via a delicate combination of
\begin{enumerate}
\item[(i)]
the fundamental characterizations of separable states in terms of positive super-operators (described in more detail below), together with
\item[(ii)]
nontrivial norm inequalities for these super-operators.
\end{enumerate}

Moreover, the techniques presented herein for LOSR operations are of sufficient generality to immediately imply a ball of separable states without the need for the aforementioned characterizations or their accompanying norm inequalities.
This simplification comes in spite of the inherently more complicated nature of LOSR operations as compared to separable states.
It should be noted, however, that the ball of separable states implied by the present work is smaller than the ball established in prior work by a factor of $2^{-m/2}d^{-3/2}$.


\end{description}

\subsubsection{Weak membership problems are $\cls{NP}$-hard}

\begin{description}

\item \emph{Prior work on separable quantum states.}
The weak membership problem for separable quantum states asks,
\begin{quote}
  ``Given a description of a quantum state $\rho$
  and an accuracy parameter $\varepsilon$,
  is $\rho$ within distance $\varepsilon$ of a separable state?''
\end{quote}
This problem was proven strongly $\cls{NP}$-complete under oracle (Cook) reductions by Gharibian \cite{Gharibian08}, who built upon the work of Gurvits \cite{Gurvits02} and Liu \cite{Liu07}.
In this context, ``strongly $\cls{NP}$-complete'' means that the problem remains $\cls{NP}$-complete even when the accuracy parameter $\varepsilon=1/s$ is given in unary as $1^s$.

$\cls{NP}$-completeness of the weak membership problem was originally established by Gurvits \cite{Gurvits02}.
The proof consists of an $\cls{NP}$-completeness result for the weak \emph{validity} problem---a decision version of linear optimization over separable states---followed by an application of the 
Yudin-Nemirovski\u\i{} Theorem~\cite{YudinN76, GrotschelL+88}, which provides an oracle-polynomial-time reduction from weak validity to weak membership for general convex sets.

As a precondition of the Yudin-Nemirovski\u\i{} Theorem, the convex set in question (in our case, the set of separable quantum states) must contain a sufficiently large ball.
In particular, this $\cls{NP}$-completeness result relies crucially upon the existence of the aforementioned ball of separable quantum states.

For \emph{strong} $\cls{NP}$-completeness, it is necessary to employ a specialized ``approximate'' version of the Yudin-Nemirovski\u\i{} Theorem due to Liu~\cite[Theorem 2.3]{Liu07}.

\item \emph{Present work on local quantum operations.}
In Chapter \ref{ch:NPhard} it is proved that the weak membership problems for LOSE and LOSR operations are both strongly $\cls{NP}$-hard under oracle reductions.
The result for LOSR operations follows trivially from Gharibian (just take the input spaces to be empty).
But it is unclear how to obtain the result for LOSE operations without invoking the contributions of the present thesis.

The proof begins by observing that the weak validity problem for LOSE operations is merely a two-player quantum game in disguise and hence is strongly $\cls{NP}$-hard \cite{KempeK+07}.
The hardness result for the weak \emph{membership} problem is then obtained via a Gurvits-Gharibian-style application of Liu's version of the Yudin-Nemirovski\u\i{} Theorem, which of course depends upon the existence of the ball revealed in Section \ref{sec:balls}.

\end{description}

\subsubsection{Characterization in terms of positive super-operators}

\begin{description}

\item \emph{Prior work on separable quantum states.}
A quantum state $\rho$ of a bipartite system $\cX_1\ot\cX_2$ is separable if and only if the operator \[ \Pa{\Phi\otimes\identity_{\cX_2}}\pa{\rho} \] is positive semidefinite whenever the super-operator $\Phi$ is positive.
This fundamental fact was first proven in 1996 by Horodecki \emph{et al.}~\cite{HorodeckiH+96}.

The multipartite case reduces inductively to the bipartite case:
the state $\rho$ of an $m$-partite system is separable if and only if \( \Pa{\Phi\otimes\identity}\pa{\rho} \) is positive semidefinite whenever the super-operator $\Phi$ is positive on $(m-1)$-partite separable operators \cite{HorodeckiH+01}.

\item \emph{Present work on local quantum operations.}
In Chapter \ref{ch:char} it is proved that a multi-party quantum operation $\Lambda$ is a LOSE operation if and only if \[ \varphi\pa{\jam{\Lambda}}\geq 0 \] whenever the linear functional $\varphi$ is ``completely'' positive on a certain cone of separable Hermitian operators, under a natural notion of complete positivity appropriate to that cone.
A characterization of LOSR operations is obtained by replacing \emph{complete} positivity of $\varphi$ with mere positivity on that \emph{same} cone.
Here $\jam{\Lambda}$ denotes the Choi-Jamio\l kowski representation of the super-operator $\Lambda$.

The characterizations presented in Chapter \ref{ch:char} do not rely upon any of the prior discussion in this thesis.
This independence contrasts favorably with prior work on separable quantum states, wherein the existence of the ball around the completely mixed state (and the subsequent $\cls{NP}$-hardness result) relied crucially upon the characterization of separable states in terms of positive super-operators.

\end{description}

\subsubsection{No-signaling operations}

A quantum operation is \emph{no-signaling} if it cannot be used by spatially separated parties to violate relativistic causality.
By definition, every LOSE operation is a no-signaling operation.
Moreover, there exist no-signaling operations that are not LOSE operations.
Indeed, it is noted in Chapter \ref{ch:no-sig} that the standard nonlocal box of Popescu and Rohrlich \cite{PopescuR94} is an example of a no-signaling operation that is separable (in the sense of Rains \cite{Rains97}), yet it is not a LOSE operation.

Two characterizations of no-signaling operations are also discussed in Chapter \ref{ch:no-sig}.
These characterizations were first established somewhat implicitly for the bipartite case in Beckman \emph{et al.}~\cite{BeckmanG+01}.
The present thesis generalizes these characterizations to the multi-party setting and recasts them more explicitly in terms of the Choi-Jamio\l kowski representation for quantum super-operators.

\section{Mathematical preliminaries} \label{sec:prelim}

In this section we summarize the background mathematical knowledge upon which the work in this thesis rests.
While extensive, by no means is this summary intended to be comprehensive.
Instead, the purpose of this section is only to review existing concepts so that we may fix terminology and notation throughout the thesis.

\subsection{Linear algebra}
\label{sec:intro:linalg}

\subsubsection{Vectors, operators, and inner products}

The vector space $\Complex^n$ of all $n$-tuples of complex numbers is called a \emph{complex Euclidean space}.
Complex Euclidean spaces are denoted by capital script letters such as $\cX$, $\cY$, and $\cZ$ so as to better facilitate discussions involving multiple distinct spaces.

Vectors in a complex Euclidean space are denoted by lowercase Roman letters such as $u$, $v$, and $w$.
The \emph{standard orthonormal basis} of each $n$-dimensional complex Euclidean space is typically written $\set{e_1,\dots,e_n}$ where $e_i$ denotes the $n$-tuple whose $i$th component equals 1 with all other components equal to 0.
The \emph{standard inner product} between two vectors $u,v\in\cX$ is denoted $\inner{u}{v}$.
In this thesis, this inner product is conjugate linear in the first argument and linear in the second argument.
In particular, if
\begin{align*}
u &= (\alpha_1,\dots,\alpha_n)\\
v &= (\beta_1,\dots,\beta_n)
\end{align*}
then \[ \inner{u}{v} = \sum_{i=1}^n \overline{\alpha_i}\beta_i. \]
The \emph{standard Euclidean norm} of a vector $u$ is denoted $\norm{u}$ and is given by \[ \norm{u}=\sqrt{\inner{u}{u}}. \]
Each vector $u\in\cX$ induces a \emph{dual vector} $u^*$, which is a linear function $u^*: \cX\to\Complex$ defined by
\[ u^* : x\mapsto\inner{u}{x}. \]
We sometimes use the alternate notation \[u^*v=\inner{u}{v}.\]

The vector $y\in\cY$ obtained by applying a linear operator $A:\cX\to\cY$ to a vector $x\in\cX$ is denoted by the simple juxtaposition $y=Ax$.
Similarly, the operator $C:\cX\to\cZ$ obtained by composing the linear operators $A:\cX\to\cY$ and $B:\cY\to\cZ$ is denoted by the juxtaposition $C=BA$.

Each pair of vectors $x\in\cX,y\in\cY$ induces an operator $yx^*:\cX\to\cY$ defined by
\[ (yx^*)u = y(x^*u) = \inner{x}{u} y \]
for all $u\in\cX$.
By analogy with the inner product, the operator $yx^*$ is sometimes called the \emph{outer product} of $x$ and $y$, but we will not use that terminology in this thesis.

The \emph{adjoint} of a linear operator $A:\cX\to\cY$ is the unique operator $A^*:\cY\to\cX$ satisfying \[ \inner{y}{Ax}=\inner{A^*y}{x} \] for every $x\in\cX,y\in\cY$.\
The \emph{standard inner product} for operators is given by
\[ \inner{A}{B} = \ptr{}{A^*B}. \]

\subsubsection{Matrix representation of vectors and operators}

It is sometimes convenient to think of a linear operator $A:\cX\to\cY$ as an $m\times n$ matrix and of an element $x\in \cX$ as a $n\times 1$ column vector, so that the $m\times 1$ column vector $y\in\cY$ obtained by applying $A$ to $x$ is given by standard matrix multiplication:
\[
  \left[
    \begin{array}{c}
      \delta_1\\ \vdots \\ \delta_m
    \end{array}
  \right]
  = y = Ax =
  \left[
    \begin{array}{ccc}
      a_{1,1}&\dots &a_{1,n}\\
      \vdots & \ddots & \vdots\\
      a_{m,1}&\dots &a_{m,n}
    \end{array}
  \right]
  \left[
    \begin{array}{c}
      \gamma_1\\ \vdots \\ \gamma_n
    \end{array}
  \right]
\]
In this view, the asterisk superscript indicates the conjugate-transpose matrix operation.
In particular, the \emph{conjugate} operator $\overline{A}:\cX\to\cY$ and \emph{transpose} operator $A^\trans:\cY\to\cX$ are the linear operators whose matrix representations are given by
\[
  \overline{A}=
  \left[
    \begin{array}{ccc}
      \overline{a_{1,1}}&\dots &\overline{a_{1,n}}\\
      \vdots & \ddots & \vdots\\
      \overline{a_{m,1}}&\dots &\overline{a_{m,n}}
    \end{array}
  \right],
  \qquad
  A^\trans=
  \left[
    \begin{array}{ccc}
      a_{1,1}&\dots &a_{m,1}\\
      \vdots & \ddots & \vdots\\
      a_{1,n}&\dots &a_{n,m}
    \end{array}
  \right]
\]
and the matrix representation of the adjoint operator $A^*$ is given by \[ A^* = (\overline{A})^\trans = \overline{(A^\trans)}. \]
The asterisk notation for matrices is consistent with the definition of the dual vector.
Indeed, matrix multiplication can be used to compute the vector inner product:
\[
  u^*v = 
  \left[
    \begin{array}{ccc}
      \overline{\alpha_1} & \dots & \overline{\alpha_n}
    \end{array}
  \right]
  \left[
    \begin{array}{c}
      \beta_1\\ \vdots \\ \beta_n
    \end{array}
  \right]
  = \sum_{i=1}^n \overline{\alpha_i}\beta_i.
\]
Matrix multiplication is also employed to compute the matrix representation of the operator $yx^*:\cX\to\cY$:
\[
  yx^* = 
  \left[
    \begin{array}{c}
      \delta_1\\ \vdots \\ \delta_m
    \end{array}
  \right]
  \left[
    \begin{array}{ccc}
      \overline{\gamma_1} & \dots & \overline{\gamma_n}
    \end{array}
  \right]
  =
  \left[
    \begin{array}{ccc}
      \delta_1\overline{\gamma_1} &\dots & \delta_1\overline{\gamma_n}\\
      \vdots & \ddots & \vdots\\
      \delta_m\overline{\gamma_1} &\dots & \delta_m\overline{\gamma_n}
    \end{array}
  \right].  
\]
Similarly, matrix multiplication can be used to compute the matrix inner product---the composition $A^*B$ of linear operators is computed by matrix multiplication of $A^*$ and $B$, and the familiar trace function $\ptr{}{A^*B}$ equals the sum of the diagonal entries of the resulting matrix $A^*B$.


\subsubsection{Tensor products}

The tensor product is a mathematical concept that is so fundamental to quantum information that discussion always implicitly assumes a working knowledge.
The concept is important because it is the mechanism by which two separate quantum systems are viewed as a single system.

For any two complex Euclidean spaces $\cX,\cY$ and any two vectors $x\in\cX,y\in\cY$, the \emph{tensor product} associates a third vector $x\ot y$, which is an element of a third vector space $\cX\ot\cY$ of dimension $\dim(\cX)\dim(\cY)$.
Specifically, letting
\begin{align*}
  \set{e_0,\dots,e_{\dim(\cX)-1}}&\subset\cX,\\
  \set{f_0,\dots,f_{\dim(\cY)-1}}&\subset\cY,\\
  \set{g_0,\dots,g_{\dim(\cX\ot\cY)-1}}&\subset\cX\ot\cY
\end{align*}
denote the standard bases of $\cX$, $\cY$, and $\cX\ot\cY$, respectively, the vector $x\ot y$ is specified in terms of standard basis elements as follows:
\[ e_i\ot f_j = g_{i\dim(\cY)+j}. \]
The complex Euclidean space $\cX\ot\cY$ is called the \emph{tensor product} of the spaces $\cX$ and $\cY$.
(Indeed, any complex Euclidean space $\cZ$ whose dimension is not a prime number may be viewed as a tensor product $\cZ=\cX\ot\cY$ of two nontrivial spaces $\cX$ and $\cY$.)

From this definition, it is possible to derive many of the widely known basic properties of the tensor product.
For example, the bilinear mapping defined by \[(x,y)\mapsto x\ot y \] is \emph{universal}, meaning that \emph{any} bilinear transformation $\psi:\cX\times\cY\to\cV$ for some complex Euclidean space $\cV$ can alternately be written as a linear operator $A_\psi:\cX\ot\cY\to\cV$ satisfying \[ \psi(x,y) = A_\psi(x\ot y). \]

As the set of all linear operators from one complex Euclidean space to another is itself a complex vector space, the tensor product may be extended in the obvious way to linear operators.
From there it is possible to derive all the widely known properties of the tensor product for operators.
We shall not list these properties here.

The tensor product of two or more vectors, spaces, or operators is defined inductively, as suggested by the expression
\[ \cX \ot \cY \ot \cZ = \cX \ot (\cY\ot\cZ). \]
As the tensor product operation is associative, there is no ambiguity in writing simply $\cX\ot\cY\ot\cZ$.

Under the matrix representation, properties of the tensor product may be derived from the following straightforward definition.
For $2\times 2$ matrices
\[
  A=
  \left[
    \begin{array}{cc}
      a & b \\ c & d
    \end{array}
  \right]
  ,\qquad P=
  \left[
    \begin{array}{cc}
      p & q \\ r & s
    \end{array}
  \right]
\]
we have
\[
  A\ot P = 
  \left[
    \begin{array}{cc}
    aP & bP \\ cP & dP
    \end{array}
  \right]
  =
  \left[
    \begin{array}{cc}
    a
    \left[
      \begin{array}{cc}
        p & q \\ r & s
      \end{array}
    \right]
    &
    b
    \left[
      \begin{array}{cc}
        p & q \\ r & s
      \end{array}
    \right]
    \\
    c
    \left[
      \begin{array}{cc}
        p & q \\ r & s
      \end{array}
    \right]    
    &
    d
    \left[
      \begin{array}{cc}
        p & q \\ r & s
      \end{array}
    \right]    
    \end{array}
  \right]
  =
  \left[
    \begin{array}{cccc}
    ap & aq & bp & bq \\
    ar & as & br & bs \\
    cp & cq & dp & dq \\
    cr & cs & dr & ds
    \end{array}
  \right].
\]
This definition extends in the obvious way to arbitrary matrices of any dimension, including column vectors and non-square matrices.
In this context, the tensor product might also be called the \emph{Kronecker product}.

Often in this thesis our discussion involves tensor products of finite sequences of vectors, spaces, operators, and so on.
As such, it is convenient to adopt the following shorthand for such a product: if $X_1,\dots,X_m$ are arbitrary operators then we define
\[ \kprod{X}{i}{j} \defeq X_i\ot\cdots\ot X_j \]
for integers $1\leq i\leq j \leq m$.
A similar notation shall be used for vectors and complex Euclidean spaces.

\subsubsection{Sets of operators}

A linear operator $A:\cX\to\cY$ is called an \emph{isometry} if it holds that $\norm{Ax}=\norm{x}$ for all $x\in\cX$.
This condition can only be met when $\dim(\cY)\geq\dim(\cX)$.
When $\cX$ and $\cY$ have equal dimension, an isometry is also called a \emph{unitary} operator.

The (complex) vector space of linear operators of the form $A:\cX\to\cX$ is denoted $\lin{\cX}$.
The identity operator in $\lin{\cX}$ is denoted $I_\cX$ and the subscript is dropped whenever the space $\cX$ is clear from the context.

An element $A$ of $\lin{\cX}$ is \emph{Hermitian} (or \emph{self-adjoint}) if $A^*=A$.
The set of all Hermitian operators within $\lin{\cX}$ forms a (real) vector space, which we denote by $\her{\cX}$.

An operator $A\in\lin{\cX}$ is \emph{positive semidefinite} if $u^*Au$ is a nonnegative real number for each vector $u\in\cX$.
Every positive semidefinite operator is also Hermitian, and the set of all positive semidefinite operators within $\her{\cX}$ is denoted $\pos{\cX}$.
In general, the use of bold font is reserved for sets of operators.
For an arbitrary set $\bS\subset\her{\cX}$ of Hermitian operators, we let \[\bS^+=\bS\cap\pos{\cX}\] denote the set of positive semidefinite elements in $\bS$.

We adopt the notation $P\succeq 0$ to indicate that the operator $P$ is positive semidefinite.
As suggested by this notation, the \emph{semidefinite partial ordering} on Hermitian operators is defined so that $P\succeq Q$ if and only if $P-Q$ is positive semidefinite.
For each positive semidefinite operator $P$ there exists a unique positive semidefinite operator $\sqrt{P}$ called the \emph{square root} of $P$ with the property that $(\sqrt{P})^2=P$.

\subsubsection{Super-operators}

A \emph{super-operator} is a linear operator of the form $\Phi:\lin{\cX}\to\lin{\cY}$.
The identity super-operator from $\lin{\cX}$ to itself is denoted $\idsup{\cX}$ and the subscript is dropped at will.

Whereas the application of an operator to a vector is denoted by simple juxtaposition, the operator $Y\in\lin{\cY}$ obtained by applying the super-operator $\Phi:\lin{\cX}\to\lin{\cY}$ to an operator $X\in\lin{\cX}$ is always denoted with parentheses: $Y=\Phi(X)$.
Similarly, whereas operator composition is denoted by simple juxtaposition, the super-operator $\Gamma:\lin{\cX}\to\lin{\cZ}$ obtained by composing the linear operators $\Phi:\lin{\cX}\to\lin{\cY}$ and $\Psi:\lin{\cY}\to\lin{\cZ}$ is denoted $\Gamma=\Psi\circ\Phi$.

The \emph{standard inner product} for super-operators $\Phi,\Psi:\lin{\cX}\to\lin{\cY}$ is given by \[ \inner{\Phi}{\Psi} = \sum_{i,j=1}^{\dim(\cX)} \inner{\Phi(e_ie_j^*)}{\Psi(e_ie_j^*)} \]
where $\set{e_1,\dots,e_{\dim(\cX)}}\subset\cX$ is the standard basis for $\cX$.
That this definition is a natural extension of the operator inner product can be argued from the fact that the operator inner product satisfies \[ \inner{A}{B} = \sum_{i=1}^{\dim(\cX)} \inner{Ae_i}{Be_i} \] for any two operators $A,B:\cX\to\cY$.

Just as the definition of the tensor product is extended in a natural way to operators, so too can it be extended to super-operators.

A super-operator $\Phi:\lin{\cX}\to\lin{\cY}$ is said to be
\begin{itemize}
\item
  \emph{Hermitian-preserving} if $\Phi(X)$ is Hermitian whenever $X$ is Hermitian.
\item
  \emph{positive on $\bK$} if $\Phi(X)$ is positive semidefinite whenever $X\in\bK$.
\item
  \emph{positive} if $\Phi$ is positive on $\pos{\cX}$.
\item
  \emph{completely positive} if $\Phi\ot\idsup{\cZ}$ is positive for every choice of complex Euclidean space $\cZ$.
\item
  \emph{trace-preserving} if $\ptr{}{\Phi(X)}=\ptr{}{X}$ for all $X$.
\end{itemize}
While the definition of complete positivity might seem awkward at first, it is significantly simplified by the observation that $\Phi$ is completely positive if and only if $\Phi\ot\idsup{\cZ}$ is positive for a space $\cZ$ with $\dim(\cZ)=\dim(\cX)$.
In particular, there is no need to verify positivity of $\Phi\ot\idsup{\cZ}$ for infinitely many spaces $\cZ$.

An important example of a completely positive and trace-preserving super-operator is the \emph{partial trace}.
For any complex Euclidean spaces $\cX,\cY$, this super-operator has the form $\trace_\cX:\lin{\cX\ot\cY}\to\lin{\cY}$.
It is most easily specified by its actions upon product operators $X\ot Y$ for $X\in\lin{\cX},Y\in\lin{\cY}$ by the expression \[ \ptr{\cX}{X\ot Y} = \ptr{}{X} Y. \]
The importance of the partial trace stems from the fact that this super-operator is used to compute the quantum state of a portion of some larger system whose quantum state is already known.

Just as with operators, the \emph{adjoint} of a super-operator $\Phi:\lin{\cX}\to\lin{\cY}$ is the unique super-operator $\Phi^*:\lin{\cY}\to\lin{\cX}$ satisfying \[ \inner{Y}{\Phi(X)}=\inner{\Phi^*(Y)}{X} \] for every $X\in\lin{\cX},Y\in\lin{\cY}$.

Every super-operator $\Phi:\lin{\cX}\to\lin{\cY}$ may be expressed in the \emph{operator-sum} notation, whereby there exist operators $A_1,\dots,A_k,B_1,\dots,B_k:\cX\to\cY$ such that \[ \Phi(X) = \sum_{i=1}^k A_iXB_i^* \] for all $X$, from which it follows that \[ \Phi^*(Y) = \sum_{i=1}^k A_i^*YB_i \] for all $Y$.
It holds that $\Phi$ is completely positive if and only if it has a symmetric operator-sum decomposition, so that \[ \Phi(X)=\sum_{i=1}^k A_iXA_i^* \] for all $X$.

Each super-operator also has a \emph{Stinespring representation}, whereby there exists a complex Euclidean space $\cZ$ and operators $A,B:\cX\to\cY\ot\cZ$ such that \[ \Phi(X) = \ptr{\cZ}{AXB^*} \] for all $X$.
By analogy to the operator-sum representation, $\Phi$ is completely positive if and only if it has a symmetric Stinespring representation, so that \[ \Phi(X) = \ptr{\cZ}{AXA^*} \] for all $X$.
If, in addition, $\Phi$ is also trace-preserving then the operator $A$ must be an isometry.

\subsubsection{The operator-vector and Choi-Jamio\l kowski isomorphisms}

In this thesis we make use of an unconventional but useful isomorphism $\vectorize$ that associates with each operator $A:\cX\to\cY$ a unique vector $\col{A}$ in the complex Euclidean space $\cY\ot\cX$.
Letting $\set{e_1,\dots,e_{\dim(\cX)}}\subset\cX$ and $\set{f_1,\dots,f_{\dim(\cY)}}\subset\cY$ denote the standard bases of $\cX$ and $\cY$, this isomorphism defined by
\[ \col{f_je_i^*} = f_j\ot e_i. \]
Note that this correspondence is basis-dependent, and we have chosen the standard basis in our definition.

In the matrix representation, the column vector $\col{A}$ is obtained from the entries of the matrix $A$ by taking each row of $A$, transposing that row to form a column vector, and then stacking each of these column vectors so as to form one large column vector.
For example, the $\vectorize$ mapping acts as follows on $2\times 2$ matrices:
\[
  A =
  \left[
    \begin{array}{cc}
      a & b \\ c & d
    \end{array}
  \right]
  , \qquad
  \col{A} =
  \left[
    \begin{array}{c}
      a \\ b \\ c \\ d
    \end{array}
  \right].
\]
The $\vectorize$ isomorphism has many convenient properties, some of which we list here.
Each of these identities may be verified by straightforward calculation.

\begin{proposition}[Properties of the $\vectorize$ isomorphism]

The following hold:
\begin{enumerate}
\item
The $\vectorize$ mapping preserves the standard inner products of operators and vectors.
That is, for each $A,B:\cX\to\cY$ it holds that \[ \inner{A}{B} = \inner{\col{A}}{\col{B}}. \]
\item
For any operators $A$, $B$, and $X$ for which the composition $AXB$ is defined it holds that \[ \Pa{A\ot B^\trans}\col{X} = \col{AXB}. \]
\item
For each $A,B:\cX\to\cY$ it holds that
\begin{align*}
  \ptr{\cX}{\col{A}\row{B}} &= AB^*, \\
  \ptr{\cY}{\col{A}\row{B}} &= (B^*A)^\trans.
\end{align*}
\item
For each $x\in\cX$ and $y\in\cY$ it holds that \[ \col{yx^*} = y\ot\overline{x}. \]
\end{enumerate}

\end{proposition}

The \emph{Choi-Jamio\l kowski isomorphism} associates with each super-operator $\Phi:\lin{\cX}\to\lin{\cY}$ a unique operator $\jam{\Phi}\in\lin{\cY\ot\cX}$.
Letting $\set{e_1,\dots,e_{\dim(\cX)}}\subset\cX$ denote the standard basis of $\cX$, this isomorphism defined by
\[ \jam{\Phi} = \sum_{i,j=1}^{\dim(\cX)} \Phi(e_ie_j^*)\ot e_ie_j^*. \]
As with the $\vectorize$ isomorphism, the Choi-Jamio\l kowski isomorphism is basis-dependent and it is always defined with respect to the standard basis.

Given that $I_\cX=\sum_{i=1}^{\dim(\cX)} e_ie_i^*$, we obtain the following alternate characterization of the Choi-Jamio\l kowski isomorphism:
\[ \jam{\Phi} = \Pa{\Phi\ot\idsup{\cX}}(\col{I_\cX}\row{I_\cX}). \]
Like the $\vectorize$ isomorphism, the Choi-Jamio\l kowski isomorphism has many convenient properties, some of which we list here.

\begin{proposition}[Properties of the Choi-Jamio\l kowski isomorphism]

The following hold for all super-operators $\Phi:\lin{\cX}\to\lin{\cY}$:
\begin{enumerate}
\item
The Choi-Jamio\l kowski isomorphism preserves the standard inner product of super-operators and operators.
That is, for each $\Phi,\Psi:\lin{\cX}\to\lin{\cY}$ it holds that \[ \inner{\Phi}{\Psi} = \inner{\jam{\Phi}}{\jam{\Psi}}. \]
\item \label{item:prop:jam:act}
For each $X\in\lin{\cX}$ it holds that \[ \Phi(X) = \Ptr{\cX}{\Pa{I_\cY\ot X^\trans}\jam{\Phi}}. \]
\item
If $\Phi$ has operator-sum and Stinespring representations given by
\[ \Phi(X) = \sum_{i=1}^k A_iXB_i^* \qquad \textrm{and} \qquad \Phi(X) = \ptr{\cZ}{AXB^*} \]
for all $X$ then it holds that
\[ \jam{\Phi} = \sum_{i=1}^k \col{A_i}\row{B_i} \qquad \textrm{and} \qquad \jam{\Phi} = \ptr{\cZ}{\col{A}\row{B}}, \]
respectively.
\item
$\Phi$ is Hermitian-preserving if and only if $\jam{\Phi}$ is Hermitian.
\item
$\Phi$ is completely positive if and only if $\jam{\Phi}$ is positive semidefinite.
\item
$\Phi$ is trace-preserving if and only if $\ptr{\cY}{\jam{\Phi}}=I_\cX$.
\end{enumerate}

\end{proposition}

A generalization of item \ref{item:prop:jam:act} is proven in Proposition \ref{prop:identities} of Chapter \ref{ch:char}.
An additional identity involving the Choi-Jamio\l kowski isomorphism appears later in this section in Proposition \ref{prop:explicit-bound}.

\subsubsection{Operator decompositions}

There are two operator decompositions that are so fundamental to quantum information that, like the tensor product, they are often used implicitly.
Indeed, in this thesis we make little explicit mention of these decompositions, yet our discussion always assumes a working knowledge of their existence.

The \emph{Singular Value Theorem} states that every operator $A:\cX\to\cY$ has at least one \emph{singular value decomposition}, whereby there exist orthonormal sets
\( \set{x_1,\dots,x_r}\subset\cX \) and \( \set{y_1,\dots,y_r}\subset\cY \)
and positive real numbers $s_1,\dots,s_r\in\Real$ such that
\[ A = \sum_{i=1}^r s_i y_i x_i^*. \]
Here $r$ is the rank of the operator $A$.
The real numbers $s_1,\dots,s_r$ are called the \emph{singular values} of $A$.
Sometimes, the vectors $y_1,\dots,y_r$ are called the \emph{left singular vectors} of $A$, whereas $x_1,\dots,x_r$ are called the \emph{right singular vectors} of $A$.

The \emph{Spectral Theorem} implies that an operator $A\in\lin{\cX}$ is Hermitian if and only if there exists at least one \emph{spectral decomposition} whereby there is an orthonormal set $\set{x_1,\dots,x_r}\subset\cX$ and real numbers $\lambda_1,\dots,\lambda_r$ such that
\[ A = \sum_{i=1}^r \lambda_r x_i x_i^*. \]
Again, $r$ is the rank of $A$.
The real numbers $\lambda_1,\dots,\lambda_r$ are called the \emph{eigenvalues} of $A$ and the vectors $x_1,\dots,x_r$ are called the \emph{eigenvectors} of $A$.

It is easy to see that a Hermitian operator $A$ is positive semidefinite if and only if each of its eigenvalues is nonnegative.
It follows immediately from the Spectral Theorem that every Hermitian operator $A$ has a \emph{Jordan decomposition} whereby there exist positive semidefinite operators $P$ and $Q$ with the property that \[ A=P-Q \] and $PQ=0$---that is, $P$ and $Q$ act on orthogonal subspaces.
The \emph{absolute value} $\abs{A}$ of $A$ is a positive semidefinite operator defined via the Jordan decomposition by \[ \abs{A} = P+Q. \]

A positive semidefinite operator $P$ is called a \emph{projection} if each of its eigenvalues is either zero or one.

\subsubsection{Norms of vectors and operators}

For each real number $p\geq 1$ and each vector \( u=(\alpha_1,\dots,\alpha_n) \) in some $n$-dimensional complex Euclidean space, the \emph{vector $p$-norm} $\norm{u}_p$ of $u$ is defined by
\[ \norm{u}_p = \Pa{ \sum_{i=1}^n \abs{\alpha_i}^p }^{1/p} \]
with $\norm{u}_\infty$ given by
\[ \norm{u}_\infty = \lim_{p\to\infty} \norm{u}_p = \max_i \abs{\alpha_i}. \]
These norms satisfy $\norm{u}_p\leq\norm{u}_q$ whenever $p\geq q$, and it also holds that
\[ \norm{u}_1 \leq \sqrt{n}\norm{u}_2 \leq n\norm{u}_\infty. \]
The standard Euclidean norm $\norm{u}$ is an instance of the vector $p$-norm with $p=2$, so that $\norm{u}=\norm{u}_2$ for all $u$.

Each vector $p$-norm induces a norm on operators via the singular value decomposition.
In particular, the \emph{Schatten $p$-norm} $\norm{A}_p$ of an operator $A$ is defined as the vector $p$-norm of the singular values of $A$.
As such, the Schatten $p$-norms inherit many properties from the vector $p$-norms.
For example, $\norm{A}_p\leq\norm{A}_q$ whenever $p\geq q$ and
\[ \norm{A}_1 \leq \sqrt{n}\norm{A}_2 \leq n\norm{A}_\infty \]
for operators $A\in\lin{\cX}$ with $\dim(\cX)=n$.

In this thesis we are interested only in the Schatten $p$-norms for the values $p=1,2,\infty$.
The Schatten $p$-norm for $p=1$ is also called the \emph{trace norm} and is alternately denoted $\tnorm{A}$.

The Schatten $p$-norm for $p=2$ is also called the \emph{Frobenius norm} and is alternately denoted $\fnorm{A}$.
The Frobenius norm is merely the standard Euclidean norm for complex Euclidean spaces applied to the complex vector space of linear operators:
\[ \fnorm{A} = \sqrt{\inner{A}{A}} = \sqrt{\inner{\col{A}}{\col{A}}} = \norm{\col{A}}. \]

The Schatten $p$-norm for $p=\infty$ is also called the \emph{standard operator norm} and is alternately denoted without any subscript by $\norm{A}$.
The operator norm is \emph{induced} from the standard Euclidean norm for vectors:
\[ \norm{A} = \max_{\norm{u}=1} \norm{Au}. \]
In this thesis we prefer the notation $\tnorm{A}$, $\fnorm{A}$, and $\norm{A}$ to $\norm{A}_1$, $\norm{A}_2$, and $\norm{A}_\infty$.

We now prove a simple identity involving the trace norm of the Choi-Jamio\l kowski representation of a completely positive and trace-preserving super-operator.

\begin{proposition}
\label{prop:explicit-bound}

For any completely positive and trace-preserving super-operator $\Phi:\lin{\cX}\to\lin{\cY}$ it holds that $\tnorm{\jam{\Phi}}=\tnorm{\jam{\Phi^*}}=\dim(\cX)$.

\end{proposition}

\begin{proof}

  By definition, we have
  \begin{align*}
    \jam{\Phi} &= \Pa{\Phi\ot\idsup{\cX}}(\col{I_\cX}\row{I_\cX}),\\
    \jam{\Phi^*} &= \Pa{\Phi^*\ot\idsup{\cY}}(\col{I_\cY}\row{I_\cY}).
  \end{align*}
  As $\jam{\Phi}$ is positive semidefinite, so too must be $\jam{\Phi^*}$, from which we obtain
  \begin{align*}
    \tnorm{\jam{\Phi}}&=\ptr{}{\jam{\Phi}},\\
    \tnorm{\jam{\Phi^*}}&=\ptr{}{\jam{\Phi^*}}.
  \end{align*}
  As $\Phi$  is trace-preserving, it holds that
  \[ \ptr{}{\jam{\Phi}} = \row{I_\cX}\col{I_\cX} = \dim(\cX). \]
  Moreover, it is easy to verify that $\Phi^*(I_\cY)=I_\cX$, from which it follows that
  \[
    \ptr{}{\jam{\Phi^*}} =
    \ptr{}{ \Phi^*(I_\cY) } = \ptr{}{I_\cX} = \dim(\cX). \]
\end{proof}

\subsection{Convexity}

\subsubsection{Carath\'eodory's Theorem}

A subset $C$ of a real vector space $\Real^n$ is called \emph{convex} if for every $x,y\in C$ and every real number $\alpha\in[0,1]$ it holds that \[ \alpha x + (1-\alpha) y \in C. \]
For vectors $v_1,\dots,v_m\in\Real^n$ and nonnegative real numbers $\alpha_1,\dots,\alpha_m\geq 0$ with \[ \sum_{i=1}^m \alpha_i = 1 \]
the vector \[ \sum_{i=1}^m \alpha_i v_i \] is called a \emph{convex combination} of $v_1,\dots,v_m$.
For an arbitrary set $S\subset\Real^n$, the \emph{convex hull} of $S$ is the subset of $\Real^n$ consisting of all convex combinations of elements in $S$.
\emph{Carath\'eodory's Theorem} is a useful result that bounds the number of terms in the sum of a given convex combination.

\begin{fact}[Carath\'eodory's Theorem] \label{fact:Carateodory}

Let $S\subset\Real^n$ be an arbitrary set.
Every element $x$ of the convex hull of $S$ can be written as a convex combination of no more than $n+1$ elements of $S$.

\end{fact}

Proofs of Carath\'eodory's Theorem can be found, for example, in Rockafellar \cite{Rockafellar70} or Barvinok \cite{Barvinok02}.

\subsubsection{Separation Theorem}

The fundamental \emph{Separation Theorem} tells us that every pair of disjoint convex sets may be separated by a hyperplane.
There are many variants of the Separation Theorem, each of which differs only slightly based upon properties of the sets in question such as whether or not they are open, closed, bounded, or cones.
(A subset $K\subset\Real^n$ is called a \emph{cone} if for every $x\in K$ and every positive real number $\lambda>0$ it holds that $\lambda x \in K$.)
The following variant of the Separation Theorem serves all our needs in this thesis.

\begin{fact}[Separation Theorem]
\label{fact:separation}

Let $C,D\subset\Real^n$ be nonempty disjoint convex sets such that $D$ is open.
There exists a vector $h\in\Real^n$ and a real number $\alpha\in\Real$ with the property that
\begin{align*}
  \inner{h}{x} &\geq \alpha \textrm{ for all $x\in C$},\\
  \inner{h}{y} &< \alpha \textrm{ for all $y\in D$}.
\end{align*}
Moreover, if $C$ is a cone then we may take $\alpha=0$.

\end{fact}

\begin{proof}

It follows from Theorem 11.3 of Rockafellar \cite{Rockafellar70} that there exists a nonzero vector $h\in\Real^n$ and a real number $\alpha\in\Real$ with $\inner{h}{x} \geq \alpha \geq \inner{h}{y}$ for all $x\in C$ and $y\in D$.
Suppose toward a contradiction that there is a $y'\in D$ with $\inner{h}{y'}=\alpha$.
Choose any vector $b$ with $\inner{h}{b}>0$, so that $\inner{h}{y'+b}>\alpha$.
By rescaling $b$ and using the fact that $D$ is open, we may assume $y'+b\in D$, which contradicts the fact that $\inner{h}{y}\leq\alpha$ for all $y\in D$.

Next, suppose that $C$ is a cone and let $\alpha'$ denote the infimum of $\inner{h}{x}$ over all $x\in C$, so that \[ \inner{h}{x} \geq \alpha' \geq \alpha > \inner{h}{y} \]
for all $x\in C$ and $y\in D$.
Our proof that $\alpha'=0$ follows that of Theorem 11.7 of Rockafellar \cite{Rockafellar70}.
First, suppose toward a contradiction that $\alpha'<0$ and choose $x\in C$ with $\inner{h}{x}<0$.
Then $\inner{h}{\lambda x}$ can be made into an arbitrarily large negative number by an appropriately large choice of $\lambda$, contradicting the lower bound $\alpha$ on the infimum $\alpha'$.
Conversely, suppose toward a contradiction that $\alpha'>0$ and choose $\lambda$ small enough so that $\inner{h}{\lambda x}<\alpha'$, contradicting the definition of $\alpha'$.
\end{proof}

\subsubsection{Convexity and Hermitian operators}

As noted in Section \ref{sec:intro:linalg}, the set $\her{\cX}$ of all Hermitian operators is a real vector space of dimension $\dim(\cX)^2$.
As such, the formalism of convexity translates without complication to Hermitian operators.
For example, the set $\pos{\cX}$ of all positive semidefinite operators within $\her{\cX}$ is a closed convex cone.
Moreover, Carath\'eodory's Theorem and the Separation Theorem have natural analogues in the context of $\her{\cX}$.
Indeed, it is worth restating Fact \ref{fact:separation} (Separation Theorem) in terms of Hermitian operators.

\begin{fact}[Separation Theorem, Hermitian operator version]
\label{fact:herm-sep}

Let $\bC,\bD\subset\her{\cX}$ be nonempty disjoint convex sets such that $\bD$ is open.
There exists a Hermitian operator $H\in\her{\cX}$ and a real number $\alpha\in\Real$ with the property that
\begin{align*}
  \inner{H}{X} &\geq \alpha \textrm{ for all $X\in \bC$},\\
  \inner{H}{Y} &< \alpha \textrm{ for all $Y\in \bD$}.
\end{align*}
Moreover, if $\bC$ is a cone then we may take $\alpha=0$.

\end{fact}

\subsection{Quantum information}

\subsubsection{Quantum states}

Associated with each $d$-level physical system is a $d$-dimensional complex Euclidean space $\cX$.
The \emph{quantum state} of such a system at some fixed point in time is described uniquely by a positive semidefinite operator $\rho\in\pos{\cX}$ with
trace equal to one.
Such an operator might also be called a \emph{density operator}.

A quantum state $\rho$ called \emph{pure} if $\rho$ has rank one---that is $\rho=uu^*$ for some unit vector $u\in\cX$.
It follows immediately from the Spectral Theorem that every quantum state may be written as a convex combination (or probabilistic ensemble) of orthogonal pure states.

Two distinct physical systems with associated complex Euclidean spaces $\cX$ and $\cY$ may be viewed as a single larger system with associated complex Euclidean space $\cX\ot\cY$.
If $\rho\in\pos{\cX\ot\cY}$ is the state of this larger system then the state $\sigma\in\pos{\cX}$ of the subsystem associated only with $\cX$ is given by the partial trace \[\sigma=\ptr{\cY}{\rho}.\]
Conversely, if the systems with associated spaces $\cX$ and $\cY$ are in states $\sigma$ and $\sigma'$ respectively then it must hold that the state $\rho\in\pos{\cX\ot\cY}$ of the larger system satisfies $\sigma=\ptr{\cY}{\rho}$ and $\sigma'=\ptr{\cX}{\rho}$.
Any state $\rho$ meeting these conditions is said to be \emph{consistent} with $\sigma$ and $\sigma'$.
If there is no correlation or entanglement between these two subsystems then $\rho$ is given by the tensor product \[\rho=\sigma\ot\sigma'.\]
Such a state is called a \emph{product state}.
For each $\sigma$ there are many states $\rho$ consistent with $\sigma$ other than product states.
For example, each state $\sigma\in\pos{\cX}$ may be \emph{purified}.
In other words, there always exists a complex Euclidean space $\cZ$ with $\dim{\cZ}=\rank(\sigma)$ and a pure state $u\in\cX\ot\cZ$ such that \[\sigma = \ptr{\cZ}{uu^*}.\]
The state $uu^*$ is called a \emph{purification} of $\sigma$.

\subsubsection{Quantum operations}

A \emph{quantum operation} is a physically realizable discrete-time mapping (at least in an ideal sense) that takes as input a quantum state $\rho\in\pos{\cX}$ of some system $\cX$ and produces as output a quantum state $\sigma\in\pos{\cY}$ of some system $\cY$.
Every quantum operation is described uniquely by a completely positive and trace-preserving super-operator $\Psi:\lin{\cX}\to\lin{\cY}$, so that, for each input state $\rho$, the corresponding output state $\sigma$ is given by \[\sigma = \Psi(\rho).\]
As noted in Section \ref{sec:intro:linalg}, each quantum operation $\Psi$ has a Stinespring representation whereby there exists a complex Euclidean space $\cZ$ with $\dim(\cZ)=\rank(\jam{\Psi})$ and an isometry $A:\cX\to\cY\ot\cZ$ such that \[\Phi:X\mapsto\ptr{\cZ}{AXA^*}. \]
Naturally, two quantum operations $\Phi_b:\lin{\cX_b}\to\lin{\cY_b}$ for $b=1,2$ acting upon distinct physical systems may be viewed as one large quantum operation acting upon one large system via the tensor product:
\[ \Phi_1\ot\Phi_2:\lin{\cX_1\ot\cX_2}\to\lin{\cY_1\ot\cY_2}. \]
Using our shorthand notation, this tensor product can alternately be written
\[ \kprod{\Phi}{1}{2}:\lin{\kprod{\cX}{1}{2}}\to\lin{\kprod{\cY}{1}{2}}.\]

\subsubsection{Quantum measurements}

Many applications of quantum information necessitate the ability to extract classical information from a quantum state of some physical system $\cX$.
This extraction is accomplished via \emph{quantum measurement}.
Each measurement is uniquely specified by a finite set $\Gamma$ of \emph{measurement outcomes} and a finite set $\set{P_a}_{a\in\Gamma}\subset\pos{\cX}$ of positive semidefinite \emph{measurement operators} obeying the condition
\[ \sum_{a\in\Gamma} P_a = I_\cX. \]
(These measurement operators are sometimes called \emph{POVM elements} for historical reasons---we avoid that terminology in this thesis.)
For any state $\rho$ of the system $\cX$, the probability with which the measurement $\set{P_a}$ yields a particular outcome $a\in\Gamma$ is given by the inner product
\[ \Pr[\textrm{$\set{P_a}$ yields outcome $a$ on $\rho$}] = \inner{P_a}{\rho}. \]

It is convenient to adopt the convention that the physical system associated with $\cX$ is destroyed when the measurement is applied.
There is a simple and common modification of the formalism of quantum measurements that admits non-destructive measurements that do not destroy the system upon which they act, but we do not require such a formalism in this thesis.

A measurement is called \emph{projective} if each of its measurement operators is a projection.
An arbitrary measurement $\set{P_a}_{a\in\Gamma}\subset\pos{\cX}$ may be ``simulated'' by a projective measurement in the sense that there always exists a complex Euclidean space $\cZ$, an isometry $A:\cX\to\cX\ot\cZ$ and a projective measurement $\set{\Pi_a}_{a\in\Gamma}\subset\pos{\cX\ot\cZ}$ with the property that \[ \inner{P_a}{\rho} = \inner{\Pi_a}{A\rho A^*} = \inner{A^*\Pi_a A}{\rho} \] for every quantum state $\rho\in\pos{\cX}$.

\part{Quantum Strategies} \label{part:strategies}

\chapter{Introduction to Quantum Strategies}
\label{ch:strategies}

In this chapter we develop two distinct formalisms for quantum strategies.
We begin with a naive operational formalism in Section \ref{sec:naive}, and we provide a new formalism in Section \ref{sec:strategy}.

\section{Operational formalism} \label{sec:naive}

In this section we develop an intuitive, operational formalism for quantum strategies and discuss some properties and problems associated with that formalism.

\subsection{Formal definitions}



At a high level, a \emph{strategy} is a complete description of one party's actions in a multiple-round interaction involving the exchange of quantum information with one or more other parties.
For convenience, let us call this party \emph{Alice}.
As we are only concerned for the moment with Alice's actions during the interaction, it is convenient to bundle the remaining parties into one party, whom we call \emph{Bob}.

From Alice's point of view, every finite interaction decomposes naturally into a finite number $r$ of \emph{rounds}.
In a typical round a message comes in, the message is processed, and a reply is sent out.
Naturally, this reply might depend upon messages exchanged during previous rounds of the interaction.
To account for such a dependence, we allow for a memory workspace to be maintained between rounds.

In order to facilitate discussion of the distinct strategies available to Alice for a given interaction, it is convenient to adopt the convention that the interaction specifies the number $r$ of rounds and the size of each message in the interaction.
In particular, individual strategies for Alice are \emph{not} free to specify messages whose number or size deviates from those dictated by the interaction.
By contrast, an interaction does \emph{not} dictate the size of the memory workspace to be used by each party.
In particular, different strategies for Alice may call for different amounts of memory between rounds, and there is \emph{no limit} on the amount of memory workspace a strategy may use.
(However, we shall see in Chapter \ref{ch:properties} that every possible strategy can be implemented with a reasonable fixed amount of memory workspace.)

The complex Euclidean spaces corresponding to the incoming and outgoing messages in an arbitrary round $i$ shall be denoted $\cX_i$ and $\cY_i$, respectively.
The space corresponding to the memory workspace to be stored for the next round shall be denoted $\cZ_i$.
As the size of each message is fixed by the interaction, the dimension of the spaces $\cX_1,\dots,\cX_r$ and $\cY_1,\dots,\cY_r$ are also fixed by the interaction.
Conversely, the dimension of the spaces $\cZ_1,\dots,\cZ_r$ may be arbitrarily large and these dimensions may vary among different strategies for the same interaction.

Thus, in a typical round $i$ of the quantum interaction, Alice's actions are faithfully represented by a quantum operation
\[ \Phi_i : \lin{\cX_i\ot\cZ_{i-1}}\to\lin{\cY_i\ot\cZ_i}. \]
The first round of the interaction is a special case: there is no need for an incoming memory space for this round, so the quantum operation $\Phi_1$ has the form
\[ \Phi_1 : \lin{\cX_1}\to\lin{\cY_1\ot\cZ_1}. \]
The final round of the interaction is also a special case: there is no immediate need for an outgoing memory space for this round.
However, the presence of this final memory space better facilitates the forthcoming discussion of strategies involving measurements.
Thus, the quantum operation $\Phi_r$ representing Alice's actions in the final round of the interaction has the same form as those from previous rounds:
\[ \Phi_r : \lin{\cX_r\ot\cZ_{r-1}}\to\lin{\cY_r\ot\cZ_r}. \]


In order to extract classical information from the interaction, it suffices to permit Alice to perform a single quantum measurement on her final memory workspace.
Sufficiency of a single measurement at the end of the interaction follows from foundational results on mixed state quantum computations \cite{AharonovK+98}, which tell us that any quantum operation calling for one or more intermediate measurements can be efficiently simulated by an operation with a single measurement at the end.

Discussion thus far in this section is aptly summarized by the following formal definition.
\def\defnaivestrategy{Operational definition of a strategy}
\begin{definition}[\defnaivestrategy]
\label{def:naive-strategy}

  Let $\cX_1,\ldots,\cX_r$ and $\cY_1,\ldots,\cY_r$ be complex Euclidean spaces.
  An \emph{operational $r$-round non-measuring strategy} for an interaction with \emph{input spaces} $\cX_1,\ldots,\cX_r$ and \emph{output spaces} $\cY_1,\ldots,\cY_r$ consists of:
  \begin{enumerate}

  \item[1.]
    complex Euclidean spaces $\cZ_1,\ldots,\cZ_r$, called \emph{memory spaces}, and

  \item[2.]
    an $r$-tuple of quantum operations $(\Phi_1,\ldots,\Phi_r)$ of the form
    \begin{align*}
      \Phi_1 &: \lin{\cX_1}\to\lin{\cY_1\ot\cZ_1}\\
      \Phi_i &: \lin{\cX_i\ot\cZ_{i-1}}\to\lin{\cY_i\ot\cZ_i} \quad (2 \leq i \leq r).
    \end{align*}

  \end{enumerate}
  An \emph{operational $r$-round measuring strategy} with outcomes indexed by $a$ consists of items 1 and 2 above, as well as:
  \begin{enumerate}

  \item[3.]
    a measurement $\set{P_a}$ on the last memory space $\cZ_r$.

  \end{enumerate}
  Figure~\ref{fig:strategy} illustrates an $r$-round non-measuring strategy.
\end{definition}

\begin{figure}[t]
  \begin{center}
\includegraphics{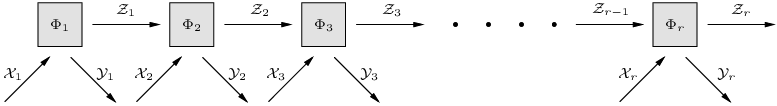}
  \end{center}
  \caption{An $r$-round strategy.}
  \label{fig:strategy}
\end{figure}

Having formalized the notion of a strategy for Alice, it is straightforward to formalize the notion of a compatible strategy for Bob.
Essentially, Bob must provide to Alice incoming messages with corresponding spaces $\cX_1,\dots,\cX_r$ and accept her outgoing replies with corresponding spaces $\cY_1,\dots,\cY_r$.

\def\defnaivecostrategy{Operational definition of a co-strategy}
\begin{definition}[\defnaivecostrategy]
\label{def:naive-co-strategy}

  Let $\cX_1,\dots,\cX_r$ and $\cY_1,\dots,\cY_r$ be complex Euclidean spaces.
  An \emph{operational $r$-round non-measuring co-strategy} for an interaction with \emph{input spaces} $\cX_1,\ldots,\cX_r$ and \emph{output spaces} $\cY_1,\ldots,\cY_r$ consists of:
  \begin{enumerate}

  \item[1.]
    complex Euclidean {\it memory} spaces $\cW_0,\ldots,\cW_r$,

  \item[2.]
    a quantum state $\rho_0 \in \pos{\cX_1\otimes\cW_0}$, and

  \item[3.]
    a $r$-tuple of quantum operations $(\Psi_1,\ldots,\Psi_r)$ of the form
    \begin{align*}
      \Psi_i & : \lin{\cY_i\ot\cW_{i-1}}\to\lin{\cX_{i+1}\ot\cW_i} \quad (1\leq i\leq r-1)\\
      \Psi_r & : \lin{\cY_r\ot\cW_{r-1}}\to\lin{\cW_r}.
    \end{align*}

  \end{enumerate}
  An \emph{operational $r$-round measuring co-strategy} with outcomes indexed by $b$ consists of items 1, 2 and 3 above, as well as:
  \begin{enumerate}

  \item[4.]
    a measurement $\set{Q_b}$ on the last memory space $\cW_r$.

  \end{enumerate}
  Figure~\ref{fig:interaction} depicts the interaction between a $r$-round strategy and co-strategy.
\end{definition}

\begin{figure}[t]
  \begin{center}
\includegraphics{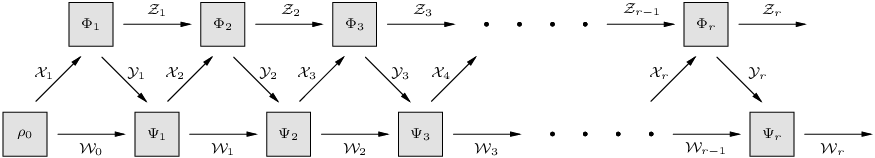}
  \end{center}
  \caption{An interaction between a $r$-round strategy and co-strategy.}
  \label{fig:interaction}
\end{figure}

\subsection{Immediate observations}
\label{subsec:naive-observations}

\subsubsection{Generality of the definition}

Our definition of a strategy allows for trivial input or output spaces with dimension one---such spaces correspond to empty messages.
Hence, simple actions such as the preparation of a quantum state or performing a measurement without producing a quantum output can be viewed as special cases of strategies.

While we require that Alice receive the first message and send the last message, any interaction that deviates from this format can easily be recast to fit this mold at the possible expense of one additional round of messages, some of which might be empty.

Indeed, our definition is flexible enough that any $r$-round co-strategy may equivalently be viewed as a $(r+1)$-round strategy with input spaces $\mathbb{C},\cY_1,\dots,\cY_r$ and output spaces $\cX_1,\dots,\cX_r,\mathbb{C}$.
It is only for later convenience that we have chosen to define strategies and co-strategies as distinct objects.

\subsubsection{Restriction to finite interactions}

In this thesis we restrict our attention to interactions in which the size and number of messages exchanged is fixed.
As we shall see in Section \ref{sec:strategy}, this restriction permits us the luxury of representing every conceivable strategy for a given interaction by a single positive semidefinite operator of fixed finite dimension.
Many of the results proven in this thesis rely crucially upon the finiteness of these representations of strategies.

At first, this restriction might seem overly constraining.
For example, one might reasonably wish to study strategies for, say, coin-flipping protocols wherein two parties continue exchanging messages until they reach an agreement.
With each round of messages, this agreement might be reached only probabilistically according to some quantum measurement.
Such an interaction is infinite in the sense that it admits strategies with the property that, for every number $r$ of rounds, there is a nonzero probability that it will not terminate in $r$ rounds or fewer.
(Indeed, this scenario also admits silly strategies in which the interaction \emph{never} terminates.)
To forbid interactions such as this might seem careless, as our study aspires to absolute generality for quantum strategies.

While it would certainly be desirable to include in our formalism potentially infinite interactions, 
the generality that we sacrifice in the name of finiteness is less of a liability than it might at first seem.
Many ``infinite'' interactions of any interest---such as the previous coin-flipping example---will terminate after a finite number of messages with probability 1.
Interactions such as this can be approximated to arbitrary precision with only a finite number of messages simply by truncating the interaction appropriately.


Admittedly, the precise extent to which generality is sacrificed by our restriction to finite interactions remains to be seen.
But it cannot be denied that the set of finite interactions encompasses a very wide swath of interesting quantum interactions.
In the real world, all interactions are finite.

\subsubsection{Restriction to isometric quantum operations}

When convenient, we shall assume without loss of generality that each of Alice's quantum operations $\Phi_i$ are actually linear isometries, meaning that \[ \Phi_i : X \mapsto A_iXA_i^* \] for some linear isometry $A_i$. 
Similarly, we may assume that Bob's initial state $\rho_0$ is actually a pure state, and that each of Bob's quantum operations $\Psi_i$ have the form \[ \Psi_i : X \mapsto B_iXB_i^* \] for some linear isometry $B_i$.
Moreover, if Alice's or Bob's strategy is a measuring strategy then we may also assume that the measurements $\set{P_a}$ for Alice and $\set{Q_b}$ for Bob are actually projective measurements, meaning that each $P_a$ and $Q_b$ is actually a projection operator.

These assumptions follow from the Stinespring representation for quantum operations mentioned in Section \ref{sec:intro:linalg}.
Specifically, an arbitrary quantum operation $\Phi : \lin{\cX}\to\lin{\cY}$ may be written $\Phi : X \mapsto \ptr{\cH}{UXU^*}$ for some auxiliary space $\cH$ and some isometry $U : \cX\to\cY\ot\cH$.
As the memory spaces for strategies may be arbitrary, the auxiliary space $\cH$ for a given round $i$ may be ``absorbed'' into the private memory space $\cZ_i$ of Alice's quantum operation $\Phi_i$, thus removing the need for the partial trace over $\cH$ and leaving us with $\Phi_i : X \mapsto A_iXA_i^*$ as desired.

The justification for projective measurements is similar:
Naimark's Theorem implies that any measurement on the final memory space $\cZ_r$ may be simulated by a projective measurement on $\cZ_r\ot\cH$ for some auxiliary space $\cH$.
As before, $\cH$ may be ``absorbed'' into the definition of $\cZ_r$, leaving us with only a projective measurement on this space.

\subsubsection{Undesirable properties of the operational formalism}

Definition \ref{def:naive-strategy} (\defnaivestrategy) is very natural
in the sense that it is clear that any conceivable actions taken by Alice during an interaction can be represented by a strategy of that form.
Unfortunately, this definition has several undesirable mathematical properties that make this formalization cumbersome to use.

Picking an example arbitrarily, it is easy to see that the set of all strategies for a given interaction lacks a convenient distributive property for probabilistic combinations of strategies.
In particular, suppose Alice plays according to $\Phi=(\Phi_1,\dots,\Phi_r)$ with probability $p$, otherwise she plays according to $\Phi'=(\Phi_1',\dots,\Phi_r')$.
What is the $r$-tuple of quantum operations that describes this probabilistic combination of strategies?
In an ideal world, this $r$-tuple would be given by the simple convex combination
\[ p\Phi + (1-p)\Phi' \]
where the $i$th component of the resulting $r$-tuple is given in the usual way by
\[ p\Phi_i + (1-p)\Phi_i'. \]
Alas, this identity does not always hold.
Indeed, unless $(\Phi_1,\dots,\Phi_r)$ and $(\Phi_1',\dots,\Phi_r')$ happen to agree on the choice of memory spaces $\cZ_1,\dots,\cZ_r$, the above convex combination is not even well-defined!
Of course, this issue is easily overcome by ``padding'' the smaller memory spaces so that their dimensions agree with those of the larger spaces.
But this frivolous observation ignores the larger problem.

Indeed, it is easy to exhibit a pair of strategies for which the above identity does not hold, even when the two strategies agree on the dimension of their memory spaces.
The simple example we present here is a toy classical strategy expressed in the quantum formalism.

\begin{example}[Lack of distributive property]
\label{ex:naive-distributive}

  Consider a two-round interaction wherein Alice's incoming messages are empty.
  Letting $\rho_\alpha,\rho_\beta$ denote two distinct quantum states, the 2-tuples $(\rho_\alpha,\rho_\alpha)$ and $(\rho_\beta,\rho_\beta)$ each specify a silly strategy for Alice that uses no memory space and returns the same state in both rounds.
  Suppose Alice plays one of $(\rho_\alpha,\rho_\alpha)$ and $(\rho_\beta,\rho_\beta)$ uniformly at random.
  At the end of this interaction, Bob has received from Alice a combined state
  \[ \frac{1}{2} \rho_\alpha\ot\rho_\alpha + \frac{1}{2} \rho_\beta\ot\rho_\beta. \]
  But the 2-tuple resulting from the naive computation
  \[
    \frac{1}{2}(\rho_\alpha,\rho_\alpha) + \frac{1}{2}(\rho_\beta,\rho_\beta) =
    \Pa{\frac{\rho_\alpha+\rho_\beta}{2},\frac{\rho_\alpha+\rho_\beta}{2}}
  \]
  results in strategy for Alice in which Bob is instead left with the state
  \[
    \frac{1}{4} \rho_\alpha\ot\rho_\alpha + \frac{1}{4} \rho_\beta\ot\rho_\beta +
    \frac{1}{4} \rho_\alpha\ot\rho_\beta + \frac{1}{4} \rho_\beta\ot\rho_\alpha
  \]
  at the end of the interaction.
  As these final states for Bob are not equal, it follows that the naive convex combination $1/2(\rho_\alpha,\rho_\alpha) + 1/2(\rho_\beta,\rho_\beta)$ does not denote the desired probabilistic combination of the strategies described by $(\rho_\alpha,\rho_\alpha)$ and $(\rho_\beta,\rho_\beta)$.

\end{example}

This lack of a convenient distributive property suggests a cumbersome nonlinear dependence of probabilistic combinations $p\Phi + (1-p)\Phi'$ of strategies upon their constituent strategies $\Phi,\Phi'$.
Other examples of undesirable properties include:
\begin{description}

\item[Strategies are not unique.]
  There exist pairs of strategies
  that differ in each component, yet they specify the same actions for an interaction in the sense that no interacting co-strategy could possibly distinguish the two.
  
  This lack of uniqueness suggests that our operational formalism is carrying around some unnecessary information, which typically leads to unnecessary complication.

\item[Outcomes depend nonlinearly on strategies.]
  For a multi-round interaction between Alice and Bob, the
  probability of obtaining a given measurement outcome
  at the end of the interaction depends multilinearly upon \emph{each component} $\Phi_i$ of Alice's strategy and $\Psi_i$ of Bob's strategy.
  But the dependence of this probability as a function of the two $r$-tuples $(\Phi_1,\dots,\Phi_r)$ and $(\Psi_1,\dots,\Psi_r)$ is highly nonlinear.
  
  As a consequence, existing algorithms for standard linear and semidefinite optimization problems cannot be employed to efficiently compute some important properties of strategies, such as the maximum probability with which a given strategy can be forced to yield a given measurement outcome.
  
\end{description}

\section{New formalism}
\label{sec:strategy}

In the previous section we saw that the operational representation for quantum strategies has several drawbacks.
In this section, we develop a new representation for quantum strategies that rectifies these problems.

\subsection{Formal definitions}

We begin with several new definitions for non-measuring and measuring quantum strategies and co-strategies.

\def\defstrategy{New formalism for non-measuring strategies}
\begin{definition}[\defstrategy]
\label{def:strategy}

Let $(\Phi_1,\dots,\Phi_r)$ be an operational representation of a non-measuring strategy for input spaces $\cX_1,\dots,\cX_r$ and output spaces $\cY_1,\dots,\cY_r$.
Our new representation for quantum strategies associates with every such operational strategy a single positive semidefinite operator
\[ Q \in \pos{\kprod{\cY}{1}{r}\ot\kprod{\cX}{1}{r}} \]
via the following construction.
Let \[ \Xi : \lin{\kprod{\cX}{1}{r}} \to \lin{\kprod{\cY}{1}{r}} \] denote the super-operator composed of $\Phi_1,\dots,\Phi_r$ as suggested by Figure \ref{fig:choijam}.
\begin{figure}[t]
  \begin{center}
\includegraphics{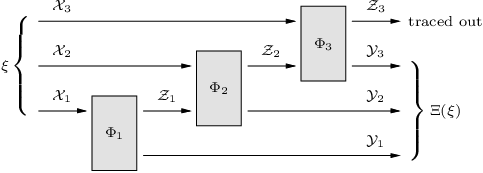}
  \end{center}
  \caption{The super-operator $\Xi$ associated with a three-round strategy.}
  \label{fig:choijam}
\end{figure}
Specifically, $\Xi$ is the $r$-fold composition of $\Phi_1,\dots,\Phi_r$ on the memory spaces $\cZ_1,\dots,\cZ_{r-1}$ followed by a partial trace on $\cZ_r$.
With some abuse of notation, this composition may be expressed succinctly as
\[ \Xi = \trace_{\cZ_r} \circ\: \Phi_r \circ \cdots \circ \Phi_1. \]
(Here a tensor product with the identity super-operator $\identity$ on the appropriate spaces is implicitly inserted where necessary in order for this composition to make sense.)

The \emph{$r$-round non-measuring strategy} is given by $Q\defeq\jam{\Xi}$.
An operator $Q$ is a valid representation of an $r$-round non-measuring strategy if and only if it has this form.
\end{definition}

\def\defmstrategy{New formalism for measuring strategies}
\begin{definition}[\defmstrategy]
\label{def:m-strategy}

With each operational representation $(\Phi_1,\dots,\Phi_r, \set{P_a})$ of an $r$-round measuring strategy  for input spaces $\cX_1,\dots,\cX_r$ and output spaces $\cY_1,\dots,\cY_r$ we associate a finite set 
\[ \set{Q_a} \subset \pos{\kprod{\cY}{1}{r}\ot\kprod{\cX}{1}{r}} \]
of positive semidefinite operators
via the following construction.
For each measurement outcome $a$ let
\[ \Xi_a = \Gamma_a \circ\: \Phi_r \circ \cdots \circ \Phi_1 \]
where the super-operator $\Gamma_a$ is given by
\[ \Gamma_a : X \mapsto \Ptr{\cZ_r}{\Pa{P_a\ot I_{\kprod{\cY}{1}{r}}}X}. \]
(Compare: for non-measuring strategies we defined $\Xi$ via the partial trace $\trace_{\cZ_r}$.
For measuring strategies we define $\Xi_a$ via $\Gamma_a$ instead of the partial trace.)

The \emph{$r$-round measuring strategy} $\set{Q_a}$ is given by $Q_a\defeq\jam{\Xi_a}$ for each $a$.
A set $\set{Q_a}$ of operators is a valid representation of an $r$-round measuring strategy if and only if it has this form.
\end{definition}

\def\defcostrategy{New formalism for co-strategies}
\begin{definition}[\defcostrategy]
\label{def:co-strategy}

Non-measuring co-strategies are defined in the same manner as non-measuring strategies, except the operator $Q$ is given by $Q\defeq\jam{\Xi^*}$ instead of $\jam{\Xi}$.
%
Similarly, measuring co-strategies are defined in the same manner as measuring strategies, except the operators $\set{Q_a}$ are given by $Q_a\defeq\jam{\Xi_a^*}$ instead of $\jam{\Xi_a}$ for each $a$.

Let us clarify this new representation for co-strategies.
Let $(\rho_0,\Psi_1,\dots,\Psi_r)$ denote an operational representation of an $r$-round non-measuring co-strategy for input spaces $\cX_1,\dots,\cX_r$ and output spaces $\cY_1,\dots,\cY_r$, so that
\begin{align*}
\rho_0 & \in \pos{\cX_1\ot\cW_0}\\
\Psi_i & : \lin{\cY_i\ot\cW_{i-1}}\to\lin{\cX_{i+1}\ot\cW_i} \quad (1\leq i\leq r-1)\\
\Psi_r & : \lin{\cY_r\ot\cW_{r-1}}\to\lin{\cW_r}.
\end{align*}
The induced super-operator
\( \Xi : \lin{\kprod{\cY}{1}{r}} \to \lin{\kprod{\cX}{1}{r}} \)
is depicted in Figure \ref{fig:co-trace-out} for the case $r=2$.
\begin{figure}[t]
  \begin{center}
\includegraphics{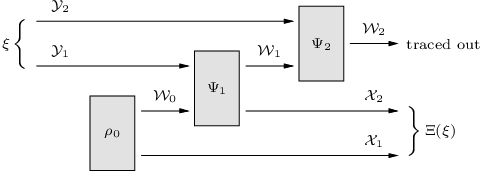}
  \end{center}
  \caption{The super-operator $\Xi$ associated with a two-round co-strategy.}
  \label{fig:co-trace-out}
\end{figure}
The adjoint super-operator has the form
\( \Xi^* : \lin{\kprod{\cX}{1}{r}} \to \lin{\kprod{\cY}{1}{r}}. \)
In particular, we have
\begin{align*}
\jam{\Xi} &\in \pos{\kprod{\cX}{1}{r}\ot\kprod{\cY}{1}{r}}, \\
\jam{\Xi^*} &\in \pos{\kprod{\cY}{1}{r}\ot\kprod{\cX}{1}{r}}.
\end{align*}
By taking $Q=\jam{\Xi^*}$, we ensure that strategies and co-strategies for the same input and output spaces lie within the same space $\her{\kprod{\cY}{1}{r}\ot\kprod{\cX}{1}{r}}$ of Hermitian operators---with $\kprod{\cY}{1}{r}$ occurring \emph{before} $\kprod{\cX}{1}{r}$ in the tensor product.

Incidentally, it is instructive to note that $\jam{\Xi}$ is a valid representation of an $(r+1)$-round \emph{strategy} for input spaces $\mathbb{C},\cY_1,\dots,\cY_r$ and output spaces $\cX_1,\dots,\cX_r,\mathbb{C}$ according to Definition \ref{def:strategy}.

Conversely, if $\Lambda:\lin{\kprod{\cX}{1}{r}}\to\lin{\kprod{\cY}{1}{r}}$ is a super-operator for which $\jam{\Lambda}$ is an $r$-round strategy for input spaces $\cX_1,\dots,\cX_r$ and output spaces $\cY_1,\dots,\cY_r$ then $\jam{\Lambda^*}$ denotes a valid $(r+1)$-round co-strategy for input spaces $\mathbb{C},\cY_1,\dots,\cY_r$ and output spaces $\cX_1,\dots,\cX_r,\mathbb{C}$.
\end{definition}

\subsection{Immediate observations}

Some basic properties of this new representation for quantum strategies may be pointed out immediately.
However, a more comprehensive discussion of basic properties of strategies must wait until Section \ref{sec:other-properties}, at which point the results of Chapter \ref{ch:properties} are employed to establish with ease some additional basic properties.

\subsubsection{Zero-round strategies}

It is convenient to adopt the convention that an $r$-round non-measuring strategy or co-strategy $Q$ denotes the scalar 1 when $r=0$.
Similarly, if $\set{Q_a}$ is a zero-round measuring strategy or co-strategy then $\set{Q_a}$ denotes a finite set of nonnegative real numbers that sum to one.

\subsubsection{Representation in terms of isometries}

Let $(\Phi_1,\dots,\Phi_r)$ be an operational description of an $r$-round strategy.
As discussed in Section \ref{subsec:naive-observations}, we may assume without loss of generality that each of the quantum operations $\Phi_i$ is actually an isometric operation, so that $\Phi_i : X \mapsto A_i X A_i^*$ for some isometry $A_i$.
In this case, the super-operator $\Xi$ in Definition \ref{def:strategy} (\defstrategy) has the form
\[ \Xi : X \mapsto \ptr{\cZ_r}{AXA^*} \]
where $A : \kprod{\cX}{1}{r} \to \kprod{\cY}{1}{r}\ot\cZ_r$ is an isometry defined as the composition of $A_1,\dots,A_r$ on the memory spaces $\cZ_1,\dots,\cZ_{r-1}$.
With some abuse of notation, this composition may be expressed succinctly as
\[ A = A_r A_{r-1} \cdots A_2 A_1 \]
where a tensor product with the identity operator $I$ on the appropriate spaces is inserted where necessary in order for this product to make sense.

Under the new formalism for strategies, the operator $Q$ representing $(\Phi_1,\dots,\Phi_r)$ has the special form
\[ Q = \jam{\Xi} = \ptr{\cZ_r}{\col{A}\row{A}}. \]
For a measuring strategy with measurement $\set{P_a}$, each element of $\set{Q_a}$ may be written
\[ Q_a = \Ptr{\cZ_r}{\Pa{P_a\ot I_{\kprod{\cY}{1}{r}\ot\kprod{\cX}{1}{r}}}\col{A}\row{A}}. \]

If $Q$ is a co-strategy, rather than a strategy, then by Definition \ref{def:co-strategy} (\defcostrategy) we have
\[ Q = \jam{\Xi^*} = \ptr{\cZ_r}{\col{A^*}\row{A^*}} \]
for non-measuring co-strategies and
\[ Q_a = \Ptr{\cZ_r}{\Pa{P_a\ot I_{\kprod{\cY}{1}{r}\ot\kprod{\cX}{1}{r}}}\col{A^*}\row{A^*}} \]
for measuring co-strategies.

\subsubsection{Measuring strategy elements sum to a non-measuring strategy}

We conclude this chapter with an exercise that should help to wet our feet before diving into Chapter \ref{ch:properties}.

\begin{proposition}
\label{prop:measure-sum}

A set $\set{Q_a}$
of positive semidefinite operators
is a measuring strategy or co-strategy if and only if $\sum_a Q_a$ is a non-measuring strategy or co-strategy, respectively.

\end{proposition}

\begin{proof}

The proof for co-strategies is completely symmetric to the proof for strategies, so we only prove the proposition for strategies.

The ``only if'' portion of the proof is straightforward.
As $\set{Q_a}$ is a measuring strategy, there is an $r$-tuple $(A_1,\dots,A_r)$ of isometries and a measurement $\set{P_a}$ with the property that each $Q_a$ is given by
\[ Q_a = \Ptr{\cZ_r}{\Pa{P_a\ot I_{\kprod{\cY}{1}{r}\ot\kprod{\cX}{1}{r}}}\col{A}\row{A}} \]
where the isometry $A$ is given by the $r$-fold composition $A=A_r\cdots A_1$ is as noted earlier in this section.
Then by linearity we have
\begin{align*}
\sum_a Q_a &= \Ptr{\cZ_r}{\Pa{\Pa{\sum_aP_a}\ot I_{\kprod{\cY}{1}{r}\ot\kprod{\cX}{1}{r}}}\col{A}\row{A}} \\
&= \Ptr{\cZ_r}{\col{A}\row{A}},
\end{align*}
which is an $r$-round non-measuring strategy as required.

We now proceed to the ``if'' portion of the proof.
As $\sum_a Q_a$ is a non-measuring strategy, there is an $r$-tuple $(A_1,\dots,A_r)$ of isometries with the property that
\[ \sum_a Q_a = \Ptr{\cZ_r}{\col{A}\row{A}} \]
where again the isometry
\[ A : \kprod{\cX}{1}{r}\to\kprod{\cY}{1}{r}\ot\cZ_r \]
is given by $A=A_r\cdots A_1$.
Let \[S:\cZ_r\to\kprod{\cY}{1}{r}\ot\kprod{\cX}{1}{r}\] be the operator
obtained from $A$ by swapping the spaces $\kprod{\cX}{1}{r}$ and $\cZ_r$.
(In other words, $S$ is the image of $A$ under the mapping $(y\ot z)x^*\mapsto(y\ot x)z^*$ on standard basis states.)
We may thus write
\[ \sum_a Q_a = \Ptr{\cZ_r}{\col{S}\row{S}} = SS^*. \]
Let $T$ denote the Moore-Penrose pseudo-inverse of $S$, so that
\[ ST = T^*S^* = \Pi_S \]
where $\Pi_S$ denotes the projection onto the image of $S$.
For each outcome $a$ we let $P_a=TQ_aT^*$.
As $Q_a$ is positive semidefinite, it is clear that $P_a$ is also positive semidefinite.
Observe that $Q_a=SP_aS^*$, which follows from
\[SP_aS^* = STQ_aT^*S^* = \Pi_S Q_a \Pi_S = Q_a \]
where the final equality follows from the fact that the image of $Q_a$ is contained in the image of $S$.
Then
\begin{align*}
  Q_a &= SP_aS^* \\
  &= \Ptr{\cZ_r}{\col{SP_a}\row{S}} \\
  &= \Ptr{\cZ_r}{\Pa{P_a\ot I_{\kprod{\cY}{1}{r}\ot\kprod{\cX}{1}{r}}}\col{S}\row{S}} \\
  &= \Ptr{\cZ_r}{\Pa{P_a\ot I_{\kprod{\cY}{1}{r}\ot\kprod{\cX}{1}{r}}}\col{A}\row{A}}.
\end{align*}
As
\( \sum_a Q_a = \Ptr{\cZ_r}{\col{A}\row{A}}, \)
it follows that
\[ \sum_a P_a = I_{\cZ_r}. \]
As each $P_a$ is positive semidefinite, the set $\set{P_a}$ is a quantum measurement and hence $\set{Q_a}$ is a measuring strategy.
\end{proof}

\chapter{Three Important Properties}
\label{ch:properties}

In this chapter we establish three fundamental and powerful properties of our new representation for quantum strategies.
For ease of reference, we begin with rigorous statements for each of these properties.
Their proofs appear in the subsequent sections, and the chapter concludes with a short list of other basic properties of strategies.
  

The first of these three properties establishes a bilinear dependence of the probability of a given measurement outcome upon the interacting strategy and co-strategy.
(By contrast, recall that the operational representation for strategies led to a nonlinear dependence of outcomes on strategies in Section \ref{subsec:naive-observations}.)

\begin{theorem}[\theoreminnerproduct]
\label{theorem:inner-product}

  Let $\set{Q_a}$ be a measuring strategy and let $\set{R_b}$ be a compatible measuring co-strategy.
  For each pair $(a,b)$ of measurement outcomes, the probability with which the interaction between $\set{Q_a}$ and $\set{R_b}$ yields $(a,b)$ is given by the inner product $\inner{Q_a}{R_b}$.

\end{theorem}

The second property provides a recursive characterization of $r$-round strategies in terms of $(r-1)$-round strategies.

\begin{theorem}[\thmchar]
\label{thm:char}

  Let $Q\in\pos{\kprod{\cY}{1}{r}\ot\kprod{\cX}{1}{r}}$ be an arbitrary positive semidefinite operator.
  The following hold:
  \begin{enumerate}
  
  \item 
    $Q$ is an $r$-round strategy if and only if there exists an $(r-1)$-round strategy $R$ with the property that \( \ptr{\cY_r}{Q} = R\ot I_{\cX_r}. \)
    Moreover, $R$ is obtained from $Q$ by terminating that strategy after $r-1$ rounds.
  
  \item 
    $Q$ is an $r$-round co-strategy if and only if there exists an operator $R$ for which \( Q=R\ot I_{\cY_r} \) and $\ptr{\cX_r}{R}$ is an $(r-1)$-round co-strategy.
    Moreover, $\ptr{\cX_r}{R}$ is obtained from $Q$ by terminating that co-strategy after $r-1$ rounds.
  
  \item 
    Every $r$-round strategy or co-strategy $Q$ may be described by isometries in such a way that the final memory space has dimension equal to $\rank(Q)$.
  \end{enumerate}

\end{theorem}

The recursive characterization of Theorem \ref{thm:char} may be equivalently expressed as an explicit list of linear constraints on positive semidefinite operators.
These linear constraints are efficiently checkable and hence amenable to standard algorithms for semidefinite optimization problems.

\begin{numberedtheorem}{\ref{thm:char}}[Characterization of strategies, alternate version]

  An operator $Q\in\pos{\kprod{\cY}{1}{r}\ot\kprod{\cX}{1}{r}}$ is an $r$-round strategy if and only if there exist operators
  \[ Q_k\in\pos{\kprod{\cY}{1}{k}\ot\kprod{\cX}{1}{k}}\] for $k=1,\dots,r-1$
  such that
  \begin{align*}
    \ptr{\kprod{\cY}{k}{r}}{Q} &= Q_{k-1} \ot I_{\kprod{\cX}{k}{r}} \quad (2\leq k\leq r),\\
    \ptr{\kprod{\cY}{1}{r}}{Q} & = I_{\kprod{\cX}{1}{r}}.
  \end{align*}
  Moreover, each $Q_k$ is obtained from $Q$ by terminating that strategy after $k$ rounds.

\end{numberedtheorem}

It follows immediately from this characterization that the sets of strategies and co-strategies are compact and convex.
For completeness, we provide a formal proof of this important fact as Proposition \ref{prop:convexity} in Section \ref{sec:other-properties}.

Our third property of strategies provides a formula for the maximum probability with which some co-strategy can force a given measuring strategy to output a given measurement outcome.
Whereas the previous two properties are fundamental---both in their statements and their proofs---this third property is more advanced.
It's proof relies crucially upon the previous two properties, as well as more advanced ideas from analysis such as convex polarity or semidefinite optimization duality.

\def\thmmaxprob{Maximum output probability}
\begin{theorem}[\thmmaxprob]
\label{thm:max-prob}

  Let $\set{Q_m}$ be a measuring strategy.
  For each outcome $m$, the maximum probability with which $\set{Q_m}$ can be forced to output $m$ by a compatible co-strategy is given by
  \[ \min \Set{ p\in[0,1] : Q_m \preceq pR \textrm{ for some strategy } R }. \]
  An analogous result holds for co-strategies.

\end{theorem}

\section{\theoreminnerproduct}

This section is devoted to a proof of Theorem \ref{theorem:inner-product} (\theoreminnerproduct).
We preface the proof with a discussion of a super-operator called the \emph{contraction}, whose useful properties will be employed in the proof that follows.

\subsubsection{The contraction operation}

Informally speaking, the contraction operation is defined so that a composition of operators
\[ B_r A_r B_{r-1} \cdots B_1 A_1 \]
is the image under the contraction of the tensor product
\[ B_r \ot \cdots \ot B_1 \ot A_r \ot \cdots \ot A_1. \]
The ability to ``unravel'' operator compositions in this fashion is useful for our purpose because it allows us to isolate the actions of Alice $(A_1,\dots,A_r)$ from those of Bob $(B_1,\dots,B_r)$ in some $r$-round interaction.

Still speaking informally, the existence of an operation such as the contraction follows from the fact that the composition $BA$ is multilinear in the operators $A,B$.
Of course, the tensor product $A\ot B$ is also multilinear in these operators.
Indeed, as noted in Section \ref{sec:intro:linalg}, the tensor product possesses a special \emph{universality} property whereby any multilinear mapping on $A,B$ could equivalently be expressed as a linear mapping on the tensor product $A\ot B$.
In our case, that linear mapping is the contraction operation.

Let us formally define this operation.

\begin{definition}[Contraction operation]

  The \emph{contraction operation} in its full generality is more easily defined as a linear functional on vectors.
  For any complex Euclidean space $\cV$ we define
  \[ \contract \ : \ \cV\ot\cV \to \mathbb{C} \ : \ u \mapsto \Pa{ \sum_{i=1}^{\dim(\cV)} e_i^*\ot e_i^* } u \]
  where $\set{e_1,\dots,e_{\dim(\cV)}}$ denotes the standard basis for $\cV$.
  
  Just as the trace function is tensored with the identity to yield the partial trace, the contraction is often tensored with the identity to yield the \emph{partial contraction over $\cV$}:
  \[ \contract[\cV] \ : \ \cV\ot\cA\ot\cV \to \cA \ : \ u \mapsto \Pa{ \sum_{i=1}^{\dim(\cV)} e_i^*\ot I_\cA \ot e_i^* } u. \]
  (The exact ordering of the three spaces in the above tensor product $\cV\ot\cA\ot\cV$ is immaterial---the contraction is defined similarly for other orderings.)
  
  The partial contraction is often viewed as a super-operator as follows.
  If $u=\col{X}$ for some operator $X:\cV\to\cV\ot\cA$ then we may dispense with the $\col{\cdot}$ notation and simply write
  \[ \con{\cV}{X} = \sum_{i=1}^{\dim(\cV)} \Pa{e_i^*\ot I_\cA} X e_i. \]
  Finally, we write $\contract[\cV,\cW]$ as shorthand for the composition $\contract[\cV] \circ \contract[\cW]$.
\end{definition}

Two useful properties of the contraction operation---including the ability to ``unravel'' operator compositions---are noted in the following proposition.
Each item in this proposition is proven by a straightforward but tedious exercise in ``index gymnastics.''

\begin{proposition}

  The following hold:
  \begin{enumerate}
  
  \item \label{item:contraction:hookup}
  For any operators $A:\cX\to\cA\ot\cY$ and $B:\cB\ot\cY\to\cZ$ we have
  \[ \con{\cY}{A\ot B} = \Pa{B\ot I_\cA}\Pa{A\ot I_\cB}. \]
  In particular, if $\cA=\cB=\mathbb{C}$ then $\con{\cY}{A\ot B} = BA$.
  
  \item \label{item:contraction:inner}
  For any operators $A:\cX\to\cY$ and $B:\cY\to\cX$ we have
  \[ \con{\cX,\cY}{A\ot B} = \ptr{}{AB} = \inner{A^*}{B}. \]
  
  \end{enumerate}

\end{proposition}

\begin{proof}
We begin with item \ref{item:contraction:hookup}.
Throughout the proof it might be helpful to remember that $A\ot B$ and its contraction take the following forms:
\begin{align*}
  A\ot B \ &: \ \cX\ot\cB\ot\cY \to \cA\ot\cY\ot\cZ, \\
  \con{\cY}{A\ot B} \ &: \ \cX\ot\cB \to \cA\ot\cZ.
\end{align*}
We denote the standard bases of $\cA$, $\cB$, $\cX$, $\cY$, and $\cZ$ by
\begin{align*}
\Set{a_1,\dots,a_{\dim(\cA)}} &\subset \cA, & \Set{x_1,\dots,x_{\dim(\cX)}} &\subset \cX, \\
\Set{b_1,\dots,b_{\dim(\cB)}} &\subset \cB, & \Set{y_1,\dots,y_{\dim(\cY)}} &\subset \cY, \\
& & \Set{z_1,\dots,z_{\dim(\cZ)}} &\subset \cZ.
\end{align*}
Then for some complex numbers $\alpha_{i,j,k},\beta_{l,m,n}\in\Complex$ we may write
\begin{align*}
  A &= \sum_{i=1}^{\dim(\cY)}\sum_{j=1}^{\dim(\cX)}\sum_{k=1}^{\dim(\cA)} \alpha_{i,j,k} (a_k\ot y_i)x_j^*, \\
  B &= \sum_{l=1}^{\dim(\cY)}\sum_{m=1}^{\dim(\cB)}\sum_{n=1}^{\dim(\cZ)} \beta_{l,m,n} z_n(b_m\ot y_l)^*.
\end{align*}
By the definitions of the contraction and of matrix multiplication, we obtain
\begin{align*}
  \con{\cY}{A\ot B}
  &= \sum_{o=1}^{\dim(\cY)} \Pa{I_{\cA\ot\cZ} \ot y_o^*} \Pa{A\ot B} \Pa{I_{\cX\ot\cB} \ot y_o} \\
  &= \sum_{j,k,m,n} \Pa{\sum_{o=1}^{\dim(\cY)} \alpha_{o,j,k} \beta_{o,m,n}} (a_k\ot z_n)(x_j\ot b_m)^* \\
  &= \Pa{B\ot I_\cA}\Pa{A\ot I_\cB}.
\end{align*}

For item \ref{item:contraction:inner}, we denote the standard bases of $\cX$ and $\cY$ by
\begin{align*}
\Set{x_1,\dots,x_{\dim(\cX)}} &\subset \cX, \\
\Set{y_1,\dots,y_{\dim(\cY)}} &\subset \cY.
\end{align*}
Then for some complex numbers $\alpha_{i,j},\beta_{l,m}\in\Complex$ we may write
\begin{align*}
  A &= \sum_{i=1}^{\dim(\cX)}\sum_{j=1}^{\dim(\cY)} \alpha_{j,i} y_jx_i^*, \\
  B &= \sum_{l=1}^{\dim(\cX)}\sum_{m=1}^{\dim(\cY)} \beta_{l,m} x_ly_m^*.
\end{align*}
By the definitions of the contraction, matrix multiplication, and the trace we obtain
\begin{align*}
  \con{\cX,\cY}{A\ot B}
  &= \sum_{o=1}^{\dim(\cX)}\sum_{p=1}^{\dim(\cY)} (y_p\ot x_o)^* \Pa{A\ot B} (x_o\ot y_p) \\
  &= \sum_{o=1}^{\dim(\cX)}\sum_{p=1}^{\dim(\cY)} \alpha_{p,o} \beta_{o,p} \\
  &= \ptr{}{AB} = \inner{A^*}{B}
\end{align*}
as desired.
\end{proof}

\subsubsection{Proof of Theorem \ref{theorem:inner-product} (\theoreminnerproduct)}

We are now ready to provide the promised proof.
The theorem is restated here for convenience.

\begin{numberedtheorem}{\ref{theorem:inner-product}}[\theoreminnerproduct]

  Let $\set{Q_a}$ be a measuring strategy and let $\set{R_b}$ be a compatible measuring co-strategy.
  For each pair $(a,b)$ of measurement outcomes, the probability with which the interaction between $\set{Q_a}$ and $\set{R_b}$ yields $(a,b)$ is given by the inner product $\inner{Q_a}{R_b}$.

\end{numberedtheorem}

\begin{proof}

Suppose $\set{Q_a}$ is described by isometries $A_1,\dots,A_r$ and a projective measurement $\set{\Pi_a}$.
Similarly, suppose $\set{R_b}$ is described by a pure state $u_0$, isometries $B_1,\dots,B_r$, and a projective measurement $\set{\Delta_b}$.
For each pair $(a,b)$ let $v_{a,b}\in\cZ_r\ot\cW_r$ denote the vector obtained by applying the measurement operator $\Pi_a\ot\Delta_b$ to the pure state of the entire system at the end of the interaction.
That is, $v_{a,b}$ is given by
\[ v_{a,b} \defeq \Pa{\Pi_a\ot\Delta_b} B_r A_r B_{r-1} \cdots B_1 A_1 u_0. \]
(Our notation here suppresses the numerous tensors with identity.)
Then the desired probability is equal to $\norm{v_{a,b}}^2$.
The remainder of this proof is dedicated to proving that
\[ \norm{v_{a,b}}^2 = \Inner{Q_a}{R_b}. \]

We use the contraction operation to pull apart the composition of $v_{a,b}$ and express it as a contraction of the tensor product
\[ \Pa{\Pi_a\ot I_{\cY_r}}A_r \ot A_{r-1} \ot \cdots \ot A_1 \ot \Delta_b B_r \ot B_{r-1} \ot \cdots \ot B_1 \ot u_0 \]
over every space except $\cZ_r$ and $\cW_r$.
For convenience, this expression for $v_{a,b}$ is written
\[ v_{a,b} = \Con{\cX_1,\dots,\cX_r,\cY_1,\dots,\cY_r}{A\ot B} \]
where the operators
\begin{align*}
  A: \kprod{\cX}{1}{r}\to\kprod{\cY}{1}{r}\ot\cZ_r,\\
  B:\kprod{\cY}{1}{r}\to\kprod{\cX}{1}{r}\ot\cW_r
\end{align*}
are given by
\begin{align*}
  A &\defeq \Con{\cZ_1,\dots,\cZ_{r-1}}{\Pa{\Pi_a\ot I_{\cY_r}}A_r \ot A_{r-1} \ot \cdots \ot A_1}, \\
  B &\defeq \Con{\cW_0,\dots,\cW_{r-1}}{\Delta_b B_r \ot B_{r-1} \ot \cdots \ot B_1 \ot u_0}.
\end{align*}
With an eye toward the end of the proof, we observe that
\begin{align*}
  Q_a &= \ptr{\cZ_r}{\col{A}\row{A}},\\
  R_b &= \ptr{\cW_r}{\col{B^*}\row{B^*}}
\end{align*}
as per Definition \ref{def:strategy} (\defstrategy).

Let $\set{e_1,\dots,e_{\dim(\cZ_r)}}$ and $\set{f_1,\dots,f_{\dim(\cW_r)}}$ denote the standard bases of $\cZ_r$ and $\cW_r$, respectively, and for each $i,j$ define the operators
\begin{align*}
  A_{(i)} &: \kprod{\cX}{1}{r}\to\kprod{\cY}{1}{r},\\
  B_{(j)} &: \kprod{\cY}{1}{r}\to\kprod{\cX}{1}{r}
\end{align*}
by
\begin{align*}
  A_{(i)} &\defeq \Pa{e_i^*\ot I_{\kprod{\cY}{1}{r}}} A, \\
  B_{(j)} &\defeq \Pa{f_j^*\ot I_{\kprod{\cX}{1}{r}}} B.
\end{align*}
Again, with an eye toward the end of the proof we observe that
\begin{align*}
  \sum_{i=1}^{\dim(\cZ_r)} \Col{A_{(i)}}\Row{A_{(i)}} &= \ptr{\cZ_r}{\col{A}\row{A}} = Q_a,\\
  \sum_{j=1}^{\dim(\cW_r)}\Col{B_{(j)}^*}\Row{B_{(j)}^*}  &= \ptr{\cW_r}{\col{B^*}\row{B^*}} = R_b.
\end{align*}
The $(i,j)$th component of $v_{a,b}$ (in the standard basis) is given by
\[
  \Pa{e_i^*\ot f_j^*} v_{a,b}
  = \Con{\cX_1,\dots,\cX_r,\cY_1,\dots,\cY_r}{A_{(i)}\ot B_{(j)}}
  = \Inner{B_{(j)}^*}{A_{(i)}}.
\]
Employing the cyclic property of the trace, we find that the modulus squared of the $(i,j)$th component of $v_{a,b}$ is
\begin{align*}
  \Inner{B_{(j)}^*}{A_{(i)}} \Inner{A_{(i)}}{B_{(j)}^*}
  &= \Inner{ \Col{A_{(i)}}\Row{A_{(i)}} }{ \Col{B_{(j)}^*}\Row{B_{(j)}^*} }.
\end{align*}
Then the desired norm $\norm{v_{a,b}}^2$ is the sum of these moduli squared:
\begin{align*}
  \norm{v_{a,b}}^2
  &= \sum_{i,j} \Inner{ \Col{A_{(i)}}\Row{A_{(i)}} }{ \Col{B_{(j)}^*}\Row{B_{(j)}^*} } \\
  &= \Inner{ \sum_{i=1}^{\dim(\cZ_r)} \Col{A_{(i)}}\Row{A_{(i)}} }{ \sum_{j=1}^{\dim(\cW_r)}\Col{B_{(j)}^*}\Row{B_{(j)}^*} } \\
  &= \Inner{ Q_a }{ R_b }.
\end{align*}
\end{proof}

\section{\thmchar}

This section is devoted to a proof of Theorem \ref{thm:char} (\thmchar).
Before providing the proof, it is appropriate to comment on the relationship between our characterization of strategies and prior work on so-called ``no-signaling'' quantum operations.



\subsubsection{Relationship between strategies and no-signaling operations}

The content of Theorem \ref{thm:char} was originally established in 2002 within the context of no-signaling operations via the combined work of 
Beckman \emph{et al.}~\cite{BeckmanG+01}
and
Eggeling, Schlingemann, and Werner \cite{EggelingSW02}.
The proof presented in this thesis was developed by the present author and Watrous in 2007 \cite{GutoskiW07} within the context of quantum strategies and without any knowledge of this prior work.

Let us elaborate upon the connection between strategies and no-signaling operations.
Simply put, Theorem \ref{thm:char} states that $Q$ is an $r$-round strategy if and only if it obeys the partial trace condition \[ \ptr{\cY_r}{Q}=R\ot I_{\cX_r} \] for some $(r-1)$-round strategy $R$.
This partial trace condition also appears in Beckman \emph{et al.}~\cite[Theorem 8]{BeckmanG+01}, wherein it was established that the above condition also captures the \emph{one-directional no-signaling} property of certain quantum operations.


For quantum strategies, this property means that it is impossible to send information (a ``signal'') from the system $\cX_i$ to any of the systems $\cY_1,\dots,\cY_{i-1}$ for each $i$.
Intuitively, we should expect quantum strategies to obey a no-signaling condition of this sort.
Any ``strategy'' that disobeys this condition could be used to communicate backwards in time---from future rounds of interaction to previous rounds.
A world that permitted such clairvoyant strategies would be an interesting world indeed!
Alas, such a world is not causally consistent.

While Beckman \emph{et al.} showed that the partial trace condition of Theorem \ref{thm:char} is a \emph{necessary} condition for quantum strategies, it was Eggeling, Schlingemann, and Werner who established that this condition is also \emph{sufficient}.
In particular, they showed that any quantum operation that forbids signaling in \emph{one} direction may be implemented by two separate quantum operations, possibly with communication in the \emph{other} direction.

For quantum strategies, signaling is forbidden in the future-past direction, but permitted in the past-future direction.
In this context, such past-future communication is better known as ``memory'' and the two separate quantum operations can be taken to represent the actions of the strategy in two distinct rounds of the interaction.
In this way, a quantum strategy is constructed from an operation that obeys the one-directional no-signaling operation, just as in Theorem \ref{thm:char}.


%
%
%

\subsubsection{Proof of Theorem \ref{thm:char} (\thmchar)}

Let us present the promised proof.
The theorem is restated here for convenience.

\begin{numberedtheorem}{\ref{thm:char}}[\thmchar]

  Let $Q\in\pos{\kprod{\cY}{1}{r}\ot\kprod{\cX}{1}{r}}$ be an arbitrary positive semidefinite operator.
  The following hold:
  \begin{enumerate}
  
  \item \label{item:char:st}
    $Q$ is an $r$-round strategy if and only if there exists an $(r-1)$-round strategy $R$ with the property that \( \ptr{\cY_r}{Q} = R\ot I_{\cX_r}. \)
    Moreover, $R$ is obtained from $Q$ by terminating that strategy after $r-1$ rounds.
  
  \item \label{item:char:cst}
    $Q$ is an $r$-round co-strategy if and only if there exists an operator $R$ for which \( Q=R\ot I_{\cY_r} \) and $\ptr{\cX_r}{R}$ is an $(r-1)$-round co-strategy.
    Moreover, $\ptr{\cX_r}{R}$ is obtained from $Q$ by terminating that co-strategy after $r-1$ rounds.
  
  \item \label{item:char:mem}
    Every $r$-round strategy or co-strategy $Q$ may be described by isometries in such a way that the final memory space has dimension equal to $\rank(Q)$.
  \end{enumerate}
  
\end{numberedtheorem}

\begin{proof}

We begin with a proof of item \ref{item:char:st}.
Along the way, we will also prove item \ref{item:char:mem}.
Item \ref{item:char:cst} follows from item \ref{item:char:st}---a fact we establish at the end of this proof.

Suppose first that $Q$ is an $r$-round strategy and let $R$ denote the $(r-1)$-round strategy obtained from $Q$ by terminating that strategy after the first $r-1$ rounds.
We will prove that $\ptr{\cY_r}{Q} = R\ot I_{\cX_r}$.

Toward that end, let $\Xi_r, \Xi_{r-1}$ be the quantum operations
satisfying $Q=J(\Xi_r)$ and $R=J(\Xi_{r-1})$.
As illustrated in Figure \ref{fig:trace-out}, is clear that
the super-operators $\Pa{{\trace_{\cY_r}}\circ\Xi_r}$ and $\Pa{\Xi_{r-1}\ot{\trace_{\cX_r}}}$ are equal.
We have
\[
  \ptr{\cY_r}{Q}=\ptr{\cY_r}{J(\Xi_r)}=J({\trace_{\cY_r}}\circ\Xi_r)
  = J(\Xi_{r-1}\ot{\trace_{\cX_r}})=R\ot I_{\cX_r}
\]
as desired.
\begin{figure}[h]
  \begin{center}
\includegraphics{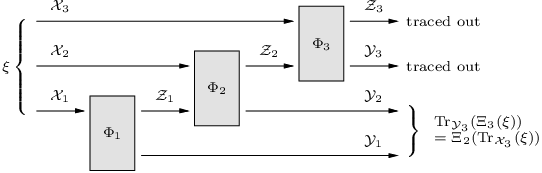}
  \end{center}
  \caption{For a three-round strategy, the super-operators $\Pa{{\trace_{\cY_3}}\circ\Xi_3}$ and $\Pa{\Xi_2\ot{\trace_{\cX_3}}}$ are equal.}
  \label{fig:trace-out}
\end{figure}

Next, assume that $Q$ satisfies $\ptr{\cY_r}{Q}=R\ot I_{\cX_r}$ for some $(r-1)$-round strategy $R$.
We will prove that $Q$ is an $r$-round strategy.
Along the way, we will also establish item \ref{item:char:mem} in the statement of the theorem.

This portion of the proof is by induction on the number of rounds $r$.
We begin with the base case $r=1$.
It is clear that $Q$ is a one-round strategy if and only if $Q=J(\Phi_1)$ for some quantum operation $\Phi_1:\lin{\cX_1}\to\lin{\cY_1}$.
As $\Phi_1$ is a quantum operation, it must be the case that \( \ptr{\cY_1}{Q}=I_{\cX_1}. \)
By convention, the zero-round strategy $R$ that we seek must be represented by the scalar 1.
The desired expression \( \ptr{\cY_1}{Q}=R\ot I_{\cX_1}\) follows from the simple observation that $I_{\cX_1}=1\ot I_{\cX_1}$.
Item \ref{item:char:mem} for this case follows from the usual Stinespring representation: there exists a space $\cZ_1$ with $\dim(\cZ_1)=\rank(Q)$ and an isometry $A_1:\cX_1\to\cY_1\ot\cZ_1$ such that $\Phi_1 : X \mapsto \ptr{\cZ_1}{A_1XA_1^*}$.

For the case $r\geq 2$, let $A_1,\dots,A_{r-1}$ be isometries that describe $R$ and let
$A=A_{r-1}\cdots A_1$ denote the $(r-1)$-fold composition of these isometries, so that
\[R=\ptr{\cZ_{r-1}}{\col{A}\row{A}}.\]
By the induction hypothesis, the memory space $\cZ_{r-1}$ has $\dim(\cZ_{r-1})=\rank(R)$.
As required for item \ref{item:char:mem}, we let $\cZ_r$ be a complex Euclidean space with dimension equal to $\rank(Q)$.
Let $B:\kprod{\cX}{1}{r}\to\kprod{\cY}{1}{r}\ot\cZ_r$ be any operator satisfying
\[ \ptr{\cZ_r}{\col{B}\row{B}} = Q.\]
Such a choice of $B$ must exist given that the dimension of $\cZ_r$ is large enough to admit a purification of $Q$.
Note that $\col{B}$ is also a purification of $R\ot I_{\cX_r}$:
\[ \ptr{\cY_r\ot\cZ_r}{\col{B}\row{B}} = \ptr{\cY_r}{Q} = R\ot I_{\cX_r}.\]
We will now identify a second purification of $R\ot I_{\cX_r}$.
Toward that end, let $\cV$ be a complex Euclidean space with $\dim(\cV)=\dim(\cX_r)$ and let $V:\cX_r\to\cV$ be an arbitrary unitary operator.
Then \[ \ptr{\cV}{\col{V}\row{V}}=I_{\cX_r} \] and so \[\ptr{\cZ_{r-1}\ot\cV}{\col{A\ot V}\row{A\ot V}} = R\ot I_{\cX_r}.\]
We will now use the isometric equivalence of purifications to define an isometry
\[ A_r : \cX_r\ot\cZ_{r-1}\to\cY_r\ot\cZ_r \]
for the $r$th round that will complete the proof.
Because $\cZ_{r-1}\ot\cV$ has the minimal dimension required to admit a purification of $R\ot I_{\cX_r}$, it follows that there exists an isometry \( U:\cZ_{r-1}\ot\cV\to\cY_r\ot\cZ_r \) such that
\[ \Pa{ I_{\kprod{\cY}{1}{r-1}} \ot U \ot I_{\kprod{\cX}{1}{r}} } \col{A\ot V} = \col{B}.\]
This expression may equivalently be written \[B=\Pa{I_{\kprod{\cY}{1}{r-1}}\ot U}(A\ot V).\]
We now define $A_r\defeq U(I_{\cZ_{r-1}}\ot V)$ so that
\[ B= \pa{ I_{\kprod{\cY}{1}{r-1}} \ot A_r } (A\ot I_{\cX_r}).\]
In other words, $B$ is given by the $r$-fold composition $A_r\cdots A_1$.
As \[Q=\ptr{\cZ_r}{\col{B}\row{B}},\] it follows that $Q$ is an $r$-round strategy described by the isometries $A_1,\dots,A_r$.
The proofs of items \ref{item:char:st} and \ref{item:char:mem} are thus complete.

Let us now prove item \ref{item:char:cst}.
Essentially, the proof consists of several applications of item \ref{item:char:st} with different choices of input and output spaces.

As noted in Definition \ref{def:co-strategy} (\defcostrategy), each $r$-round co-strategy may be viewed as an $(r+1)$-round strategy with input spaces $\mathbb{C},\cY_1,\dots,\cY_r$ and output spaces $\cX_1,\dots,\cX_r,\mathbb{C}$ and \emph{vice versa}.
With this fact in mind, item \ref{item:char:st} tells us
\begin{align*}
  &\textrm{$Q$ is an $(r+1)$-round strategy}\\
  &\qquad\textrm{for input spaces $\mathbb{C},\cY_1,\dots,\cY_r$ and output spaces $\cX_1,\dots,\cX_r,\mathbb{C}$}\\
  \iff {}&\ptr{\mathbb{C}}{Q}=R\ot I_{\cY_r} \textrm{ where $R$ is an $r$-round strategy}\\
  &\qquad\textrm{for the input spaces $\mathbb{C},\cY_1,\dots,\cY_{r-1}$ and output spaces $\cX_1,\dots,\cX_r$.}
\end{align*}
(Of course, $\ptr{\mathbb{C}}{X}=X$ for every $X$.)
To complete the proof, it suffices to show that
\begin{align*}
  &\textrm{$R$ is an $r$-round strategy}\\
  &\qquad \textrm{for input spaces $\mathbb{C},\cY_1,\dots,\cY_{r-1}$ and output spaces $\cX_1,\dots,\cX_r$}\\
  \iff {}&\textrm{$\ptr{\cX_r}{R}$ is an $r$-round strategy}\\
  &\qquad \textrm{for input spaces $\mathbb{C},\cY_1,\dots,\cY_{r-1}$and output spaces $\cX_1,\dots,\cX_{r-1},\mathbb{C}$}.
\end{align*}
For once this equivalence is shown,
item \ref{item:char:cst} follows from the fact that
the linear conditions in question are robust with respect to taking adjoints.
In particular, it holds that
\[
  \ptr{\cX}{\jam{\Xi}}=\jam{\Lambda}\ot I_\cY
  \ \iff \
  \ptr{\cX}{\jam{\Xi^*}}=\jam{\Lambda^*}\ot I_\cY
\]
for all super-operators $\Xi,\Lambda$ and all appropriate choices of spaces $\cX,\cY$.

It remains only to establish the stated equivalence between $R$ and $\ptr{\cX_r}{R}$.
Toward that end, suppose first that $R$ is an $r$-round strategy for the input spaces $\mathbb{C},\cY_1,\dots,\cY_{r-1}$ and output spaces $\cX_1,\dots,\cX_r$.
By item \ref{item:char:st} it follows that
\[ \ptr{\cX_r}{R}=S\ot I_{\cY_{r-1}} \]
where $S$ is an $(r-1)$-round strategy for the input spaces $\mathbb{C},\cY_1,\dots,\cY_{r-2}$ and output spaces $\cX_1,\dots,\cX_{r-1}$.
It is not hard to see that $S\ot I_{\cY_{r-1}}$ denotes a valid $r$-round strategy for the input spaces $\mathbb{C},\cY_1,\dots,\cY_{r-1}$ and output spaces $\cX_1,\dots,\cX_{r-1},\mathbb{C}$.
Hence, $\ptr{\cX_r}{R}$ is as claimed.

Conversely, suppose that $\ptr{\cX_r}{R}$ is an $r$-round strategy for the input spaces $\mathbb{C},\cY_1,\dots,\cY_{r-1}$ and output spaces $\cX_1,\dots,\cX_{r-1},\mathbb{C}$.
By item \ref{item:char:st} it follows that 
\[ \ptr{\mathbb{C}}{\ptr{\cX_r}{R}}=S\ot I_{\cY_{r-1}} \]
where $S$ is an $(r-1)$-round strategy for the input spaces $\mathbb{C},\cY_1,\dots,\cY_{r-2}$ and output spaces $\cX_1,\dots,\cX_{r-1}$.
By item \ref{item:char:st} again, it must be that $R$ is an $r$-round strategy for the input spaces $\mathbb{C},\cY_1,\dots,\cY_{r-1}$ and output spaces $\cX_1,\dots,\cX_r$ as desired.
\end{proof}

\section{Maximum output probabilities}

This section is devoted to a proof of Theorem \ref{thm:max-prob} (\thmmaxprob).
We offer two distinct proofs of this theorem.
The first is provided in Section \ref{subsec:convex-polarity} and employs the formalism of convex polarity, while the second is provided in Section \ref{subsec:semidefinite-duality} and employs the formalism of semidefinite optimization duality.
Both proofs make use of the fact that the sets of strategies and co-strategies are compact and convex---an immediate implication of Theorem \ref{thm:char} that was noted at the start of this chapter and shall be proven at the end of this chapter.

The theorem is restated here for convenience and is followed by a simple lemma employed in both our proofs.

\begin{numberedtheorem}{\ref{thm:max-prob}}[\thmmaxprob]

  Let $\set{Q_m}$ be a measuring strategy.
  For each outcome $m$, the maximum probability with which $\set{Q_m}$ can be forced to output $m$ by a compatible co-strategy is given by
  \[ \min \Set{ p\in[0,1] : Q_m \preceq pR \textrm{ for some strategy } R }. \]
  An analogous result holds for co-strategies.

\end{numberedtheorem}

\begin{lemma}
\label{lm:max-prob-switch}

  Suppose that Theorem \ref{thm:max-prob} is known to hold for strategies.
  Then Theorem \ref{thm:max-prob} also holds for co-strategies.
  The converse is also true.

\end{lemma}

\begin{proof}
Suppose that Theorem \ref{thm:max-prob} is known to hold for strategies.
Let $\set{R_b}$ be an $r$-round measuring co-strategy for input spaces $\cX_1,\dots,\cX_r$ and output spaces $\cY_1,\dots,\cY_r$.
For each outcome $b$ let $\Xi_b:\lin{\kprod{\cY}{1}{r}}\to\lin{\kprod{\cX}{1}{r}}$ be the super-operator with $R_b=\jam{\Xi_b^*}$.
Then $\set{\jam{\Xi_b}}$ denotes an $(r+1)$-round measuring strategy for input spaces $\mathbb{C},\cY_1,\dots,\cY_r$ and output spaces $\cX_1,\dots,\cX_r,\mathbb{C}$.
As such, it holds that
\begin{align*}
& \max \Set{ \inner{\jam{\Xi_b}}{\jam{\Phi}} : \jam{\Phi} \textrm{ is a compatible co-strategy} } \\
={}& \min \Set{ p\in[0,1] : \jam{\Xi_b} \preceq p\jam{\Psi} \textrm{ for some strategy } \jam{\Psi} }.
\end{align*}
The result follows from the fact that $\inner{\jam{\Xi_b}}{\jam{\Phi}}=\inner{\jam{\Xi_b^*}}{\jam{\Phi^*}}$ and
\[ \jam{\Xi_b} \preceq p\jam{\Psi} \iff \jam{\Xi_b^*} \preceq p\jam{\Psi^*} \]
for all choices of super-operators $\Xi_b$, $\Phi$, and $\Psi$.
The proof of the converse statement is identical.
%
%
\end{proof}

\subsection{Proof by convex polarity}
\label{subsec:convex-polarity}

Our first proof of Theorem \ref{thm:max-prob} (\thmmaxprob) employs a notion from convex analysis known as ``polarity.''
But before we discuss this notion in detail, let us first introduce some notation.
We let 
\[ \st_r\subset\pos{\kprod{\cY}{1}{r}\ot\kprod{\cX}{1}{r}} \]
denote the set of all $r$-round non-measuring strategies for input spaces $\cX_1,\dots,\cX_r$ and output spaces $\cY_1,\dots,\cY_r$.
Similarly, we let 
\[ \cst_r\subset\pos{\kprod{\cY}{1}{r}\ot\kprod{\cX}{1}{r}} \]
denote the set of all $r$-round non-measuring co-strategies for these input and output spaces.
In keeping with our convention, we have $\st_0=\cst_0=\set{1}$.
We mentioned earlier that the sets $\st_r$ and $\cst_r$ are compact and convex.

For any set $\bC$ of positive semidefinite operators, we write
\[ \sub{\bC} \defeq \Set{ X : 0 \preceq X \preceq Y \textrm{ for some } Y\in\bC }. \]
A key component of this proof of Theorem \ref{thm:max-prob} is a characterization of the polar sets of $\subst_r$ and $\csubst_r$.

\subsubsection{Introduction to polarity}

For any non-empty set $\bC$ of Hermitian operators, the \emph{polar} $\bC^\circ$ of $\bC$ is defined as
\[ \bC^\circ \defeq \Set{ A : \inner{B}{A}\leq 1 \textrm{ for all } B\in\bC } \]
and the \emph{support} and \emph{gauge} functions for $\bC$ are defined as
\begin{align*}
s(X\mid \bC) &\defeq \sup \Set{ \inner{X}{Y} : Y\in\bC } \\
g(X\mid \bC) &\defeq \inf \Set{ \lambda\geq 0 : X\in\lambda\bC }.
\end{align*}
Let us list some basic facts concerning these objects.

\begin{proposition}
\label{prop:polar-items}

  Let $\bC,\bD$ be non-empty sets of Hermitian operators.  The following hold:
  \begin{enumerate}
  \item
    If $\bC\subseteq\bD$ then $\bD^\circ\subseteq\bC^\circ$.
  \item
    If $-X\in\bC$ for each positive semidefinite operator $X$ then every element of $\bC^\circ$ is positive semidefinite.
  \item
    If $\bC$ is closed, convex, and contains the origin then the same is true of $\bC^\circ$.
    In this case we have $\bC^{\circ\circ}=\bC$ and
    \(
      s(\cdot\mid\bC) = g(\cdot\mid\bC^\circ).
    \)
  \end{enumerate}

\end{proposition}

The first two items of Proposition \ref{prop:polar-items} are elementary; a proof of the third may be found in Rockafellar \cite{Rockafellar70}.

\subsubsection{A characterization of polar sets of strategies}

We now establish a useful characterization of polar sets derived from strategies and co-strategies.

\def\proppolar{Polar sets of strategies}
\begin{proposition}[\proppolar]
\label{prop:polar}
  The following polarity relations hold for the sets $\subst_r$ and $\csubst_r$:
  \begin{align*}
    (\subst_r)^\circ  &= \Set{ X : X\preceq Q \textrm{ for some } Q\in\cst_r },\\
    (\csubst_r)^\circ &= \Set{ X : X\preceq Q \textrm{ for some } Q\in\st_r }.
  \end{align*}


\end{proposition}

We begin with a simple observation implying that the two equalities in Proposition \ref{prop:polar} are equivalent.

\begin{lemma} \label{lemma:switch-polar}
  Let $\bA,\bB$ be non-empty, closed, and convex sets of positive semidefinite operators, and suppose
  \[  (\sub{\bA})^\circ = \Set{ X : X\preceq Q \text{ for some } Q\in\bB}. \]
  Then
  \[  (\sub{\bB})^\circ = \Set{ Y : Y\preceq R \text{ for some } R\in\bA}. \]
\end{lemma}

\begin{proof}
Let \[ \bC = \Set{ Y : Y\preceq R \text{ for some } R\in\bA }. \]
The lemma is proved by showing $\bC^\circ=\sub{\bB}$, from which the desired result $\bC=(\sub{\bB})^\circ$ immediately follows.
Let us start by proving $\bC^\circ\subseteq\sub{\bB}$.
As $-P\in\bC$ for every positive semidefinite $P$, it follows that every element of the polar
$\bC^\circ$ is positive semidefinite.
Clearly $\sub{\bA}\subseteq\bC$, and therefore $\bC^\circ\subseteq(\sub{\bA})^\circ$.
By definition, $\sub{\bB}$ consists of the positive semidefinite elements of $(\sub{\bA})^\circ$, so we have $\bC^\circ\subseteq\sub{\bB}$.

On the other hand, we have that every $Q\in\sub{\bB}$ is contained in $(\sub{\bA})^\circ$, implying that $\inner{Q}{R}\leq 1$ for all $R\in\bA$.
As $Q$ is positive semidefinite, this also implies that $\inner{Q}{X}\leq 1$ for all $X\preceq R$.
Consequently, $Q\in \bC^\circ$.
Thus $\sub{\bB} = \bC^{\circ}$ as desired.
\end{proof}

We also require a technical statement that simplifies computations involving $\subst_r$ and $\csubst_r$.

\begin{lemma} \label{lemma:polar-technical}
  Let $\bD \subseteq \her{\cV}$ be any closed, convex set that contains the origin and let
  $\bC\subseteq \her{\cV\ot\cW}$ be given by
  \[ \bC = \Set{ X :  X \preceq Y\otimes I_\cW \textrm{ for some } Y\in\bD }. \]
  Then $\bC$ is also a closed, convex set that contains the origin and
  \[ \bC^\circ = \Set{ Q : Q\succeq 0 \textrm{ and } \ptr{\cW}{Q} \in \bD^\circ }. \]
\end{lemma}

\begin{proof}
  That $\bC$ is closed, convex, and contains the origin follows immediately from its definition.
  Let us compute $\bC^\circ$.
  The assumption $0\in\bD$ implies that $-R \in \bC$ for every positive semidefinite $R$, and hence every element of the polar $\bC^{\circ}$ must be positive semidefinite.
  If it is the case that $Q\in\bC^{\circ}$ then for all $Y\in\bD$ we have
  \[ 1 \geq \inner{Q}{Y\otimes I_{\cW}} = \inner{\ptr{\cW}{Q}}{Y} \]
  and so $\ptr{\cW}{Q} \in \bD^{\circ}$ as desired.

  On the other hand, if $\ptr{\cW}{Q} \in \bD^{\circ}$ then
  \[ 1 \geq  \inner{\ptr{\cW}{Q}}{Y} = \inner{Q}{Y\otimes I_{\cW}} \]
  for all $Y\in\bD$.
  If in addition $Q$ is positive semidefinite then it also holds that $\inner{Q}{X}\leq 1$ for all $X\preceq Y\otimes I_{\cW}$ and therefore $Q\in\bC^\circ$.
\end{proof}

We are now ready to prove the desired characterization of polar sets of strategies.

\begin{proof}[Proof of Proposition \ref{prop:polar} (\proppolar)]
  The proof is by induction on the number of rounds $r$.
  For the base case $r=0$, we have
  $\subst_0=\csubst_0=[0,1]$  and
  $(\subst_0)^\circ = (\csubst_0)^\circ = (-\infty,1]$
  and so the lemma holds when $r=0$.

  For the general case, we note that the two items in the statement of the proposition are equivalent by Lemma \ref{lemma:switch-polar}, so it suffices to prove only the first.
  Let
  \begin{align*}
    \bA &= \Set{ X : X \succeq 0 \textrm{ and } \ptr{\cY_r}{X} \in \bB },\\
    \bB &= \Set{ Y : Y \preceq R\ot I_{\cX_r} \textrm{ for some } R\in\subst_{r-1} }.
  \end{align*}
  Then by Theorem \ref{thm:char} (\thmchar) we have \[\subst_r = \bA.\]
  We apply Lemma \ref{lemma:polar-technical} twice---once with
  $(\bC^\circ,\bD^\circ,\cW)=(\bA,\bB,\cY_r)$ and once with
  $(\bC,\bD,\cW)=(\bB,\subst_{r-1},\cX_r)$---to obtain
  \begin{align*}
    (\subst_r)^\circ = \bA^\circ &=
      \Set{ X : X\preceq Q\ot I_{\cY_r} \textrm{ for some } Q\in\bB^\circ },\\
    \bB^\circ &= \Set{ Y : Y\succeq 0 \textrm{ and } \ptr{\cX_r}{Y}\in(\subst_{r-1})^\circ }.
  \end{align*}
  By the induction hypothesis we have
  \[ (\subst_{r-1})^\circ = \Set{ X : X\preceq Q \textrm{ for some } Q\in\cst_{r-1} }. \]
  Substituting this expression into the above expression for $(\subst_r)^\circ$ we find that the proposition follows from Theorem \ref{thm:char} (\thmchar).
\end{proof}

\subsubsection{First proof of Theorem~\ref{thm:max-prob} (\thmmaxprob)}

\begin{proof}[First proof of Theorem~\ref{thm:max-prob}]
  We prove the theorem for strategies---the result for co-strategies then follows from Lemma \ref{lm:max-prob-switch}.
  Let $p$ denote the maximum probability with which $\set{Q_m}$ can be forced to output $m$ in an interaction with some compatible co-strategy.
  It follows from Theorem~\ref{theorem:inner-product} (\theoreminnerproduct) that $p = s(Q_m \mid \cst_r)$.
  Using Proposition~\ref{prop:polar} (\proppolar), along with the fact that $Q_m$ is positive semidefinite, we have
  \[
    s(Q_m \mid \cst_r) = s(Q_m \mid \csubst_r)
    = g(Q_m \mid (\csubst_r)^{\circ}) = g(Q_m \mid \subst_r),
  \]
  which completes the proof.
\end{proof}

\subsection{Proof by semidefinite optimization duality}
\label{subsec:semidefinite-duality}

Our second proof of Theorem \ref{thm:max-prob} (\thmmaxprob) employs the powerful machinery of semidefinite optimization duality.
The idea is to construct a semidefinite optimization problem that captures the maximum output probability, compute its dual problem, and then show that the primal and dual problems satisfy the conditions for so-called ``strong'' duality.

\subsubsection{Overview of semidefinite optimization}

The semidefinite optimization problems we consider are expressed in \emph{super-operator form}.
While the super-operator form differs superficially from the more conventional \emph{standard form} for these problems, the two forms can be shown to be equivalent and the super-operator form is more convenient for our purpose.
Watrous provides a helpful overview of this form of semidefinite optimization \cite{Watrous09}.
For completeness, that overview is reproduced here.

Let $\cX,\cY$ be complex Euclidean spaces.
A \emph{semidefinite optimization problem} for these spaces is specified by a triple $(\Phi,A,B)$ where $\Phi:\lin{\cX}\to\lin{\cY}$ is a Hermitian-preserving super-operator and $A\in\her{\cX}$ and $B\in\her{\cY}$.
This triple specifies two optimization problems:
\begin{align*}
\textrm{\underline{Primal problem}} && \textrm{\underline{Dual problem}} \\
\textrm{maximize} \quad & \inner{A}{X} & \textrm{minimize} \quad & \inner{B}{Y} \\
\textrm{subject to} \quad & \Phi(X) \preceq B & \textrm{subject to} \quad & \Phi^*(Y) \succeq A\\
& X \in \pos{\cX} & & Y \in \pos{\cY}
\end{align*}
An operator $X$ obeying the constraints of the primal problem is said to be \emph{primal feasible}, while an operator $Y$ obeying the constraints of the dual problem is called \emph{dual feasible}.
The functions $X\mapsto\inner{A}{X}$ and $Y\mapsto\inner{B}{Y}$ are called the primal and dual \emph{objective functions}, respectively.
We let
\begin{align*}
  \alpha &\defeq \sup \Set{ \inner{A}{X}: \textrm{ $X$ is primal feasible} } \\
  \beta  &\defeq \inf \Set{ \inner{B}{Y}: \textrm{ $Y$ is dual feasible} }
\end{align*}
denote the \emph{optimal values} of the primal and dual problems.
(If there are no primal or dual feasible operators then we adopt the convention $\alpha=-\infty$ and $\beta=\infty$, respectively.)
It is not always the case that the optimal value is actually attained by a feasible operator for these problems.

Semidefinite optimization problems derive great utility from the notions of \emph{weak} and \emph{strong duality}.
Essentially, weak duality asserts that $\alpha\leq\beta$ for all triples $(\Phi,A,B)$, whereas strong duality provides conditions on $(\Phi,A,B)$ under which $\alpha=\beta$.
Two such conditions are stated explicitly as follows.

\def\factstrongduality{Strong duality conditions}
\begin{fact}[\factstrongduality]
\label{fact:strong-duality}

  Let $(\Phi,A,B)$ be a semidefinite optimization problem.
  The following hold:
  \begin{enumerate}
  
  \item \label{item:strong-primal}
    (Strict primal feasibility.)
    Suppose $\beta$ is finite and there exists a primal feasible operator $X$ such that $X$ is positive definite and $\Phi(X)\prec B$.
    Then $\alpha=\beta$ and $\beta$ is achieved by some dual feasible operator.

  \item \label{item:strong-dual}
    (Strict dual feasibility.)
    Suppose $\alpha$ is finite and there exists a dual feasible operator $Y$ such that $Y$ is positive definite and $\Phi^*(Y)\succ A$.
    Then $\alpha=\beta$ and $\alpha$ is achieved by some primal feasible operator.
  
  \end{enumerate}
  
\end{fact}

\subsubsection{A semidefinite optimization problem for maximum output probabilities}

We now construct a triple $(\Phi,A,B)$ whose primal problem captures the maximum probability with which a given measuring co-strategy $\set{R_b}$ can be forced to output a given outcome $b$ by some compatible strategy.
The dual problem will capture the minimization stated in Theorem \ref{thm:max-prob}.
Equality between these two values will follow when we establish strong duality for $(\Phi,A,B)$.

It is convenient to express the components of our triple $(\Phi,A,B)$ in block form.
The operators $A,B$ are given by
\[
  A = \Pa{
  \begin{array}{cccc}
    0 &&& 0\\ & \ddots \\ && 0 \\ 0 &&& R_b
  \end{array} }
  \qquad
  B = \Pa{
  \begin{array}{cccc}
    I_{\cX_1} &&& 0 \\ & 0 \\ && \ddots \\ 0 &&& 0
  \end{array} }
\]
and the super-operator $\Phi$ is given by
\begin{align*}
  \Phi :&
  \Pa{
  \begin{array}{ccc}
    S_1 \\
    & \ddots\\
    && S_r
  \end{array} } \\
  \mapsto&
  \Pa{
  \begin{array}{cccc}
    \ptr{\cY_1}{S_1} &&& 0 \\
    & \ptr{\cY_2}{S_2} - S_1\ot I_{\cX_2} \\
    && \ddots\\
    0 &&& \ptr{\cY_r}{S_r} - S_{r-1}\ot I_{\cX_r}
  \end{array} }.
\end{align*}
To be clear, $\Phi$ depends only upon the diagonal blocks $S_1,\dots,S_r$ of the input matrix.

\subsubsection{Correctness of the primal problem}

Let us verify that the primal problem expresses the desired maximum output probability for the measurement outcome $b$.
For any operator $S$ with diagonal blocks $S_1,\dots,S_r$, the primal objective value at $S$ is given by $\inner{A}{S}=\inner{R_b}{S_r}$.
Hence, the desired quantity is precisely the supremum of $\inner{A}{S}$ over all $S$ for which $S_r$ is a strategy.
It remains only to verify that the constraint $\Phi(S)\preceq B$ enforces this property of $S_r$.
The following lemma serves that purpose.

\begin{lemma}[Correctness of the primal problem]
\label{lm:primal}

For every primal feasible solution $S$ there exists another primal feasible solution $S'$ whose objective value meets or exceeds that of $S$ and whose diagonal blocks $S_1',\dots,S_r'$ have the property that $S_i'$ is an $i$-round non-measuring strategy for each $i=1,\dots,r$.

\end{lemma}

\begin{proof}

Let $S_1,\dots,S_r$ denote the diagonal blocks of $S$.
As $S$ is primal feasible, we know that
\begin{align*}
  \ptr{\cY_1}{S_1} &\preceq I_{\cX_1} \\
  \ptr{\cY_2}{S_2} &\preceq S_1\ot I_{\cX_2} \\
  \vdots \\
  \ptr{\cY_r}{S_r} &\preceq S_{r-1}\ot I_{\cX_r}
\end{align*}
It is clear that there is a $S_1'\succeq S_1$ such that $\ptr{\cY_1}{S_1'}=I_{\cX_1}$.
Similarly, for each $i=2,\dots,r$ there is a $S_i'\succeq S_i$ such that $\ptr{\cY_i}{S_i'}=S_{i-1}'\ot I_{\cX_i}$.
The desired operator $S'$ is then obtained from $S$ by replacing its diagonal blocks with $S_1',\dots,S_r'$.
That $\inner{A}{S'}\geq\inner{A}{S}$ follows from the fact that $A\succeq 0$ and $S'\succeq S$.
That $S_1',\dots,S_r'$ are non-measuring strategies follows from Theorem \ref{thm:char} (\thmchar).
\end{proof}

\subsubsection{Correctness of the dual problem}


Before we can show correctness of the dual problem, we must compute the adjoint super-operator $\Phi^*$.
Let us verify that $\Phi^*=\Psi$ where $\Psi$ is given by
\begin{align*}
  \Psi :&
  \Pa{
  \begin{array}{ccc}
    T_1 \\
    & \ddots\\
    && T_r
  \end{array} } \\
  \mapsto&
  \Pa{
  \begin{array}{cccc}
    T_1\ot I_{\cY_1} - \ptr{\cX_2}{T_2} &&& 0 \\
    & \ddots \\
    && T_{r-1}\ot I_{\cY_{r-1}} - \ptr{\cX_r}{T_r}\\
    0 &&& T_r\ot I_{\cY_r}
  \end{array} }.
\end{align*}
Let $S,T$ be any operators in the domains of $\Phi,\Psi$, respectively, and let $S_1,\dots,S_r$ and $T_1,\dots,T_r$ denote the diagonal blocks of $S$ and $T$, respectively, so that
\begin{align*}
  \Inner{\Phi(S)}{T}
  &= \Inner{\ptr{\cY_1}{S_1}}{T_1} + \sum_{i=2}^r \Inner{\ptr{\cY_i}{S_i} - S_{i-1}\ot I_{\cX_i}}{T_i} \\
  &= \Inner{S_1}{T_1\ot I_{\cY_1}} + \sum_{i=2}^r \Inner{S_i}{T_i\ot I_{\cY_i}} - \Inner{S_{i-1}}{\ptr{\cX_i}{T_i}} \\
  &= \Inner{S_r}{T_r\ot I_{\cY_r}} + \sum_{i=1}^{r-1} \Inner{S_i}{T_i\ot I_{\cY_i}} - \Inner{S_i}{\ptr{\cX_{i+1}}{T_{i+1}}} \\
  &= \Inner{S_r}{T_r\ot I_{\cY_r}} + \sum_{i=1}^{r-1} \Inner{S_i}{T_i\ot I_{\cY_i} - \ptr{\cX_{i+1}}{T_{i+1}}} \\
  &= \Inner{S}{\Psi(T)}.
\end{align*}
As this equality holds for all $S,T$, it follows from the definition of the adjoint that $\Phi^*=\Psi$ as claimed.

We now prove a lemma that establishes the link between the dual problem and the minimization condition of Theorem~\ref{thm:max-prob} (\thmmaxprob).

\begin{lemma}[Correctness of the dual problem]
\label{lm:dual}

For every dual feasible solution $T$ there exists another dual feasible solution $T'$ whose objective value $p$ equals that of $T$ and whose diagonal blocks $T_1',\dots,T_r'$ have the property that $T_i'\ot I_{\cY_i}$ is a non-measuring co-strategy multiplied by $p$ for each $i=1,\dots,r$.

\end{lemma}

\begin{proof}
Let $T_1,\dots,T_r$ denote the diagonal blocks of $T$.
As $T$ is feasible, we know that
\begin{align*}
  T_1\ot I_{\cY_1} &\succeq \ptr{\cX_2}{T_2}\\
  \vdots \\
  T_{r-1}\ot I_{\cY_{r-1}} &\succeq \ptr{\cX_r}{T_r} \\
  T_r\ot I_{\cY_r} &\succeq R_b
\end{align*}
The objective value for $T$ is then given by $p=\inner{B}{T}=\ptr{}{T_1}$.
If $p=0$ then it must hold that $T_1=\cdots=T_r=0$, so the lemma holds trivially in this case.
For the remainder of the proof we shall assume $p>0$.

It is clear that there is a $T_2'\succeq T_2$ such that $T_1\ot I_{\cY_1}=\ptr{\cX_2}{T_2'}$.
Similarly, for each $i=3,\dots,r$ there is a $T_i'\succeq T_i$ such that $T_{i-1}'\ot I_{\cY_{i-1}}=\ptr{\cX_i}{T_i'}$.
As $T_r'\succeq T_r$, it must also be the case that $T_r'\ot I_{\cY_r} \succeq R_b$.
The desired operator $T'$ is then obtained from $T$ by replacing its diagonal blocks with $T_1,T_2',\dots,T_r'$.
That $\inner{B}{T'}=\inner{B}{T}$ follows from the equality $\inner{B}{T'}=\ptr{}{T_1}=p$.
Finally, it follows from Theorem \ref{thm:char} (\thmchar) that each $\frac{1}{p}T_i'\ot I_{\cY_i}$ is a non-measuring co-strategy.
\end{proof}


As any feasible solution $T$ has $\Phi^*(T)\succeq A$, it must be that the $r$th diagonal block $T_r$ of $T$ has the property that $T_r\ot I_{\cY_r}\succeq R_b$.
By Lemma \ref{lm:dual} we may also assume that $T_r\ot I_{\cY_r}$ is a scalar multiple of a co-strategy.
It is clear then that the optimal value of the dual problem equals the infimum over all such scalar multiples, subject to $T_r\ot I_{\cY_r}\succeq R_b$, as required by Theorem~\ref{thm:max-prob}.

\subsubsection{Strong duality of the semidefinite optimization problem}

We already argued that the two quantities appearing in Theorem~\ref{thm:max-prob} are captured by the primal and dual semidefinite optimization problems associated with the triple $(\Phi,A,B)$.
To prove Theorem~\ref{thm:max-prob}, it remains only to show that these two quantities are equal.
This equality is established by showing that $(\Phi,A,B)$ satisfies the conditions for strong duality.

\begin{lemma}[Strong duality of $(\Phi,A,B)$]
\label{lm:duality}

There exists a primal feasible operator $S$ and a dual feasible operator $T$ such that $\inner{A}{S}=\inner{B}{T}$.

\end{lemma}

\begin{proof}

We prove strong duality via item \ref{item:strong-primal} of Fact \ref{fact:strong-duality} (\factstrongduality).
That is, we show that $\beta$ is finite and the primal problem is strictly feasible.
Then by Fact \ref{fact:strong-duality} it follows that $\alpha=\beta$ and that $\beta$ is achieved for some dual feasible operator.
We complete the proof by noting that the optimal value $\alpha$ is also achieved by a primal feasible operator.

First, let us argue that $\beta$ is finite.
As $B\succeq 0$, any dual feasible solution has nonnegative objective value.
Thus, to show that $\beta$ is finite it suffices to exhibit a single dual feasible solution.
Toward that end, let $R$ be a non-measuring co-strategy with $R\succeq R_b$.
Let $T_r$ be such that $T_r\ot I_{\cY_r}=R$ and for each $i=r-1,\dots,1$ choose $T_i$ so that $T_i\ot I_{\cY_i}=\ptr{\cX_{i+1}}{T_{i+1}}$.
(That each of $T_1,\dots,T_r$ exists follows from Theorem \ref{thm:char} (\thmchar).)
Finally, let $T$ be the block-diagonal operator whose diagonal blocks are $T_1,\dots,T_r$.
It is clear that $T$ is dual feasible. 

Next, we show that the primal is strictly feasible.
Choose $\delta\in(0,1/r)$ and let $S$ be the block-diagonal operator whose $i$th diagonal block $S_i$ is given by
\[ S_i = \frac{1-i\delta}{\dim(\kprod{\cY}{1}{i})}I_{\kprod{\cY}{1}{i}\ot\kprod{\cX}{1}{i}}. \]
It is clear that $S\succ 0$ and it is tedious but straightforward to verify that $\Phi(S)\succ B$.
In particular, we have
\begin{align*}
  \ptr{\cY_1}{S_1} = \Pa{1-\delta} I_{\cX_1} &\ \prec \ I_{\cX_1} \\
  \ptr{\cY_2}{S_2} = \frac{1-2\delta}{\dim(\cY_1)}I_{\cY_1\ot\kprod{\cX}{1}{2}}
  &\ \prec \ \frac{1-\delta}{\dim(\cY_1)} I_{\cY_1\ot\kprod{\cX}{1}{2}} = S_1\ot I_{\cX_2} \\
  &\ \vdots \\
  \ptr{\cY_r}{S_r} = \frac{1-r\delta}{\dim(\kprod{\cY}{1}{r-1})} I_{\kprod{\cY}{1}{r-1}\ot\kprod{\cX}{1}{r}}
  &\ \prec \ \frac{1-(r-1)\delta}{\dim(\kprod{\cY}{1}{r-1})} I_{\kprod{\cY}{1}{r-1}\ot\kprod{\cX}{1}{r}} = S_{r-1}\ot I_{\cX_r}
\end{align*}
It now follows from item \ref{item:strong-primal} of Fact \ref{fact:strong-duality} that $\alpha=\beta$ and that $\beta$ is achieved by some dual feasible operator.

It remains only to show that $\alpha$ is also achieved by some primal feasible operator.
By Lemma \ref{lm:primal} it suffices to consider only those primal feasible $S$ whose diagonal blocks are strategies.
As the set of strategies is compact (Proposition \ref{prop:convexity}), it follows that the optimal $\alpha$ is finite and is achieved by a primal feasible solution.
\end{proof}

\subsubsection{Second proof of Theorem~\ref{thm:max-prob} (\thmmaxprob)}

\begin{proof}[Second proof of Theorem~\ref{thm:max-prob}]

We prove the theorem for co-strategies---the result for strategies then follows from Lemma \ref{lm:max-prob-switch}.
Let $p_m\in[0,1]$ denote the maximum probability with which $\set{Q_m}$ can be forced to output $m$ in an interaction with some compatible strategy.
By Lemma \ref{lm:primal}, the optimal value of the primal problem associated with $(\Phi,A,B)$ equals $p_m$.
By Lemma \ref{lm:dual}, the optimal value of the dual problem associated with $(\Phi,A,B)$ equals the minimum $p$ such that $Q_m\preceq pR$ for some non-measuring co-strategy $R$.
That these two values are equal follows from Lemma \ref{lm:duality}.
\end{proof}

\section{Other properties of strategies}
\label{sec:other-properties}

We conclude this chapter with three simple applications of the properties of strategies established in Theorems \ref{theorem:inner-product} and \ref{thm:char}.
In particular, we point out that Theorem \ref{thm:char} immediately implies that the set of all strategies is compact and convex.
We also show that representations of strategies are unique and that they satisfy a convenient distributive property.
The results in this section are proven only for strategies, but it is trivial to repeat each proof for co-strategies.

\subsubsection{The set of strategies is compact and convex}

\begin{proposition}[Convexity of strategies]
\label{prop:convexity}

  The set of all $r$-round strategies is compact and convex, as is the set of all $r$-round co-strategies.

\end{proposition}

\begin{proof}

That the set of strategies is bounded follows from Proposition \ref{prop:explicit-bound} and the fact that every strategy is the Choi-Jamio\l kowski representation of some quantum operation.

That the set of strategies is closed and convex follows from Theorem \ref{thm:char}, which characterizes this set as an intersection between two closed and convex sets---the positive semidefinite operators and those operators satisfying the linear constraints appearing in the theorem.
\end{proof}

\subsubsection{Equivalence and uniqueness of strategies}

In Section \ref{subsec:naive-observations} we noted that operational descriptions of strategies are not unique in the sense that two distinct descriptions could specify equivalent strategies.
For the new representation of strategies, this is not so.

\begin{proposition}[Uniqueness of strategies]
\label{prop:uniqueness}

  Two strategies $Q,Q'$ may be distinguished with nonzero bias by a compatible measuring co-strategy if and only if $Q\neq Q'$.
  A similar statement holds for co-strategies.

\end{proposition}

\begin{proof}


  Let $\set{R_b}$ be any measuring co-strategy.
  By Theorem \ref{theorem:inner-product}, the probability with which $\set{R_b}$ outputs a given outcome $b$ after an interaction with $Q$ or with $Q'$ is given by $\inner{R_b}{Q}$ or $\inner{R_b}{Q'}$, respectively.
  If $Q=Q'$ then it holds that $\inner{X}{Q}=\inner{X}{Q'}$ for every operator $X$ and so $\set{R_b}$ cannot distinguish $Q$ from $Q'$ with nonzero bias.

  Conversely, if $Q\neq Q'$ then for each spanning set $\bB$ of operators there must exist some $B\in\bB$ with $\inner{B}{Q}\neq\inner{B}{Q'}$.
  We claim that the set
  \[ \Set{ R_b : R_b \textrm{ is an element of some measuring co-strategy} } \]
  is a spanning set, from which the proposition follows.

  One way to verify this claim is to note that the identity operator $I$ (suitably normalized) denotes a valid non-measuring co-strategy.
  It follows from Proposition \ref{prop:measure-sum} that for each positive semidefinite operator $R\preceq I$ (again, suitably normalized) there exists a measuring co-strategy of which $R$ is an element.
  The claim then follows from the fact that the set of all operators $0\preceq R\preceq I$ is a spanning set.
\end{proof}

\subsubsection{Distributive property for probabilistic combinations of strategies}

In Example \ref{ex:naive-distributive} it was shown that the operational description of strategies lacks a convenient distributive property for probabilistic combinations of strategies.
This problem is rectified under the new representation for strategies.

\begin{proposition}[Distributive property for strategies]
\label{prop:distributive}

  Let $Q,Q'$ be strategies and let $p\in[0,1]$.
  The strategy that plays according to $Q$ with probability $p$ and according to $Q'$ otherwise is given by the convex combination $pQ+(1-p)Q'$.
  A similar statement holds for co-strategies.

\end{proposition}

\begin{proof}


  Let $Q''$ denote the probabilistic combination of $Q$ and $Q'$ described in the statement of the proposition.
  For any measuring co-strategy $\set{R_b}$ and any outcome $b$ we let $r_b$ denote the probability with which the interaction between $Q$ and $\set{R_b}$ yields the outcome $b$.
  We also define $r_b'$ and $r_b''$ similarly in terms of $Q'$ and $Q''$, respectively.
  By definition, $r_b'' = pr_b+(1-p)r_b'$.

  By Theorem \ref{theorem:inner-product} and the linearity of the inner product, we have
  \[
    \Inner{Q''}{R_b}
    = r_b''
    = pr_b + (1-p)r_b'
    = p\Inner{Q}{R_b} + (1-p)\Inner{Q'}{R_b}
    = \Inner{pQ+(1-p)Q'}{R_b}
  \]
  from which it follows that $\set{R_b}$ cannot distinguish $Q''$ from $pQ+(1-p)Q'$ with nonzero bias.
  That these two operators are equal then follows from Proposition \ref{prop:uniqueness}.
\end{proof}

\chapter{Applications} \label{ch:applications}

In this chapter we present several applications of the formalism for quantum strategies developed in Chapter \ref{ch:properties}:
\begin{itemize}

\item
In Section \ref{sec:game-theory} we develop a general formalism for zero-sum quantum games.
A quantum analogue of von Neumann's Min-Max Theorem for two-player zero-sum games is established, and an efficient algorithm to compute the value of a zero-sum quantum game is presented.

\item
Section \ref{sec:interactive-proofs} contains complexity theoretic applications of quantum strategies, including the collapse of the complexity classes $\cls{QRG}$ and $\cls{EXP}$ and parallel repetition results for single-prover quantum interactive proofs and for quantum interactive proofs with two competing provers.

\item
In Section \ref{sec:coin-flip} we provide a simplified proof of Kitaev's bound for strong quantum coin-flipping.

\end{itemize}
With the exception of the parallel repetition results, each of these applications first appeared in Ref.~\cite{GutoskiW07}.
The parallel repetition result for single-prover quantum interactive proofs is due to Watrous and was not published prior to the present thesis.

\section{Theory of zero-sum quantum games} \label{sec:game-theory}

In this section we develop a general formalism for zero-sum quantum games.
Specifically, we prove a quantum analogue of von Neumann's Min-Max Theorem for two-player zero-sum games in Section \ref{sec:game-theory:min-max} and we provide an efficient algorithm to compute the value of such a game in Section \ref{subsec:game-theory:alg}.

\subsection{Min-max theorem for zero-sum quantum games}
\label{sec:game-theory:min-max}

This section begins with a brief introduction to games and min-max theorems, after which we provide our formalism for quantum games and the quantum min-max theorem.

While it is natural to expect a min-max theorem to hold for two-player zero-sum quantum games, the absence of a mathematically convenient formalism for these games has precluded the appearance of such a theorem in the literature.
(Lee and Johnson established a min-max theorem for the special case of one-round quantum games, wherein only a single round of messages is exchanged between the players and a referee \cite{LeeJ03}.)

\subsubsection{Two-player games and min-max theorems}

A \emph{game} consists of a interaction between two or more players, followed by a \emph{payout} that is awarded to each player at the end of the interaction.
It is the goal of each player to maximize his or her own payout.
A two-player game is \emph{zero-sum} if the sum of the payouts awarded to the players is always zero.
Two-player zero-sum games are always competitive, as the players never have incentive to co-operate.
Many popular examples of games fall into this category, including Poker, Checkers, and Go.
(Win-lose games such as Checkers and Go can be represented as a zero-sum game wherein the only possible payouts are $\pm 1$.)

Let us call the two players in a zero-sum game Alice and Bob.
In the context of these games, a \emph{min-max theorem} is an assertion that every game has a \emph{value} $v$ with the following properties:
\begin{enumerate}
\item[(i)] There exists a strategy for Alice that ensures a payout of at least $v$ regardless of Bob's strategy.
\item[(ii)] There exists a strategy for Bob that ensures a payout of at most $v$ to Alice regardless of her strategy.
\end{enumerate}
In other words, there always exist strategies for the players that are \emph{optimal} in the sense that the players never have incentive to deviate from their optimal strategies.
The original Min-Max Theorem for classical games was established by von Neumann in 1928 \cite{vonNeumann28}.

In an analytical context, min-max theorems are statements about sets $A,B$ and functions $f:A\times B\to\mathbb{R}$.
While it must always hold that
\[ \sup_{a\in A} \inf_{b\in B} f(a,b) \leq \inf_{b\in B} \sup_{a\in A} f(a,b), \]
a \emph{min-max theorem} provides conditions upon $A,B,f$ under which these two quantities are equal.
For example, the following well-known min-max theorem will prove useful for our purpose.

\def\factminmax{Convex-bilinear min-max theorem}
\begin{fact}[\factminmax]
\label{fact:min-max}

  If $A,B$ are compact convex subsets of finite-dimensional real vector spaces and $f:A\times B\to\mathbb{R}$ is bilinear then
  \[ \max_{a\in A} \min_{b\in B} f(a,b) = \min_{b\in B} \max_{a\in A} f(a,b). \]

\end{fact}

While Fact \ref{fact:min-max} does not follow immediately from von Neumann's original Min-Max Theorem, it does follow from an early generalization due to Ville \cite{Ville38}.
Additional citations and an English-language proof can be found in Rockafellar \cite{Rockafellar70}.

Analytical min-max theorems such as Fact \ref{fact:min-max} can be used to establish a min-max theorem for a given game or class of games.
To do so, it suffices that $A,B$ represent the sets of all possible strategies for Alice and Bob and that $f(a,b)$ denotes the payout to Alice when she employs strategy $a$ and Bob employs strategy $b$.


\subsubsection{Formal definition of a zero-sum quantum game}

Classical two-player games have several distinct mathematical formalizations, each with its own advantages.
%
One of these formalisms---the \emph{refereed game}---lends itself particularly well to generalization to quantum games.
This formalism encapsulates the rules of a particular game into a \emph{referee}, who exchanges messages with each of the two players.
The referee enforces the rules of the game and decides when to terminate the interaction and award a payout to the players.
In a \emph{quantum} refereed game the players and referee may exchange and process quantum information and the payout is determined by a measurement made by the referee at the end of the interaction.

\def\defquantumgame{Quantum game}
\begin{definition}[\defquantumgame] \label{def:quantum-game}
  An \emph{$r$-round referee} is an $r$-round measuring co-strategy $\set{R_m}_{m\in\Sigma}$ whose input spaces $\cX_1,\dots,\cX_r$ and output spaces $\cY_1,\dots,\cY_r$ take the form
  \[ \cX_i = \cA_i\ot\cB_i \quad \textrm{and} \quad \cY_i = \cC_i\ot\cD_i \]
  for complex Euclidean spaces $\cA_i$, $\cB_i$, $\cC_i$, and $\cD_i$ for $1\leq i\leq r$.
  An \emph{$r$-round quantum game} consists of an $r$-round referee along with \emph{payout functions}
  \[ V_A,V_B : \Sigma \to \mathbb{R} \]
  defined on the referee's set $\Sigma$ of measurement outcomes.
  For each such outcome $m\in\Sigma$, Alice's payout is $V_A(m)$ and Bob's payout is $V_B(m)$.
  Such a game is \emph{zero-sum} if $V_A(m)+V_B(m)=0$ for all $m\in\Sigma$.
  
  During each round, the referee simultaneously sends a message to Alice and a message to Bob, and a response is expected from each player.
  The spaces $\cA_i$ and $\cB_i$ correspond to the messages sent by the referee during the $i$th round, while $\cC_i$ and $\cD_i$ correspond to their responses.
  After $r$ rounds, the referee produces an output $m\in\Sigma$ and awards the payouts $V_A(m)$ to Alice and $V_B(m)$ to Bob.
\end{definition}

Notice that Definition \ref{def:quantum-game} places no restrictions on the strategies available to the players.
For example, the players might employ a strategy that allows them to exchange quantum information directly, as opposed to an indirect exchange via messages to the referee.
Alternatively, they might share entanglement or randomness but be forbidden from direct communication, or they might even be forbidden from sharing entanglement or randomness altogether.

\subsubsection{Min-max theorem for zero-sum quantum games}

Let
\[ \strategy{A} \subset \pos{\kprod{\cC}{1}{r}\ot\kprod{\cA}{1}{r}} \]
denote the set of all $r$-round non-measuring strategies for Alice's input spaces $\cA_1,\dots,\cA_r$ and output spaces $\cC_1,\dots,\cC_r$.
Similarly, let
\[ \strategy{B} \subset \pos{\kprod{\cD}{1}{r}\ot\kprod{\cB}{1}{r}} \]
denote the set of all $r$-round non-measuring strategies for Bob's input spaces $\cB_1,\dots,\cB_r$ and output spaces $\cD_1,\dots,\cD_r$.

In a zero-sum quantum game it cannot simultaneously be to both players' advantage to communicate directly with each other or to share a source of randomness or entanglement.
Thus, we may assume that Alice and Bob play independent strategies represented by
$A\in\strategy{A}$ and $B\in\strategy{B}$, respectively.
In particular, their combined $r$-round strategy for the referee's input spaces $\cX_1,\dots,\cX_r$ and output spaces $\cY_1,\dots,\cY_r$ is described by the tensor product $A\ot B$.

For any zero-sum quantum game with referee $\set{R_m}$ and payout functions $V_A,V_B$ we write
\[ V(m) \defeq V_A(m) = -V_B(m) \]
and define the Hermitian operator
\[ R \defeq \sum_{m\in\Sigma} V(m)R_m. \]
By Theorem \ref{theorem:inner-product} (\theoreminnerproduct) Alice's expected payout for this game is given by
\[ \sum_{m\in\Sigma} V(m)\inner{A\ot B}{R_m}=\inner{A\ot B}{R} \]
while Bob's expected payout is $-\inner{A\ot B}{R}$.
Because the inner product $\inner{A\ot B}{R}$ is a bilinear function of $A$ and $B$ and because the sets $\strategy{A},\strategy{B}$ are compact and convex (as noted in Proposition \ref{prop:convexity}), we may employ Fact \ref{fact:min-max} (\factminmax) to obtain
\[ \max_{A\in \strategy{A}} \min_{B\in \strategy{B}} \inner{A\ot B}{R} = \min_{B\in \strategy{B}} \max_{A\in \strategy{A}} \inner{A\ot B}{R}. \]
The real number represented by the two sides of this equation is called the \emph{value} of the game.
Any strategy $A\in\strategy{A}$ achieving the maximum of the left side of this equality is an optimal strategy for Alice, while any strategy $B\in\strategy{B}$ achieving the minimum of the right side of this equality is an optimal strategy for Bob.
The following theorem is now proved.

\def\thmminmax{Min-max theorem for zero-sum quantum games}
\begin{theorem}[\thmminmax]
\label{thm:min-max}

  Every two-player zero-sum $r$-round quantum game has a \emph{value} $v$ with the following properties:
  \begin{enumerate}
  
  \item[(i)]
    There exists a strategy for Alice that ensures a payout of at least $v$ regardless of Bob's strategy.
    
  \item[(ii)]
    There exists a strategy for Bob that ensures a payout of at most $v$ to Alice regardless of her strategy.
    
  \end{enumerate}

\end{theorem}

\subsection{Efficient algorithm to compute the value} 
\label{subsec:game-theory:alg}

Theorem \ref{thm:min-max} (\thmminmax) asserts that each zero-sum quantum game has a value.
But can this value be easily computed?
In this section we answer this question in the affirmative by exhibiting an efficient deterministic algorithm that takes as input a description of a referee $\set{R_m}$
and produces as output the value of the zero-sum quantum game defined by $\set{R_m}$.

This goal is achieved by expressing the value of the game
as a semidefinite optimization problem and then employing the existence of efficient algorithms for semidefinite optimization.
As these algorithms apply only to certain semidefinite optimization problems whose feasible sets are ``well-bounded,'' it is necessary to establish the well-boundedness of our semidefinite optimization problem for zero-sum quantum games.

\subsubsection{An optimization problem for the value of a game}


Let $\set{R_m}_{m\in\Sigma}$ be an $r$-round referee and consider the following super-operator defined for each measurement outcome $m\in\Sigma$:
\begin{align*}
  \Omega_m
  &: \lin{\kprod{\cD}{1}{r}\ot\kprod{\cB}{1}{r}}\to\lin{\kprod{\cC}{1}{r}\ot\kprod{\cA}{1}{r}} \\
  &: B \mapsto \Ptr{ \kprod{\cD}{1}{r}\ot\kprod{\cB}{1}{r}
  }{
    \Pa{ B\ot I_{\kprod{\cC}{1}{r}\ot\kprod{\cA}{1}{r}} } R_m
  }.
\end{align*}
These super-operators have the property that $\inner{A\ot B}{R_m} = \inner{A}{\Omega_m(B)}$ for each measurement outcome $m\in\Sigma$ and each choice of strategies $A$ for Alice and $B$ for Bob.
Indeed, the set $\set{\Omega_m(B)}_{m\in\Sigma}$ is the $r$-round measuring co-strategy for Alice's input spaces $\cA_1,\dots,\cA_r$ and output spaces $\cC_1,\dots,\cC_r$ obtained by ``hard-wiring'' Bob's strategy $B$ into the referee.

Let $V(m)$ denote the payout to Alice in a zero-sum quantum game with referee $\set{R_m}$.
For the moment, it is convenient to restrict our attention to payout functions with $0\leq V(m)\leq 1$ for all $m$.
But we shall soon see that the ensuing discussion is easily generalized to arbitrary payout functions.

Borrowing from the previous subsection, we define the super-operator \[\Omega_R\defeq\sum_{m\in\Sigma}V(m)\Omega_m\] so that Alice's expected payout is given by \(\inner{A}{\Omega_R(B)}\) when Alice and Bob play according to the strategies $A$ and $B$, respectively.
Naturally, Alice's maximum expected payout when Bob plays according to $B$ is
\[ \max_{A\in\strategy{A}} \inner{A}{\Omega_R(B)}. \]
Let $\costrategy{A}$ denote the set of co-strategies for Alice's input and output spaces.
Because $0\leq V(m) \leq 1$ for all $m$, it follows that, regardless of the choice of $B\in\bB$, the operator $\Omega_R(B)$ is always an element of some measuring co-strategy.
Hence, we may apply Theorem \ref{thm:max-prob} (\thmmaxprob)
to obtain the following alternate expression for Alice's maximum expected payout:
\[ \min \Set{\lambda\geq 0 : \Omega_R(B)\preceq\lambda Q \textrm{ for some } Q\in\costrategy{A} }. \]
Thus, the value of the game is given by
\[
  \min_{B\in\strategy{B}} \max_{A\in \strategy{A}} \inner{A}{\Omega_R(B)}
  = \min_{B\in\strategy{B}} \min \Set{\lambda\geq 0 : \Omega_R(B)\preceq\lambda Q \textrm{ for some } Q\in\costrategy{A} }, 
\]
which can equivalently be written as an optimization problem:
\begin{align}
  \textrm{minimize} \quad & \lambda \nonumber \\
  \textrm{subject to} \quad & \Omega_R(B)\preceq\lambda Q
    \label{eq:min-max-opt}
    \\
  & B\in\strategy{B} \nonumber \\ & Q\in\costrategy{A}. \nonumber
\end{align}
Thus, an efficient solution for this optimization problem yields an efficient algorithm to compute the value of a zero-sum quantum game, provided that the payouts all fall within the interval $[0,1]$.
To compute the value of a game with an \emph{arbitrary} payout function $V$ we simply shift and scale $V$ to a payout function $V'$ with payouts in $[0,1]$ by defining 
\[ V'(m) \defeq \frac{V(m) + \Abs{ \min_n \set{V(n)} } }{ \Abs{\max_n \set{V(n)}} + \Abs{ \min_n \set{V(n)} }} \]
and then solve the optimization problem associated with $V'$, from which the original value is easily recovered.

\subsubsection{A semidefinite optimization problem for the value of a game}

Let us argue that the optimization problem \eqref{eq:min-max-opt} can be expressed as a semidefinite optimization problem in the super-operator form described in Section \ref{subsec:semidefinite-duality}.
Specifically, we construct a triple $(\Phi,E,F)$ with the property that the optimal value of \eqref{eq:min-max-opt} equals the optimal value of the dual problem
\begin{align*}
\textrm{minimize} \quad & \inner{F}{X} \\
\textrm{subject to} \quad & \Phi(X) \succeq E\\
& X \succeq 0.
\end{align*}

To this end,
we choose the super-operator $\Phi$ so that the semidefinite variable $X$ can be assumed without loss of generality to be block-diagonal with the form
\[
  X = \Pa{
  \begin{array}{cc}
    B \\
    & Q
  \end{array} },
  \qquad
  B = \Pa{
  \begin{array}{ccc}
    B_1 \\
    & \ddots\\
    && B_r
  \end{array} },
  \qquad
  Q = \Pa{
  \begin{array}{ccc}
    Q_1\\
    & \ddots\\
    && Q_r
  \end{array} }
\]
It is convenient to decompose $\Phi$, $E$, and $F$ into a hierarchy of block-diagonal operators as follows:
\[
E = \Pa{
  \begin{array}{ccc}
    E_\textrm{Bob}\\
    & E_\textrm{co-Alice}\\
    && 0
  \end{array} },
\qquad
F = \Pa{
  \begin{array}{cc}
    F_\textrm{Bob}\\
    & F_\textrm{co-Alice}
  \end{array} },
\]
and
\begin{align*}
  \Phi :&
  \Pa{
  \begin{array}{cc}
    B \\
    & Q
  \end{array} } 
  \mapsto 
  \Pa{
  \begin{array}{cccc}
    \Phi_\textrm{Bob}(B) && 0 \\
    & \Phi_\textrm{co-Alice}(Q) \\
    0 && Q_r\ot I_{\cC_r} - \Omega_R(B_r)
  \end{array} }.
\end{align*}
It is clear that the condition $\Phi(X)\succeq E$ is equivalent to
\begin{align*}
\Phi_\textrm{Bob}(B) &\succeq E_\textrm{Bob},\\
\Phi_\textrm{co-Alice}(Q) &\succeq E_\textrm{co-Alice},\\
Q_r\ot I_{\cC_r} &\succeq \Omega_R(B_r).
\end{align*}
The idea is that we will choose $\Phi_\textrm{Bob},\Phi_\textrm{co-Alice}$ so that the first two constraints ensure that $B_r$ is a valid $r$-round strategy for Bob and that $Q_r\ot I_{\cC_r}$ is (a scalar multiple of) a valid $r$-round co-strategy for Alice's input and output spaces.
The third constraint clearly captures the
semidefinite inequality condition of \eqref{eq:min-max-opt}.


%

Let us fill in the details for $\Phi$, $E$, and $F$.
Define
\begin{align*}
  \Phi_\textrm{Bob} :&
\ B
  \mapsto 
  \Pa{
  \begin{array}{cccc}
    \ptr{\cD_1}{B_1} &&& 0 \\
    & \ptr{\cD_2}{B_2} - B_1\ot I_{\cB_2} \\
    && \ddots\\
    0 &&& \ptr{\cD_r}{B_r} - B_{r-1}\ot I_{\cB_r}
  \end{array} },
  \\[2mm]
  \Phi_\textrm{co-Alice} :&
\ Q
    \mapsto 
    \Pa{
    \begin{array}{ccc}
      Q_1\ot I_{\cC_1} - \ptr{\cA_2}{Q_2} && 0 \\
      & \ddots\\
      0 && Q_{r-1}\ot I_{\cC_{r-1}} - \ptr{\cA_r}{Q_r}
  \end{array} },
\end{align*}
and
\[
  E_\textrm{Bob} = \Pa{
  \begin{array}{cccc}
    I_{\cB_1} &&& 0 \\ & 0 \\ && \ddots \\ 0 &&& 0
  \end{array} },
  \qquad
  E_\textrm{co-Alice} = 0.
\]
For the objective function, we select
\[
  F_\textrm{Bob} = 0,
  \qquad
  F_\textrm{co-Alice} = \Pa{
  \begin{array}{cccc}
    I_{\cA_1} &&& 0 \\ & 0 \\ && \ddots \\ 0 &&& 0
  \end{array} }
\]
so that \[\inner{F}{X}=\inner{F_\textrm{co-Alice}}{Q}=\inner{I_{\cA_1}}{Q_1}=\ptr{}{Q_1}.\]

\subsubsection{Correctness of the semidefinite optimization problem}

Let us verify that the optimal value of the dual problem $(\Phi,E,F)$ equals the optimal value of \eqref{eq:min-max-opt}.
To the extent that $(\Phi,E,F)$ resembles the semidefinite optimization problem appearing in
Section \ref{subsec:semidefinite-duality}, 
the material in that section---particularly the proofs of Lemmas \ref{lm:primal} and \ref{lm:dual}---can be reused in the present setting.

\begin{lemma}[Correctness of $(\Phi,E,F)$]
\label{lm:sdp-min-max} 


For each dual feasible $X$ for $(\Phi,E,F)$
there is another dual feasible $X'$ whose objective value $p$ equals that of $X$ and
whose diagonal blocks $B_1',\dots,B_r'$, $Q_1',\dots,Q_r'$ have the property that
\begin{enumerate}
\item \label{item:sdp-Bob} $B_i'$ is an $i$-round non-measuring strategy, and 
\item \label{item:sdp-co-Alice}$Q_i'\ot I_{\cC_i}$ is an $i$-round non-measuring co-strategy multiplied by $p$
\end{enumerate}
for each $i=1,\dots,r$.

\end{lemma}


\begin{proof}
Let $B_1,\dots,B_r$, $Q_1,\dots,Q_r$ denote the diagonal blocks of $X$.
First, we choose $B_1',\dots,B_r'$ and show that they satisfy item \ref{item:sdp-Bob} in the statement of the lemma.
The proof is nearly identical to that of Lemma \ref{lm:primal}.
As $X$ is feasible, we know that $\Phi_\textrm{Bob}(B)\succeq E_\textrm{Bob}$ and therefore
\begin{align*}
  \ptr{\cD_1}{B_1} &\succeq I_{\cB_1} \\
  \ptr{\cD_2}{B_2} &\succeq B_1\ot I_{\cB_2} \\
  \vdots \\
  \ptr{\cD_r}{B_r} &\succeq B_{r-1}\ot I_{\cB_r}.
\end{align*}
It is clear that there is a $B_1'\preceq B_1$ such that $\ptr{\cD_1}{B_1'}=I_{\cB_1}$.
Similarly, for each $i=2,\dots,r$ there is a $B_i'\preceq B_i$ such that $\ptr{\cD_i}{B_i'}=B_{i-1}'\ot I_{\cB_i}$.
That $B_1',\dots,B_r'$ are non-measuring strategies follows from Theorem \ref{thm:char} (\thmchar).

Next, we choose $Q_1',\dots,Q_r'$ and show that they satisfy item \ref{item:sdp-co-Alice} in the statement of the proposition.
The proof is nearly identical to that of Lemma \ref{lm:dual}.
As $X$ is feasible, we know that $\Phi_\textrm{co-Alice}(Q)\succeq E_\textrm{co-Alice}$ and therefore
\begin{align*}
  Q_1\ot I_{\cC_1} &\succeq \ptr{\cA_2}{Q_2}\\
  \vdots \\
  Q_{r-1}\ot I_{\cC_{r-1}} &\succeq \ptr{\cA_r}{Q_r}.
\end{align*}
As $X$ has objective value $p$, it must be that $p=\ptr{}{Q_1}$.
If $p=0$ then it must hold that $Q_1=\cdots=Q_r=0$,
so item \ref{item:sdp-co-Alice} holds trivially for the choice $Q_1'=\cdots=Q_r'=0$.

Assume then that $p>0$.
It is clear that there is a $Q_2'\succeq Q_2$ such that $Q_1\ot I_{\cC_1}=\ptr{\cA_2}{Q_2'}$.
Similarly, for each $i=3,\dots,r$ there is a $Q_i'\succeq Q_i$ such that $Q_{i-1}'\ot I_{\cC_{i-1}}=\ptr{\cA_i}{Q_i'}$.
Taking $Q_1'=Q_1$ and using the fact that $p=\ptr{}{Q_1}$, it follows from Theorem \ref{thm:char} (\thmchar) that each $\frac{1}{p}Q_i'\ot I_{\cC_i}$ is a non-measuring co-strategy.

At this point, we have chosen the diagonal blocks $B_1',\dots,B_r'$, $Q_1',\dots,Q_r'$ of $X'$.
That $X$ and $X'$ have the same objective value follows immediately from the choice $Q_1'=Q_1$.
It remains only to verify that $Q_r'\ot I_{\cC_r} \succeq \Omega_R(B_r')$.
But this inequality follows immediately from the facts that $Q_r'\succeq Q_r$, that $B_r'\preceq B_r$, and that the super-operator $\Omega_R$ is completely positive.
\end{proof}


\subsubsection{Computational efficiency and well-boundedness}

Now that we have expressed the value of a zero-sum quantum game as a semidefinite optimization problem, it remains only to argue that an existing algorithm for semidefinite optimization can be applied to our specific problem $(\Phi,E,F)$.

The algorithm we employ is the ellipsoid method \cite{Khachiyan79,GrotschelL+88}, which approximates the optimal value of a semidefinite optimization problem to arbitrary precision in polynomial-time, provided that the set of feasible solutions is \emph{well-bounded}---that is, it contains a ball of radius $\delta$ and is, in turn, contained in a ball of radius $\Delta$ such that the ratio $\Delta/\delta$ is not too large.
More precisely, we have the following.

\def\factsdpalg{Efficient algorithm for semidefinite optimization}
\begin{fact}[\factsdpalg]
\label{fact:sdp-alg}

The following promise problem admits a deterministic polynomial-time solution:
\begin{description}

\item[Input.]
A semidefinite optimization problem $(\Phi,A,B)$ and positive real numbers $\varepsilon,\delta,\Delta$.
The numbers $\varepsilon$, $\delta$, and $\Delta$ are given explicitly in binary, as are the real and complex parts of each entry of $\jam{\Phi}$, $A$, and $B$.

\item[Promise.]
There exists a primal [dual] feasible solution $X_0$ such that
for all Hermitian operators $H$ with $\fnorm{H}\leq\delta$ it holds that $X_0+H$
is primal [dual] feasible.
Moreover,
for all primal [dual] feasible solutions $X$ it holds that $\fnorm{X}\leq\Delta$.

\item[Output.]
A real number $\gamma$ such that $\abs{\gamma-\alpha}<\varepsilon$, where $\alpha$ is the optimal value of the primal [dual] problem associated with $(\Phi,A,B)$.

\end{description}

\end{fact}

In addition to the ellipsoid method, interior point methods are also used for semidefinite optimization \cite{deKlerk02,BoydV04}.
The reader is referred to Watrous \cite{Watrous09} and the references therein for more detailed discussion of algorithms for semidefinite optimization problems written in super-operator form.

Unfortunately, the set of dual feasible solutions associated with our problem $(\Phi,E,F)$ is unbounded---if $X$ is dual feasible then so is $\lambda X$ for all $\lambda\geq 1$.
To remedy this problem, it suffices to augment the original problem $(\Phi,E,F)$ with the additional constraint $\ptr{}{X}\leq t$ on dual feasible solutions $X$ for some appropriately large choice of $t$.
Such a constraint can be incorporated into the super-operator form by defining
\begin{align*}
\Phi' &: X\mapsto
  \left(
  \begin{array}{cc}
  \Phi(X)\\
  & -\ptr{}{X}
  \end{array}
  \right), \\
E' &=
  \left(
  \begin{array}{cc}
  E\\
  & -t
  \end{array}
  \right).
\end{align*}
If $t$ is large enough so that an optimal solution $X^\star$ for the dual problem $(\Phi,E,F)$ has $\ptr{}{X^\star}\leq t$ then $X^\star$ is also optimal for $(\Phi',E',F)$.
In this case, we can obtain the optimal value for $(\Phi,E,F)$ by solving $(\Phi',E',F)$.
Thus, it suffices to prove well-boundedness for $(\Phi',E',F)$.

\def\lmwellbounded{Well-boundedness of $(\Phi',E',F)$}
\begin{lemma}[\lmwellbounded]
\label{lm:well-bounded}

There exists a dual feasible $X_0$ for $(\Phi,E,F)$ and positive real numbers $\delta,t$ such that:
\begin{enumerate}
\item 
For all Hermitian operators $H$ with $\fnorm{H}\leq\delta$ it holds that $X_0+H$ is dual feasible with $\ptr{}{X_0+H}\leq t$.
\item 
The optimal value of the dual problem $(\Phi,E,F)$ is achieved by a dual feasible $X^\star$ with $\ptr{}{X^\star}\leq t$.
\end{enumerate}
In particular, we may select
\begin{align*}
  \delta &= \Omega\Pa{\dim(\kprod{\cD}{1}{r})^{-1} },\\
  t &= O\Pa{r^3 \dim\Pa{ \kprod{\cY}{1}{r}\ot\kprod{\cX}{1}{r}\ot\kprod{\cC}{1}{r} } }.
\end{align*}

\end{lemma}

\begin{proof}
First, we select $X_0$ and $\delta$; the quantity $t$ will be chosen at the end of the proof.
We begin by choosing the first collection $B_1,\dots,B_r$ of diagonal blocks for $X_0$.
For each $i=1,\dots,r$ let
\[ B_i = \frac{i+1}{\dim(\kprod{\cD}{1}{i})} I_{\kprod{\cD}{1}{i}\ot\kprod{\cB}{1}{i}} \]
and let $Y_i$ be any Hermitian operator with
\[ \norm{Y_i} \leq \frac{1}{3\dim(\kprod{\cD}{1}{i})} \]
so that
\[
  \frac{i+\frac{2}{3}}{\dim(\kprod{\cD}{1}{i})} I_{\kprod{\cD}{1}{i}\ot\kprod{\cB}{1}{i}}
  \ \preceq \ B_i+Y_i \ \preceq \ 
  \frac{i+\frac{4}{3}}{\dim(\kprod{\cD}{1}{i})} I_{\kprod{\cD}{1}{i}\ot\kprod{\cB}{1}{i}}.
\]
As in the proof of Lemma \ref{lm:duality}, it is tedious but straightforward to verify that
\begin{align*}
  \ptr{\cD_1}{B_1+Y_1} &\ \succ \ I_{\cB_1} \\
  \ptr{\cD_2}{B_2+Y_2} &\ \succ \ (B_1+Y_1)\ot I_{\cB_2} \\
  &\ \vdots \\
  \ptr{\cD_r}{B_r+Y_r} &\ \succ \ (B_{r-1}+Y_{r-1})\ot I_{\cB_r}.
\end{align*}

Next, we choose the remaining diagonal blocks $Q_1,\dots,Q_r$ for $X_0$.
To this end, let $\gamma$ be a real number large enough to guarantee that \[ \Omega_R(B_r+Y_r)\preceq\gamma I_{\kprod{\cC}{1}{r}\ot\kprod{\cA}{1}{r}}, \]
regardless of the choice of $Y_r$.
(The precise value of $\gamma$ will be chosen later.)
For each $i=1,\dots,r$ let
\[ Q_i = \Pa{r-i+2} \dim(\kprod{\cA}{i+1}{r}) \gamma I_{\kprod{\cC}{1}{i-1}\ot\kprod{\cA}{1}{i}} \]
and let $Z_i$ be any Hermitian operator with
\[ \norm{Z_i} \leq \frac{\dim(\kprod{\cA}{i+1}{r})\gamma}{3} \]
so that
\begin{align*}
  (Q_r+Z_r)\ot I_{\cC_r} &\ \succ \ \Omega_R(B_r+Y_r) \\
  (Q_{r-1}+Z_{r-1})\ot I_{\cC_{r-1}} &\ \succ \ \ptr{\cA_r}{Q_r+Z_r} \\
  &\ \vdots \\
  (Q_1+Z_1)\ot I_{\cC_1} &\ \succ \ \ptr{\cA_2}{Q_2+Z_2}.
\end{align*}
It follows from these semidefinite inequalities that the block-diagonal operator $X_0$ with diagonal blocks $B_1,\dots,B_r$, $Q_1,\dots,Q_r$ is a strictly dual feasible solution to $(\Phi,E,F)$.
Moreover, for any Hermitian operator $H$ with diagonal blocks $Y_1,\dots,Y_r$, $Z_1,\dots,Z_r$ and $\fnorm{H}\leq\delta$ it must be that $X_0+H$ is also a strictly feasible solution, provided
\[ \delta \leq \frac{1}{3\dim(\kprod{\cD}{1}{r})}. \]

Next, let us choose an appropriate value for $\gamma$, which will enable us to establish
the desired upper bound $t$.
For notational convenience, write \[ \nu = \frac{r+\frac{4}{3}}{\dim(\kprod{\cD}{1}{r})} \] so that \[ B_r+Y_r\preceq \nu I_{\kprod{\cD}{1}{r}\ot\kprod{\cB}{1}{r}}.\]
As $\Omega_R$ is completely positive, it follows that
\[
  \Omega_R(B_r+Y_r)
  \preceq \nu \Omega_R( I_{\kprod{\cD}{1}{r}\ot\kprod{\cB}{1}{r}} )
  = \nu \Ptr{ \kprod{\cD}{1}{r}\ot\kprod{\cB}{1}{r} }{ R }
\]
where $R$ is an operator with $0\preceq R\preceq\jam{\Psi^*}$ for some completely positive and trace-preserving super-operator $\Psi:\lin{\kprod{\cY}{1}{r}}\to\lin{\kprod{\cX}{1}{r}}$.
Proposition \ref{prop:explicit-bound} tells us that
$\tnorm{\jam{\Psi^*}}=\dim(\kprod{\cY}{1}{r})$,
from which we obtain $\norm{R}\leq\dim(\kprod{\cY}{1}{r})$ and hence
\[ R\preceq \dim(\kprod{\cY}{1}{r}) I_{\kprod{\cY}{1}{r}\ot\kprod{\cX}{1}{r}}. \]
It follows that
\[
  \Omega_R(B_r+Y_r) \preceq
  \nu \dim\Pa{\kprod{\cY}{1}{r}\ot\kprod{\cD}{1}{r}\ot\kprod{\cB}{1}{r}} I_{\kprod{\cC}{1}{r}\ot\kprod{\cA}{1}{r}}
\]
and hence we may select
\[
  \gamma =
  \Pa{r+{\textstyle\frac{4}{3}}}\dim\Pa{\kprod{\cY}{1}{r}\ot\kprod{\cB}{1}{r}}.
\]
%
It is straightforward but tedious to verify that for every Hermitian $H$ with $\fnorm{H}\leq\delta$ it holds that
\[
  \ptr{}{X_0+H} < r\Pa{r+{\textstyle\frac{4}{3}}} \Pa{\dim(\kprod{\cB}{1}{r}) + \gamma\dim(\kprod{\cA}{1}{r}\ot\kprod{\cC}{1}{r})}.
\]
Substituting our choice of $\gamma$, we find that it suffices to select
\[ t = 2(r+2)^3 \dim\Pa{ \kprod{\cY}{1}{r}\ot\kprod{\cX}{1}{r}\ot\kprod{\cC}{1}{r} }. \]
To see that this choice of $t$ also bounds the trace of an optimal solution, we note that by Lemma \ref{lm:sdp-min-max} the diagonal blocks $B^\star_1,\dots,B^\star_r$, $Q^\star_1,\dots,Q^\star_r$ of any optimal solution $X^\star$ can be assumed to be strategies and co-strategies, respectively, from which it is easy to see that
\[ \ptr{}{X^\star}< r \Pa{\dim(\kprod{\cB}{1}{r})+\dim(\kprod{\cC}{1}{r})}<t. \]
\end{proof}

Lemma \ref{lm:well-bounded} (\lmwellbounded) provides all that we need in order to apply Fact \ref{fact:sdp-alg} (\factsdpalg).
While the bound $t$ on the set of feasible solutions to $(\Phi',E',F)$ from Lemma \ref{lm:well-bounded} is stated in terms of the trace norm, it is straightforward to convert this quantity into a bound $\Delta$ in terms of the Frobenius norm via the norm inequalities listed in Section \ref{sec:intro:linalg}.
The following theorem is now proved.

\def\thmcomputevalue{Efficient algorithm to compute the value of a quantum game}
\begin{theorem}[\thmcomputevalue]
\label{thm:compute-value}


The following problem admits a deterministic polynomial-time solution:
\begin{description}

\item[Input.] A two-player zero-sum quantum game specified by an $r$-round quantum referee $\set{R_m}$, a payout function $V(m)$, and an accuracy parameter $\varepsilon> 0$.
The real numbers $V(m)$ and $\varepsilon$ are each given explicitly in binary, as are the real and complex parts of each entry of the matrices $R_m$.

\item[Output.]
A real number $v$ such that the value of the game specified by $\set{R_m}$ and $V(m)$ lies in the open interval $(v-\varepsilon,v+\varepsilon)$.

\end{description}

\end{theorem}


\section{Quantum interactive proofs} \label{sec:interactive-proofs}

In this section we provide an application of the theory of zero-sum quantum games developed in Section \ref{sec:game-theory} to quantum interactive proofs.
In particular, we observe that the existence of an efficient algorithm that computes the value of a zero-sum quantum game implies the equivalence of the complexity classes $\cls{QRG}$ and $\cls{EXP}$.
Here $\cls{QRG}$ is the class of problems that admit quantum interactive proofs with two competing provers, whereas $\cls{EXP}$ is the fundamental class of problems that admit deterministic (classical) exponential-time solutions.

We also extend an existing parallel repetition result for single-prover quantum interactive proofs so that it applies to interactions with an arbitrary number of messages, as opposed to interactions with only three messages.
The technique employed toward this end is then applied to yield a similar parallel repetition result for quantum interactive proofs with competing provers.

This section begins with a primer on interactive proofs, followed by a brief historical survey before the results are presented.

\subsubsection{Interactive proofs}


An \emph{interactive proof} consists of an interaction between a \emph{verifier} and a \emph{prover} regarding some common input $x$, which is viewed as an instance of a yes-no decision problem $P$.
Throughout the interaction the verifier is restricted to randomized polynomial-time computations, while the prover's computational power is unlimited.
After the interaction the verifier produces a binary outcome indicating acceptance or rejection of $x$.
The verifier's goal is to accept those inputs that are yes-instances of $P$ and reject those inputs that are no-instances of $P$.
Typically, the verifier does not have the computational power to make this determination himself and so he must look to the prover and his unlimited computational power for help.
However, the goal of the prover is always to convince the verifier to accept, regardless of whether the input $x$ is actually a yes-instance.
Hence, the verifier must be careful to distinguish truthful proofs from false proofs.

A decision problem $P$ is said to \emph{admit an interactive proof} if there exists a verifier with the property that a prover can convince him to accept every yes-instance of $P$ with high probability, yet no prover can convince him to accept any no-instance of $P$ except with small probability.
Such an interactive proof is said to \emph{solve} the problem $P$.
The complexity class of decision problems that admit interactive proofs is denoted $\cls{IP}$.
In a \emph{quantum} interactive proof, the verifier and prover may process and exchange quantum information.
Whereas a classical verifier is restricted to randomized polynomial-time computations, a quantum verifier is instead restricted to quantum circuits that can be generated uniformly in deterministic polynomial time.
The corresponding complexity class is denoted $\cls{QIP}$.

An interactive proof \emph{with competing provers} has two provers---one, as before, whose goal is always to convince the verifier to accept, and another, whose goal is to convince the verifier to reject.
The class of problems that admit interactive proofs with competing provers is denoted $\cls{RG}$, for ``refereed games'', owing to the similarity between these two models of interaction.
As above, we may speak of \emph{quantum} interactive proofs with competing provers and the associated complexity class $\cls{QRG}$.

\subsubsection{Background}

Interactive proofs were introduced \cite{Babai85,GoldwasserM+89} as a generalization of efficiently verifiable proofs---a fruitful concept that characterizes the ubiquitous complexity class $\cls{NP}$.
The generalization lies in the interaction:
whereas any problem in $\cls{NP}$ may be verified with a single message from the prover to the verifier, interactive proofs allow the verifier to ask a series of questions of the prover, possibly basing future questions upon previous answers.
In order to make this generalization from $\cls{NP}$ to $\cls{IP}$ nontrivial, the verifier is permitted a source of randomness.
(Hence the allowance in interactive proofs for a small probability of error.)

Naturally, every problem in $\cls{NP}$ admits an interactive proof.
Moreover, it is conjectured that interactive proofs with only a constant number of messages cannot solve problems outside $\cls{NP}$ \cite{MiltersenV99}.
On the other hand, interactive proofs with an unbounded number of messages are surprisingly powerful: every problem in $\cls{PSPACE}$ admits such an interactive proof \cite{LundF+92,Shamir92}.
By contrast, in any \emph{quantum} interactive proof it suffices that the verifier and prover exchange only three messages \cite{KitaevW00}.
It is not difficult to see that $\cls{PSPACE}$ contains $\cls{IP}$, from which the characterization $\cls{IP}=\cls{PSPACE}$ then follows.
In a recent breakthrough, Jain, Ji, Upadhyay, and Watrous proved that $\cls{PSPACE}$ also contains $\cls{QIP}$, from which one obtains $\cls{QIP}=\cls{PSPACE}$ \cite{JainJ+09}.

Feige and Kilian showed that every problem that can be solved in deterministic exponential-time can also be solved by an interactive proof with competing provers \cite{FeigeK97}.
That these proofs can be simulated in deterministic exponential time follows from Ref.~\cite{KollerM92}, and hence we have the complexity theoretic equality $\cls{RG}=\cls{EXP}$.
In this section we employ the algorithm from Section \ref{subsec:game-theory:alg} for computing the value of a zero-sum quantum game to show that $\cls{QRG}=\cls{EXP}$.

\subsubsection{Containment of $\cls{QRG}$ inside $\cls{EXP}$}

\begin{theorem}[$\cls{QRG}=\cls{EXP}$]

Every decision problem that admits a quantum interactive proof with competing provers also admits a deterministic exponential-time solution.
It follows that $\cls{QRG}=\cls{EXP}$.

\end{theorem}

\begin{proof}

Let $P$ be any problem that admits a quantum interactive proof with competing provers.
We will exhibit a deterministic exponential-time algorithm that, for each input $x$, correctly decides whether $x$ is a yes-instance of $P$ or a no-instance.

For each fixed input $x$, the actions of the verifier may be represented by a two-outcome measuring co-strategy $\set{R_\textrm{accept},R_\textrm{reject}}$ where the two provers combine to implement the corresponding strategy.
The operators in this co-strategy may be computed explicitly in deterministic exponential time by running the polynomial-time algorithm that generates a description of the quantum circuits corresponding to the actions of the verifier on input $x$ and then converting that description into an exponential-sized measuring co-strategy in the standard way.

Consider the two-player zero-sum quantum refereed game defined by the referee $\set{R_\textrm{accept},R_\textrm{reject}}$ and the payout function
\[ V(\textrm{accept}) = 1, \quad V(\textrm{reject}) = 0. \]
The value of this game equals the probability with which the verifier is convinced to accept the input when the two-provers act according to optimal strategies.
By Theorem \ref{thm:compute-value} (\thmcomputevalue), there is a deterministic algorithm that approximates this value to arbitrary precision in time polynomial in the bit length of $\set{R_\textrm{accept},R_\textrm{reject}}$, which we know to be exponential in the input $x$.
By running this algorithm and computing this probability, it is possible to determine whether or not $x$ is a yes-instance of $P$.
Thus, $P\in\cls{EXP}$.
As $P$ was selected arbitrarily from $\cls{QRG}$, it follows that $\cls{QRG}\subseteq\cls{EXP}$.
As $\cls{QRG}$ is already known to contain $\cls{EXP}$, it follows that $\cls{QRG}=\cls{EXP}.$
\end{proof}

\subsubsection{Parallel repetition of many-message quantum interactive proofs}

Thus far, our discussion of interactive proofs has been restricted to protocols that succeed with high probability.
But this condition is quite vague---what is a ``high'' probability?
Can interactive proofs be somehow transformed so as to amplify their probability of success?

To address these questions, we need a more specific definition of interactive proof.
A decision problem $P$ is said to admit an interactive proof with \emph{completeness error} $c$ and \emph{soundness error} $s$ if the following hold:
\begin{enumerate}

\item[(i)] If $x$ is a yes-instance of $P$ then there is a prover who can convince the verifier to accept with probability at least $1-c$.

\item[(ii)] If $x$ is a no-instance of $P$ then no prover can convince the verifier to accept with probability larger than $s$.

\end{enumerate}
It is not difficult to see that \emph{any decision problem whatsoever} admits an interactive proof with $1-c\leq s$.
Clearly then, interactive proofs are only interesting when $1-c > s$.
When this condition is met, the interactive proof may be transformed so as achieve
$c,s<\varepsilon$ for any desired $\varepsilon>0$.
This reduction in error is achieved by a method called \emph{sequential repetition}, whereby the verifier simply repeats the protocol many times in succession and then bases his final decision upon a weighted vote of the outcomes of the individual repetitions.
Under sequential repetition, the probability of error decreases exponentially in the number of repetitions.

A side-effect of sequential repetition is that the number of messages exchanged between the verifier and prover increases with each repetition.
Sometimes, it is desirable to achieve error reduction without increasing the number of messages in the interaction.
The standard transformation meeting this criterion is called \emph{parallel repetition}, whereby many copies of the protocol are executed simultaneously instead of sequentially.
The danger of parallel repetition is that the prover need not treat each repetition independently.
Instead, he might somehow attempt to correlate the repetitions so as to thwart the exponential reduction in error.
Therefore, before parallel repetition may be relied upon to reduce error, it must first be established that the prover cannot significantly affect the probability of acceptance by deviating from a strategy that treats the repetitions independently.

For classical single-prover interactive proofs, it can be shown that parallel repetition works as desired to achieve exponential error reduction.
For quantum single-prover interactive proofs, parallel repetition is known to work only when
\begin{enumerate}

\item[(i)] the verifier and prover exchange at most \emph{three messages}, and
\item[(ii)] the verifier's final decision depends upon a \emph{unanimous vote} of the repetitions (as opposed to some other weighting of the repetitions, such as a majority vote).

\end{enumerate}
(See Kitaev and Watrous for proofs of this and other properties of quantum interactive proofs \cite{KitaevW00}.)

In this section, we show that condition (i) is unnecessary.
Specifically, we prove the following.

\begin{theorem}[Parallel repetition of many-message quantum interactive proofs]
\label{thm:qip-parallel}

Consider the $k$-fold parallel repetition of an arbitrary $r$-round quantum interactive proof with soundness error $s$.
If the verifier in the repeated protocol accepts only when all $k$ repetitions accept then the repeated protocol has soundness error $s^k$.

\end{theorem}

\begin{proof}

Suppose that the interactive proof in the statement of the theorem solves the decision problem $P$.
The theorem makes no claim about the case where the input is a yes-instance of $P$, so we need only consider those inputs which are no-instances of $P$.
We must show that the verifier in the repeated protocol can be made to accept such an input with probability not exceeding $s^k$.

For any such input, the actions of the verifier may be represented by a two-outcome $r$-round measuring strategy $\set{R_\textrm{accept},R_\textrm{reject}}$ where the prover implements the corresponding $r$-round co-strategy.
As the verifier has soundness error $s$, it follows from Theorem \ref{thm:max-prob} (\thmmaxprob) that there exists an $r$-round non-measuring strategy $R$ for the referee's input and output spaces with the property that
\[ R_\textrm{accept}\preceq sR. \]
It follows from a semidefinite inequality proven in Ref.~\cite[Proposition 5]{CleveS+06-preprint} that
\[ R_\textrm{accept}^{\ot k} \preceq s^k R^{\ot k}. \]
(The cited inequality is not as explicit in the final journal version \cite{CleveS+08} of Ref.~\cite{CleveS+06-preprint}.)

For the repeated protocol with unanimous vote, Theorem \ref{thm:max-prob} implies that the maximum probability with which the verifier can be made to accept is given by
\[ \min \Set{\lambda\geq 0 : R_\textrm{accept}^{\ot k} \preceq\lambda Q \text{ for some $r$-round non-measuring strategy $Q$} }. \]
Taking $\lambda=s^k$ and $Q=R^{\ot k}$ completes the proof.
\end{proof}

Theorem \ref{thm:qip-parallel} makes no claim about the effect of a unanimous vote on the completeness error.
But for single-prover interactive proofs the question is moot, since these proofs can be assumed to have zero completeness error \cite{KitaevW00}.
In particular, for yes-instances of $P$ the prover in the $k$-fold repeated protocol
can achieve zero completeness error by playing an independent copy of the zero-error strategy for each of the $k$ repetitions.

\subsubsection{Parallel repetition of quantum interactive proofs with competing provers}

Fortunately, Theorem \ref{thm:qip-parallel} can be adapted with little difficulty to say something meaningful about the parallel repetition of quantum interactive proofs with competing provers.

\begin{theorem}[Parallel repetition of quantum interactive proofs with competing provers]
\label{thm:qrg-parallel}

Consider the $k$-fold parallel repetition of an arbitrary $r$-round quantum interactive proof with competing provers that has completeness error $c$ and soundness error $s$.

If the verifier in the repeated protocol accepts only when all $k$ repetitions accept then the repeated protocol has completeness error $kc$ and soundness error $s^k$.
Similarly, if the verifier in the repeated protocol rejects only when all $k$ repetitions reject then the repeated protocol has completeness error $c^k$ and soundness error $ks$.

\end{theorem}

\begin{proof}

We prove only the first claim in the statement of the theorem, as the second claim follows by symmetry.

Suppose that the interactive proof in the statement of the theorem solves the decision problem $P$ and suppose first that the input is a yes-instance of $P$.
We must show that there is a ``yes-prover'' who convinces the verifier in the repeated protocol to accept with probability at least $1-kc$, regardless of the strategy employed by the other prover, whom we call the ``no-prover''.

The yes-prover we seek merely plays an independent copy of the optimal strategy for each of the $k$ repetitions.
Of course, no no-prover can win any one of the $k$ repetitions with probability greater than $c$.
It then follows from the union bound that the repeated game with unanimous vote has completeness error at most $kc$.

Now suppose that the input is a no-instance of $P$.
For this case, we must find a no-prover who convinces the verifier in the repeated protocol to reject with probability at least $1-s^k$, regardless of the strategy employed by the yes-prover.

Toward that end, consider an optimal no-prover for the original protocol.
Borrowing from the proof of Theorem \ref{thm:qip-parallel}, the combined actions of the verifier and no-prover in this protocol may be represented by a two-outcome $r$-round measuring strategy $\set{N_\textrm{accept},N_\textrm{reject}}$ where the yes-prover implements the corresponding $r$-round co-strategy.
That the verifier-no-prover combination for the repeated protocol specified by $\set{N_\textrm{accept}',N_\textrm{reject}'}$ with \[N_\textrm{accept}'=N_\textrm{accept}^{\ot k}\] has soundness error $s^k$ now follows exactly as in the proof of Theorem \ref{thm:qip-parallel}.

Moreover, we see from this proof that the desired optimal no-prover for the repeated protocol simply plays an independent copy of the optimal strategy for each of the $k$ repetitions.
\end{proof}

Theorem \ref{thm:qrg-parallel} tells us that we may achieve an exponential reduction in completeness (or soundness) at the cost of a linear increase in soundness (or completeness).
Such a transformation is useful, for example, for protocols for which either completeness or soundness is already small---these proofs can be transformed so that \emph{both} completeness \emph{and} soundness are small.

More specifically, suppose we have a family of protocols with the property that, for some fixed constants $\varepsilon,s\in(0,1)$ and for any desired $m$, there is a protocol in the family with completeness error $\varepsilon^m$ and soundness error $s$.
Then for any desired $k$, we may choose $m$ large enough and employ $k$-fold parallel repetition with unanimous vote to obtain a protocol in which both completeness error \emph{and} soundness error are no larger than $s^k$.
Indeed, such a transformation was employed in Ref.~\cite{GutoskiW05} to reduce error for so-called ``short quantum games''.

\section{Kitaev's bound for strong coin-flipping}
\label{sec:coin-flip}

In this section we provide a simplified proof of Kitaev's bound for strong quantum coin-flipping, which states that a cheating party can always force any desired outcome upon an honest party with probability at least $1/\sqrt{2}$.
We begin with a definition of the coin-flipping problem, followed by a brief survey, before ending the section with our contribution.

\subsubsection{Strong and weak coin-flipping}

Suppose that two parties---Alice and Bob---wish to agree on a random bit.
(That is, they wish to flip a fair coin.)
The parties are physically separated, so the agreement must be reached by an exchange of messages via remote communication.
Moreover, the two parties do not trust each other, meaning that each party suspects that the other might attempt to force a particular result of the coin flip rather than settle for a uniformly random result.
The goal is to devise a protocol that produces the fairest possible coin flip, even in the case where one cheating party deliberately attempts to bias the result of the flip.


For example,
Alice could simply flip a fair coin for herself and announce the result to Bob, who meekly agrees to whatever Alice dictates.
If both parties adhere to these honest strategies then it is clear that they will produce a perfectly correlated random bit as desired.
Even if Bob cheats, there's nothing he can do to prevent Alice from producing a perfectly fair coin toss.
(Indeed, Alice's honest strategy is to completely ignore Bob.)
But if Alice decides to cheat then Bob is out of luck;
his mindless strategy of taking Alice at her word allows Alice to force Bob to produce any outcome she desires.
While this simple protocol is robust against a cheating Bob, it is completely vulnerable to a cheating Alice and is therefore a poor solution to the problem at hand.


Let us be more specific about this problem.
A \emph{strong coin-flipping protocol with bias $\varepsilon$} consists of an honest strategy for Alice and an honest strategy for Bob such that:
\begin{enumerate}

\item
If both parties follow their honest strategies then both parties produce the same outcome, and that outcome is chosen uniformly at random from $\set{0,1}$.

\item
If only one party follows his or her honest strategy then the maximum probability with which the cheating party can force the honest party to produce a given outcome
is at most $1/2+\varepsilon$.

\end{enumerate}
The adjective \emph{strong} refers to the fact that a cheater cannot bias an honest party's outcome toward either result 0 or 1.
By contrast, \emph{weak} protocols assume that one player desires outcome 0 and the other desires outcome 1.
The only requirement for weak protocols is that a cheater cannot bias the result toward his or her desired outcome.
In a \emph{quantum} coin-flipping protocol, the parties may process and exchange quantum information.

\subsubsection{Background}

This problem was introduced in the classical setting by Blum \cite{Blum81}, who called it \emph{coin-flipping by telephone}.
The problem is of interest to cryptographers because it is a primitive---a basic tool---for secure two-party computations.
In this setting there is little need to distinguish between strong and weak coin-flipping.
Indeed, Blum showed that even strong coin-flipping with zero bias is possible under certain computational assumptions.
Conversely, without any such assumptions it is not difficult to see that even weak coin-flipping
is impossible.
In particular, for each outcome $b\in\set{0,1}$ either (i) a computationally unrestricted cheating Alice can force outcome $b$ upon honest-Bob with certainty, or (ii) a computationally unrestricted cheating Bob can force the opposite outcome upon honest-Alice with certainty.

The quantum version of this problem admits some surprising contrasts to its classical counterpart.
Whereas unconditional classical coin-flipping is impossible, there is a weak quantum coin-flipping protocol that achieves \emph{arbitrarily small bias} in the limit of the number of messages exchanged between the parties \cite{Mochon07}.
Moreover, the existence of such a protocol for weak quantum coin-flipping implies the existence of a \emph{strong} quantum coin-flipping protocol that achieves bias arbitrarily close to $1/\sqrt{2}-1/2\approx 0.207$ \cite{ChaillouxK09}.
This protocol for strong quantum coin-flipping is the best possible, as it matches a lower bound due to Kitaev \cite{Kitaev02}.
(Kitaev did not publish this proof, but it appears in Refs.~\cite{AmbainisB+04,Roehrig04}.)
More comprehensive histories of quantum coin-flipping can be found in the breakthrough works of Mochon \cite{Mochon07}, and Chailloux and Kerenidis \cite{ChaillouxK09}.

\subsubsection{Our contribution}

We now provide an alternate and simplified proof of Kitaev's lower bound for strong quantum coin-flipping.

\begin{theorem}[Kitaev's bound for strong quantum coin-flipping]
\label{thm:Kitaev:coin-flip}

In any strong quantum coin-flipping protocol, a computationally unrestricted cheating party can always force a given outcome on an honest party with probability at least $1/\sqrt{2}$.

\end{theorem}

\begin{proof}[Our new proof]

In any quantum coin-flipping protocol the actions of honest-Alice are represented by a two-outcome measuring strategy $\set{A_0,A_1}$ and the actions of honest-Bob are represented by a two-outcome measuring co-strategy $\set{B_0,B_1}$.
By Theorem \ref{theorem:inner-product} (\theoreminnerproduct), the definition of a quantum coin-flipping protocol requires
\[ 1/2 = \inner{A_0}{B_0} = \inner{A_1}{B_1}. \]
Choose any outcome $b\in\set{0,1}$ and let $p$ denote the maximum probability with which a cheating Bob could force honest-Alice to output $b$.
Obviously we have $p\geq 1/2$.
Theorem \ref{thm:max-prob} (\thmmaxprob) implies that there must exist a strategy $Q$ for Alice such that $A_b\preceq pQ$.
If a cheating Alice plays this strategy $Q$ then honest-Bob outputs $b$ with probability
\[ \inner{Q}{B_b} \geq \frac{1}{p}\inner{A_b}{B_b} = \frac{1}{2p}. \]
As $\max\set{p,\frac{1}{2p}}\geq 1\sqrt{2}$ for all $p>0$, it follows that either honest-Alice or honest-Bob can be convinced to output $b$ with probability at least $1/\sqrt{2}$.
\end{proof}

This proof makes clear the limitations of strong coin-flipping protocols: the inability of Bob to force Alice to output $b$ directly implies that Alice can herself bias the outcome toward $b$.
By definition, weak coin-flipping protocols are not subject to this limitation.

\chapter{Distance Measures} \label{ch:norms}

In this chapter we introduce a new norm for super-operators that generalizes the diamond norm and we argue that this norm quantifies the observable difference between quantum strategies.
Indeed, we establish an extension of a result from Ref.~\cite{GutoskiW05} stating that each pair of convex sets $\bS_0,\bS_1$ of $r$-round strategies has a fixed $r$-round measuring co-strategy $\set{T_a}$ with the property that any pair $S_0\in\bS_0,S_1\in\bS_1$ can be distinguished by $\set{T_a}$ with probability determined by the minimal distance between $\bS_0$ and $\bS_1$ as measured by our new norm.

In order to prove this distinguishability result we also establish several properties of the new norm, including characterizations of its unit ball and dual norm.
We end the chapter with a discussion of the dual of the diamond norm.

When not explicitly stated otherwise, the results of this chapter are the sole work of the author and were not published prior to the present thesis.

\section{A new norm for strategies}
\label{sec:dist:snorm}

In order to best argue that our new norm captures the distinguishability of quantum strategies, we begin with a discussion of the trace norm and explain its use in distinguishing quantum states.
We then proceed with a discussion of the diamond norm for super-operators and explain its use in distinguishing quantum operations.
Only then do we introduce our new norm and note the similarities with the diamond norm and trace norm for the purpose of distinguishing quantum strategies.
We conclude the section with several immediate observations concerning this new norm, including the fact that it agrees with the diamond norm wherever it is defined.

\subsubsection{The trace norm as a distance measure for quantum states}

Recall from Section \ref{sec:intro:linalg} that the \emph{trace norm} $\tnorm{X}$ of an arbitrary operator $X$ is defined as the sum of the singular values of $X$.
If $X$ is Hermitian then it is a simple exercise to verify that its trace norm is given by
\begin{align*}
\tnorm{X}
&= \max \Set{ \inner{P_0-P_1}{X} \::\: P_0,P_1\succeq 0,\ P_0+P_1=I } \\
&= \max \Set{ \inner{P_0-P_1}{X} \::\: \set{P_0,P_1} \textrm{ is a two-outcome measurement} }.
\end{align*}

The trace norm provides a physically meaningful distance measure for quantum states in the sense that it captures the maximum likelihood with which two states can be correctly distinguished.
Let us illustrate this fact with a simple example involving two parties called Alice and Bob and a fixed pair of quantum states $\rho_0,\rho_1$.
Suppose Bob selects a state $\rho\in\set{\rho_0,\rho_1}$ uniformly at random and gives Alice a quantum system prepared in state $\rho$.
Alice has a complete description of both $\rho_0$ and $\rho_1$, but she does not know which of the two was selected by Bob.
Her goal is to correctly guess which of $\set{\rho_0,\rho_1}$ was selected based upon the outcome of a measurement she conducts on $\rho$.

Since Alice's guess is binary-valued and completely determined by her measurement, that measurement can be assumed to be a two-outcome measurement $\set{P_0,P_1}$ wherein outcome $a\in\set{0,1}$ indicates a guess that Bob prepared $\rho=\rho_a$.
Letting $C$ denote the event that Alice's guess is correct, basic quantum formalism tells us that
\begin{align*}
\Pr[C] &= \textstyle \frac{1}{2}\inner{P_0}{\rho_0} + \frac{1}{2}\inner{P_1}{\rho_1}\\
\Pr[\lnot C] &= \textstyle \frac{1}{2}\inner{P_1}{\rho_0} + \frac{1}{2}\inner{P_0}{\rho_1},
\end{align*}
implying that
\[ \Pr[C] - \Pr[\lnot C] = \textstyle \frac{1}{2} \inner{P_0-P_1}{\rho_0-\rho_1}. \]
As $\Pr[C] + \Pr[\lnot C] = 1$, we obtain the following alternate expression for the probability with which Alice's guess is correct:
\[
  \Pr[C]  = \textstyle \frac{1}{2} + \frac{1}{4} \inner{P_0-P_1}{\rho_0-\rho_1} \leq
  \frac{1}{2} + \frac{1}{4} \tnorm{\rho_0-\rho_1}
\]
with equality achieved at the optimal measurement $\set{P_0,P_1}$ for Alice.
This fundamental observation was originally made by Helstrom \cite{Helstrom69}.

\subsubsection{The diamond norm as a distance measure for quantum operations}

The trace norm extends naturally to super-operators, but this extension does not lead to an overly useful distance measure for quantum operations.
To achieve such a measure, the trace norm must be modified as follows.

\begin{definition}[Super-operator trace norm, diamond norm]

  For an arbitrary super-operator $\Phi$ the \emph{super-operator trace norm} $\tnorm{\Phi}$ of $\Phi$ is defined as
  \[ \tnorm{\Phi} \defeq \max_{\tnorm{X}=1} \tnorm{\Phi(X)}. \]
  The \emph{diamond norm} $\dnorm{\Phi}$ of $\Phi$ is defined as
  \[ \dnorm{\Phi} \defeq \sup_\cW \Tnorm{\Phi\ot\idsup{\cW}} \]
  where the supremum is taken over all finite-dimensional complex Euclidean spaces $\cW$.
\end{definition}

Much is known of the diamond norm.
For example, if $\Phi$ has the form $\Phi:\lin{\cX}\to\lin{\cY}$ then the supremum in the definition of $\dnorm{\Phi}$ is always achieved by some space $\cW$ whose dimension does not exceed that of the input space $\cX$ \cite{Kitaev97,AharonovK+98}.
As a consequence, the supremum in the definition of the diamond norm can be replaced by a maximum.
Moreover, if $\Phi$ is Hermitian-preserving and $\dim(\cW)\geq\dim(\cX)$ then the maximum in the definition of $\tnorm{\Phi\ot\idsup{\cW}}$ is always achieved by some positive semidefinite operator $X$ \cite{RosgenW05}.
Thus, if $\Phi$ is Hermitian-preserving then its diamond norm is given by
\begin{align*}
  \dnorm{\Phi}
  &= \max \Tnorm{\Pa{\Phi\ot\idsup{\cW}}(\rho)} \\
  &= \max \Inner{P_0-P_1}{\Pa{\Phi\ot\idsup{\cW}}(\rho)}
\end{align*}
where the maxima in these two expressions are taken over all spaces $\cW$ with dimension at most $\cX$, all states $\rho\in\pos{\cX\ot\cW}$, and all two-outcome measurements $\set{P_0,P_1}\subset\pos{\cY\ot\cW}$.

The diamond norm is to quantum operations as the trace norm is to quantum states:
it provides a physically meaningful distance measure for quantum operations in the sense that the value $\dnorm{\Phi_0-\Phi_1}$ quantifies the observable difference between two quantum operations $\Phi_0,\Phi_1$.
As before, this fact may be illustrated with a simple example.
Suppose Bob selects an operation from $\Phi\in\set{\Phi_0,\Phi_1}$ uniformly at random.
Alice is granted ``one-shot, black-box'' access to $\Phi$ and her goal is to correctly guess which of $\Phi_0,\Phi_1$ was applied.
Specifically, Alice may prepare a quantum system in state $\rho$ and send a portion of that system to Bob, who applies $\Phi$ to that portion and then returns it to Alice.
Finally, Alice performs a two-outcome measurement $\set{P_0,P_1}$ on the resulting system $\Pa{\Phi\ot\identity}(\rho)$ where outcome $a\in\set{0,1}$ indicates a guess that $\Phi=\Phi_a$.

Letting $C$ denote the event that Alice's guess is correct, we may repeat our previous derivation of $\Pr[C]$ to obtain
\begin{align*}
  \Pr[C]  &= \textstyle \frac{1}{2} + \frac{1}{4}
  \inner{P_0-P_1}{
    \Pa{\Phi_0\ot\identity}(\rho)-
    \Pa{\Phi_1\ot\identity}(\rho)
  } \\
  &\leq \textstyle \frac{1}{2} + \frac{1}{4} \dnorm{\Phi_0-\Phi_1}
\end{align*}
with equality achieved at the optimal input state $\rho$ and measurement $\set{P_0,P_1}$ for Alice.

It is interesting to note that the ability to send only \emph{part} of the input state $\rho$ to Bob and keep the rest for herself can enhance Alice's ability to distinguish some pairs of quantum operations, as compared to a simpler test that involves sending the \emph{entire} input state to Bob.
Indeed, there exist pairs $\Phi_0,\Phi_1$ of quantum operations that are perfectly distinguishable when applied to half of a maximally entangled input state---that is, $\dnorm{\Phi_0-\Phi_1}=2$---yet they appear nearly identical when an auxiliary system is not used---that is, $\tnorm{\Phi_0-\Phi_1}\approx 0$.
An example of such a pair of super-operators can be found in Watrous \cite{Watrous08}, along with much of the discussion that has occurred thus far in this section.
It is this phenomenon that renders the super-operator trace norm less useful than the diamond norm for the study of quantum information.

Additional properties of the diamond norm are established in this thesis by Lemmas \ref{lemma:cp-dnorm}, \ref{lm:dnorm}, and \ref{lemma:dnorm:1d} of Section \ref{sec:diamond-dual}.

\subsubsection{The strategy $r$-norm as a distance measure for quantum strategies}

The simple guessing game between Alice and Bob extends naturally from quantum operations to quantum strategies.
Let $S_0,S_1$ be arbitrary $r$-round quantum strategies
and suppose Bob selects $S\in\set{S_0,S_1}$ uniformly at random.
Alice's task is to interact with Bob and then decide after the interaction whether Bob selected $S=S_0$ or $S=S_1$.

Theorem \ref{theorem:inner-product} (\theoreminnerproduct) establishes an inner product relationship between measuring strategies and co-strategies that is analogous to the relationship between states and measurements.
As such, much of our previous discussion concerning the task of distinguishing pairs of \emph{states} can be re-applied to the task of distinguishing pairs of \emph{strategies}.
In particular, Alice can be assumed to act according to some two-outcome $r$-round measuring co-strategy $\set{T_0,T_1}$ for Bob's input and output spaces, with outcome $a\in\set{0,1}$ indicating a guess that Bob acted according to strategy $S_a$.
Moreover, letting $C$ denote the event that Alice's guess is correct, we have
\[ \Pr[C]  = \textstyle \frac{1}{2} + \frac{1}{4} \inner{T_0-T_1}{S_0-S_1}. \]
Naturally, Alice maximizes the probability of a correct guess by maximizing this expression over all $r$-round measuring co-strategies $\set{T_0,T_1}$.

Of course, this guessing game is symmetric with respect to strategies and co-strategies.
In particular, if Bob's actions $S_0,S_1$ are \emph{co-strategies} instead of strategies then Alice's actions $\set{T_0,T_1}$ must be a measuring \emph{strategy} instead of a measuring co-strategy.
The probability with which Alice correctly guesses Bob's strategy is still given by
\[ \Pr[C]  = \textstyle \frac{1}{2} + \frac{1}{4} \inner{T_0-T_1}{S_0-S_1} \]
except that Alice now maximizes this probability over all $r$-round measuring \emph{strategies} $\set{T_0,T_1}$.

With this example in mind, we propose two new norms---one that captures the distinguishability of strategies and one that captures the distinguishability of co-strategies.

\def\defnewnorm{Strategy $r$-norm}
\begin{definition}[\defnewnorm]
\label{def:new-norm}

For any Hermitian-preserving super-operator $\Phi:\lin{\kprod{\cX}{1}{r}}\to\lin{\kprod{\cY}{1}{r}}$
we define
\begin{align*}
\Snorm{\Phi}{r} \ &\defeq \ \max \Set{ \Inner{T_0-T_1}{\jam{\Phi}}
: \textnormal{$\set{T_0,T_1}$ is a $r$-round measuring co-strategy} },\\
\Snorm{\Phi}{r}^* \ &\defeq \ \max \Set{ \Inner{S_0-S_1}{\jam{\Phi}}
: \textnormal{$\set{S_0,S_1}$ is a $r$-round measuring strategy} }.
\end{align*}
Sometimes it is convenient to write $\snorm{\jam{\Phi}}{r}$ instead of $\snorm{\Phi}{r}$, particularly when $\jam{\Phi}$ is an operator derived from quantum strategies.
Similarly, we may sometimes write $\snorm{\jam{\Phi}}{r}^*$ instead of $\snorm{\Phi}{r}^*$.
\end{definition}

If $S_0,S_1$ are strategies for input spaces $\cX_1,\dots,\cX_r$ and output spaces $\cY_1,\dots,\cY_r$ then it follows immediately that the maximum probability with which Alice can correctly distinguish $S_0$ from $S_1$ is
\[ \textstyle \frac{1}{2} + \frac{1}{4} \snorm{S_0-S_1}{r} \]
Likewise, if $S_0,S_1$ are co-strategies rather than strategies then the maximum probability with which Alice can correctly distinguish $S_0$ from $S_1$ is
\[ \textstyle \frac{1}{2} + \frac{1}{4} \snorm{S_0-S_1}{r}^* \]

It may seem superfluous to allow both strategies and co-strategies as descriptions for Bob's actions in this simple example, as every co-strategy may be written as a strategy via suitable relabelling of input and output spaces.
But there is something to be gained by considering both the norms $\snorm{\cdot}{r}$ and $\snorm{\cdot}{r}^*$.
Later, we will show that these norms are dual to each other and we will use this duality to generalize the simple guessing game of this section.

\subsubsection{Both a trace norm and an operator norm}

We just argued that the norms $\snorm{\cdot}{r}$ and $\snorm{\cdot}{r}^*$ are to strategies and co-strategies as the trace norm is to states---a distance measure.

But these new norms also admit a different interpretation.
It is easy to see that the maximum probability over all states $\rho$ with which a standard quantum measurement $\set{P_a}$ produces a given outcome $a$ is given by $\norm{P_a}$.
(This maximum is achieved for $\rho=vv^*$ where $v$ is an eigenvector associated with the largest eigenvalue of $P_a$.)
This observation extends unhindered to measuring strategies and co-strategies.

\begin{proposition}[Maximum output probability from the strategy $r$-norm]
\label{prop:max-prob-norm}

Let $\set{S_a}$ be an $r$-round measuring strategy.
The maximum probability with which a compatible $r$-round co-strategy could force $\set{S_a}$ to produce a given measurement outcome $a$ is given by $\snorm{S_a}{r}$.

Similarly, if $\set{T_b}$ is an $r$-round measuring co-strategy then the maximum probability with which a compatible $r$-round strategy could force $\set{T_b}$ to produce a given measurement outcome $b$ is given by $\snorm{T_b}{r}^*$.

\end{proposition}

\begin{proof}

The proofs of the two statements are completely symmetric, so we prove only the first.
As each $S_a$ is positive semidefinite, the maximum in the definition of $\snorm{S_a}{r}$ is achieved at an $r$-round measuring co-strategy $\set{T_0,T_1}$ with $T_1=0$.
That is,
\begin{align*}
\snorm{S_a}{r}
&= \max \Set{\inner{T_0-T_1}{S_a}: \set{T_0,T_1} \textrm{ is an $r$-round measuring co-strategy}}\\
&= \max \Set{\inner{T_0}{S_a}: T_0 \textrm{ is an $r$-round non-measuring co-strategy}}
\end{align*}
The proposition follows from Theorem \ref{theorem:inner-product} (\theoreminnerproduct).
\end{proof}

Thus, the strategy $r$-norm and its dual are to measuring strategies and co-strategies as the operator norm is to measurements---a measure of maximum output probability.
Juxtaposing these two interpretations, we see that both the norms $\snorm{\cdot}{r}$ and $\snorm{\cdot}{r}^*$ act as \emph{both} a trace norm \emph{and} an operator norm in the appropriate contexts.

\subsubsection{The strategy $1$-norm agrees with the diamond norm}

As suggested by the notation $\snorm{\cdot}{r}$, this new norm agrees with the diamond norm on Hermitian-preserving super-operators for the case $r=1$.
The proof of this fact is a simple exercise that serves as a convenient introduction to some of the techniques that will be employed later in this chapter.

\def\propdnormgen{Agreement with the diamond norm}
\begin{proposition}[\propdnormgen]
\label{prop:dnorm-gen}

For every Hermitian-preserving super-operator $\Phi$ it holds that \( \dnorm{\Phi}=\snorm{\Phi}{1} \).

\end{proposition}

\begin{proof}

The proposition is trivially true for $\Phi=0$, so we assume $\Phi\neq 0$ throughout.
We also assume that $\Phi$ has the form $\Phi:\lin{\cX}\to\lin{\cY}$.

The proposition is a simple consequence of Theorem \ref{theorem:inner-product} that is somewhat complicated by the fact that it must be proven for \emph{all} Hermitian-preserving super-operators---not just those which represent strategies.
Strictly speaking, Theorem \ref{theorem:inner-product} applies only to measuring strategies and co-strategies and therefore cannot be directly applied to an arbitrary Hermitian-preserving super-operator $\Phi$.
In order to use Theorem \ref{theorem:inner-product}, we first decompose $\Phi$ into a linear combination of measuring strategy elements.

Toward that end, we note that the Hermitian-preserving property of $\Phi$ implies that $\jam{\Phi}$ is a Hermitian operator.
In particular, it has a Jordan decomposition \[ \jam{\Phi}=R^+-R^- \] where the operators
$R^+,R^-\in\pos{\cY\ot\cX}$ are positive semidefinite and act on orthogonal subspaces.
Let $\varepsilon$ be a positive real number with the property that $\varepsilon R^+$ and $\varepsilon R^-$ are both elements of some three-outcome one-round measuring strategy $\set{S_0,S_1,S_2}$ for the input space $\cX$ and output space $\cY$.
An easy way to accomplish this is to take 
\begin{align*}
\varepsilon &= \frac{1}{\dim(\cY)\max\set{\norm{R^+},\norm{R^-}}},\\
(S_0,S_1,S_2) &= \textstyle \Pa{\varepsilon R^+,\,\varepsilon R^-,\, \frac{1}{\dim(\cY)}I_{\cY\ot\cX}-\varepsilon R^+-\varepsilon R^-}.
\end{align*}
For any one-round measuring co-strategy $\set{T_0,T_1}$ for the input space $\cX$ and output space $\cY$, Theorem \ref{theorem:inner-product} tells us that the probability with which an interaction between $\set{S_a}$ and $\set{T_b}$ yields measurement outcomes $(a,b)$ is given by
\[
  \Pr[(a,b)]
  = \inner{T_b}{S_a}.
\]
In particular, \[ \Pr[(0,b)]-\Pr[(1,b)] = \inner{T_b}{S_0-S_1} = \varepsilon \inner{T_b}{\jam{\Phi}}. \]

Similarly,
let $(\Psi,\set{Q_a})$ be an operational description of the strategy $\set{S_a}$ and
let $(\rho,\set{P_b})$ be an operational description of the co-strategy $\set{T_b}$.
These objects take the form
\begin{align*}
\Psi &: \lin{\cX}\to\lin{\cY\ot\cZ} & \rho &\in \pos{\cX\ot\cW} \\
\set{Q_a} &\subset \pos{\cZ} & \set{P_b} &\subset \pos{\cW}
\end{align*}
for some choice of spaces $\cW$ and $\cZ$.
Basic quantum formalism tells us that the probability with which an interaction between $(\Psi,\set{Q_a})$ and $(\rho,\set{P_b})$ yields measurement outcomes $(a,b)$ is given by
\[
  \Pr[(a,b)]
  = \inner{Q_a\ot P_b}{\Pa{\Psi\ot\idsup{\cW}}(\rho)}.
\]

Let $\Phi^\pm$ be the super-operators with $\jam{\Phi^\pm}=R^\pm$, so that \( \Phi = \Phi^+ - \Phi^- \).
According to
Definition \ref{def:m-strategy} (\defmstrategy),
the actions of $\Phi^\pm$ on any $X\in\lin{\cX}$ are given by
\begin{align*}
\Phi^+ &: \textstyle X \mapsto \frac{1}{\varepsilon} \Ptr{\cZ}{\Pa{Q_0\ot I_\cY} \Psi(X)},\\
\Phi^- &: \textstyle X \mapsto \frac{1}{\varepsilon} \Ptr{\cZ}{\Pa{Q_1\ot I_\cY} \Psi(X)}.
\end{align*}
In particular,
\begin{align*}
\Pr[(0,b)]
&= \varepsilon \inner{P_b}{\Pa{\Phi^+ \ot \idsup{\cW}}(\rho)},\\
\Pr[(1,b)]
&= \varepsilon \inner{P_b}{\Pa{\Phi^- \ot \idsup{\cW}}(\rho)}.
\end{align*}
and hence
\[
  \Pr[(0,b)] - \Pr[(1,b)]
  = \varepsilon \inner{P_b}{\Pa{\Phi \ot \idsup{\cW}}(\rho)}.
\]
Then by linearity we obtain
\begin{align*}
& \inner{P_0-P_1}{\Pa{\Phi\ot\idsup{\cW}}(\rho)}\\
={}& \textstyle \frac{1}{\varepsilon} \Pa{ \Pa{\Pr[(0,0)] - \Pr[(1,0)]} - \Pa{\Pr[(0,1)] - \Pr[(1,1)] } }\\
={}&\inner{T_0-T_1}{\jam{\Phi}}.
\end{align*}

To prove the proposition,
first select the co-strategy $\set{T_0,T_1}$ so that its operational description $(\rho,\set{P_a})$ achieves the maximum in the definition of the diamond norm.
In this case, the above expression implies $\dnorm{\Phi}\leq\snorm{\Phi}{1}$.
For the reverse inequality, select a co-strategy $\set{T_0,T_1}$ that achieves the maximum in the definition of $\snorm{\Phi}{1}$.
In this case, the above expression implies $\dnorm{\Phi}\geq\snorm{\Phi}{1}$.
\end{proof}

\subsubsection{Extension to non-Hermitian-preserving super-operators?}

According to Definition \ref{def:new-norm} (\defnewnorm), the norms $\snorm{\Phi}{r}$ and $\snorm{\Phi}{r}^*$ are defined only when $\Phi$ is a Hermitian-preserving super-operator.
Is there a natural extension of these norms to \emph{all} super-operators?
Ideally, such an extension would agree with Definition \ref{def:new-norm} on Hermitian-preserving super-operators and with the diamond norm when $r=1$.

What would such an extension look like?
Presumably, it would require a generalization of quantum strategies as they appear in Definitions \ref{def:naive-strategy} and \ref{def:strategy}.
For example, an $r$-tuple $(\Phi_1,\dots,\Phi_r)$ might denote an operational representation of a generalized strategy if and only if each of the super-operators $\Phi_i$ preserves trace norm.

Indeed, a generalization of Theorem \ref{theorem:inner-product} (\theoreminnerproduct) would probably be required in order to reason about these generalized strategies and define a generalized strategy $r$-norm.
Moreover, in the next section we will see that several basic properties of the strategy $r$-norm are established via Proposition \ref{prop:polar} (\proppolar) and ultimately Theorem \ref{thm:char} (\thmchar).
Presumably, a proper generalization of the strategy $r$-norm would also require corresponding generalizations of these two results.
Essentially, much of the formalism developed so far in this thesis for quantum strategies would need to be re-derived in a more general setting.

What use could such a generalization find in quantum information theory?
In looking to the diamond norm for inspiration, we find few examples of its use on non-Hermitian-preserving super-operators in a quantum information context.
Perhaps the most notable such use is the ``maximum output fidelity'' characterization of the diamond norm:

\begin{fact}[Maximum output fidelity characterization of the diamond norm]
\label{fact:dnorm-fidelity-char}

Let $\Phi:\lin{\cX}\to\lin{\cY}$ be an arbitrary super-operator.
Let $\cZ$ be a space and let $A,B:\cX\to\cY\ot\cZ$ be operators with $\Phi:X\mapsto\ptr{\cZ}{AXB^*}$.
Then
\[
  \dnorm{\Phi}
  = \max_{\rho,\sigma}\fid\Pa{\Psi_A(\rho),\Psi_B(\sigma)}
\]
where the super-operators $\Psi_A,\Psi_B:\lin{\cX}\to\lin{\cZ}$ are given by
\begin{align*}
\Psi_A &: X\mapsto\ptr{\cY}{AXA^*},\\
\Psi_B &: X\mapsto\ptr{\cY}{BXB^*}
\end{align*}
and $\fid(\rho,\sigma)$ denotes the fidelity function on pairs of quantum states
and the maximum is taken over all quantum states $\rho,\sigma\in\pos{\cX}$.
\end{fact}

This characterization has been employed in the study of quantum interactive proofs \cite{RosgenW05} and a proof may be found in Ref.~\cite{KitaevS+02}.
Could this characterization be generalized so as to yield an interesting quantum information theoretic application of the strategy $r$-norm to non-Hermitian-preserving super-operators?

While
an endeavor to generalize the strategy $r$-norm
would no doubt be an interesting mathematical exercise, its applicability to quantum information is not readily apparent.
In this thesis, 
our foray into the realm of mathematical curiosity of this flavor is limited to the
study of the dual of the diamond norm in Section \ref{sec:diamond-dual}.

\section{Distinguishing convex sets of strategies}

In this section we generalize the guessing game example from the previous section from a problem of distinguishing \emph{individual} states, operations, or strategies to distinguishing \emph{convex sets} of states, operations, or strategies.

Specifically, suppose two convex sets $\bA_0,\bA_1$ of quantum states are fixed.
Suppose that Bob arbitrarily selects $\rho_0\in\bA_0$ and $\rho_1\in\bA_1$ and then selects $\rho\in\set{\rho_0,\rho_1}$ uniformly at random and gives Alice a quantum system prepared in state $\rho$.
Alice's goal is to correctly guess whether $\rho\in\bA_0$ or $\rho\in\bA_1$ based upon the outcome of a measurement she conducts on $\rho$.
It is clear that this problem is a generalization of that from the previous section, as the original problem is recovered by considering the singleton sets $\bA_0=\set{\rho_0}$ and $\bA_1=\set{\rho_1}$.

This problem of distinguishing convex sets of states was solved in Ref.~\cite{GutoskiW05}, wherein it was shown that there exists a \emph{single} measurement $\set{P_0,P_1}$ that depends only upon the sets $\bA_0,\bA_1$ with the property that \emph{any} pair $\rho_0\in\bA_0$, $\rho_1\in\bA_1$ may be correctly distinguished with probability at least
\[ \frac{1}{2} + \frac{1}{4} \min_{\sigma_a\in\bA_a} \Tnorm{\sigma_0-\sigma_1}. \]
In particular, even if two distinct pairs $\rho_0,\rho_1$ and $\rho_0',\rho_1'$ both minimize the trace distance between $\bA_0$ and $\bA_1$ then \emph{both} pairs may be optimally distinguished by the \emph{same} measurement $\set{P_0,P_1}$.

What about distinguishing convex sets of quantum operations or strategies?
Nothing was known of either problem prior to the work of the present thesis.
In this section, we prove that the distinguishability result for convex sets of states extends unhindered to operations and strategies.
In particular, we prove that two convex sets $\bS_0,\bS_1$ of $r$-round strategies can be correctly distinguished with probability at least
\[ \frac{1}{2} + \frac{1}{4} \min_{S_a\in\bS_a} \Snorm{S_0-S_1}{r}.\]
It then follows trivially that two convex sets $\bT_0,\bT_1$ of $r$-round co-strategies can be correctly distinguished with probability at least
\[ \frac{1}{2} + \frac{1}{4} \min_{T_a\in\bT_a} \Snorm{T_0-T_1}{r}^*.\]
As a special case, it holds that two convex sets $\mathbf{\Phi}_0,\mathbf{\Phi}_1$ of quantum operations can be distinguished with probability at least
\[ \frac{1}{2} + \frac{1}{4}\min_{\Phi_a\in\mathbf{\Phi}_a} \Dnorm{\Phi_0-\Phi_1}. \]

\subsubsection{Properties of the strategy $r$-norm}

Our proof of the distinguishability of convex sets of strategies is essentially a copy of the proof appearing in Ref.~\cite{GutoskiW05}
with states and measurements replaced by strategies and co-strategies and the trace and operator norms replaced with the strategy $r$-norm and its dual.

To ensure correctness of the new proof, we must identify the relevant properties of the trace and operator norms employed in Ref.~\cite{GutoskiW05} and then establish suitable analogues of those properties for the new norms $\snorm{\cdot}{r}$ and $\snorm{\cdot}{r}^*$.
The relevant properties of the trace and operator norms are:
\begin{enumerate}

\item \label{item:op-tr-dual}

The trace norm and operator norm are \emph{dual} to each other, meaning that
\begin{align*}
\tnorm{X} &= \max_{\norm{Y}\leq 1} \abs{\inner{Y}{X}},\\
\norm{X} &= \max_{\tnorm{Y}\leq 1} \abs{\inner{Y}{X}}
\end{align*}
for all operators $X$.

\item \label{item:unit-ball}

If $X$ is Hermitian then
\[ \norm{X}\leq 1 \iff \abs{X}\preceq I. \]
Here $\abs{X} \defeq X^+ + X^-$ where $X=X^+-X^-$ is a Jordan decomposition of $X$.

\end{enumerate}
Property \ref{item:unit-ball} is really just a circuitous way of saying that a Hermitian operator has operator norm at most $1$ if and only if all its eigenvalues are no larger than $1$ in absolute value.
Compared to the duality of the trace and operator norms (property \ref{item:op-tr-dual}), this observation is not deep or significant in and of itself.
It is only phrased as such so as to highlight the forthcoming generalization to the strategy $r$-norm and its dual.

Our generalization of property \ref{item:op-tr-dual} states that the norms $\snorm{\cdot}{r}$ and $\snorm{\cdot}{r}^*$ are dual to each other.
Because our proof of this fact relies upon the generalization of property \ref{item:unit-ball}, we first establish the latter.

\def\propunitball{Unit ball of the strategy $r$-norms}
\subsubsection{\propunitball}

\begin{proposition}[\propunitball]
\label{prop:unit-ball}

For any Hermitian-preserving super-operator $\Phi:\lin{\kprod{\cX}{1}{r}}\to\lin{\kprod{\cY}{1}{r}}$ it holds that
\[ \snorm{\Phi}{r}\leq 1 \iff \abs{\jam{\Phi}}\preceq S \]
for some $r$-round non-measuring strategy $S$ for input spaces $\cX_1,\dots,\cX_r$ and output spaces $\cY_1,\dots,\cY_r$.
Similarly,
\[ \snorm{\Phi}{r}^*\leq 1 \iff \abs{\jam{\Phi}}\preceq T \]
for some $r$-round non-measuring co-strategy $T$ for input spaces $\cX_1,\dots,\cX_r$ and output spaces $\cY_1,\dots,\cY_r$.

\end{proposition}

Proposition \ref{prop:unit-ball} follows immediately from the following two lemmas.

\begin{lemma}
\label{lemma:unit-ball-technical}

For any Hermitian-preserving super-operator $\Phi:\lin{\kprod{\cX}{1}{r}}\to\lin{\kprod{\cY}{1}{r}}$ it holds that
\begin{align*}
\snorm{\Phi}{r}\leq 1 &\iff \inner{\abs{\jam{\Phi}}}{T}\leq 1 \textrm{ for every $r$-round non-measuring co-strategy $T$},\\
\snorm{\Phi}{r}^*\leq 1 &\iff \inner{\abs{\jam{\Phi}}}{S}\leq 1 \textrm{ for every $r$-round non-measuring strategy $S$}.
\end{align*}

\end{lemma}

\begin{lemma}
\label{lemma:polar-restate}

For each positive semidefinite operator $X\in\pos{\kprod{\cY}{1}{r}\ot\kprod{\cX}{1}{r}}$ it holds that
\begin{align*}
& \inner{X}{T}\leq 1 \textrm{ for every $r$-round non-measuring co-strategy $T$} \\
\iff{}& X\preceq S \textrm{ for some $r$-round non-measuring strategy $S$},\\[2mm]
& \inner{X}{S}\leq 1 \textrm{ for every $r$-round non-measuring strategy $S$} \\
\iff{}& X\preceq T \textrm{ for some $r$-round non-measuring co-strategy $T$}.
\end{align*}

\end{lemma}

Lemma \ref{lemma:polar-restate} is an immediate corollary of Proposition \ref{prop:polar} (\proppolar), which was proven in Section \ref{subsec:convex-polarity}.
Hence, we do not prove Lemma \ref{lemma:polar-restate} here---it is restated only for ease of reference without the polarity notation.
It remains only to prove Lemma \ref{lemma:unit-ball-technical}.

\begin{proof}[Proof of Lemma \ref{lemma:unit-ball-technical} and hence also of Proposition \ref{prop:unit-ball}]

The proofs of the two desired equivalences are completely symmetric, so we prove only the first.
As $\Phi$ is Hermitian-preserving, it holds that $\jam{\Phi}$ is a Hermitian operator.
As such, it has a Jordan decomposition \( \jam{\Phi}=R^+-R^- \)
with $\abs{\jam{\Phi}}=R^++R^-$.

We first prove the easier implication.
Suppose $\inner{\abs{\jam{\Phi}}}{T}\leq 1$ for all co-strategies $T$ and
choose any measuring co-strategy $\set{T_0,T_1}$.
We must show that $\inner{\jam{\Phi}}{T_0-T_1}\leq 1$.
As $\set{T_0,T_1}$ is a measuring co-strategy, it holds that $T_0+T_1$ is a nonmeasuring co-strategy.
Taking $T=T_0+T_1$, we get
\[ 1 \geq \inner{R^++R^-}{T_0+T_1} \geq \inner{R^+-R^-}{T_0-T_1} \]
as desired.

For the other direction, suppose $\snorm{\Phi}{r}\leq 1$ and
choose any co-strategy $T$.
We must show that $\inner{\abs{\jam{\Phi}}}{T}\leq 1$.
Our goal is to decompose $T=T_0 + T_1$ in such a way that $\set{T_0,T_1}$ is a measuring co-strategy with $\inner{R^+}{T_1}=\inner{R^-}{T_0}=0$.
Such a choice of $T_0,T_1$ ensures that
\[ \inner{R^+-R^-}{T_0-T_1} = \inner{R^++R^-}{T_0+T_1}. \]
The lemma is then established by noting that the left side of this expression is at most 1 and the right side of this expression equals $\inner{\abs{\jam{\Phi}}}{T}$.

It remains to choose the appropriate $T_0,T_1$.
Letting $\Pi^+$ denote the projection onto the support of $R^+$, we choose
\begin{align*}
T_0 &= \Pi^+ T \Pi^+, \\
T_1 &= T - T_0.
\end{align*}
That $\inner{R^+}{T_1}=\inner{R^-}{T_0}=0$ follows from the fact that $R^+$ and $R^-$ have orthogonal support.
Finally, it is easy to verify that $\set{T_0,T_1}$ denotes a valid measuring co-strategy:
as $T$ is a co-strategy, so too is the sum $T_0+T_1=T$.
As $T$ is positive semidefinite, so too are $T_0,T_1$.
\end{proof}

\def\propduality{Duality of the strategy $r$-norms}
\subsubsection{\propduality}

We are now ready to prove that the norms $\snorm{\cdot}{r}$ and $\snorm{\cdot}{r}^*$ are dual to each other.

\begin{proposition}[\propduality]
\label{prop:duality}

The norms $\snorm{\cdot}{r}$ and $\snorm{\cdot}{r}^*$ are dual to each other.
In other words, for any Hermitian-preserving super-operator $\Phi:\lin{\kprod{\cX}{1}{r}}\to\lin{\kprod{\cY}{1}{r}}$ it holds that
\begin{align*}
\snorm{\Phi}{r} &= \max_{\snorm{\Psi}{r}^*\leq 1} \inner{\Psi}{\Phi}, \\
\snorm{\Phi}{r}^* &= \max_{\snorm{\Psi}{r}\leq 1} \inner{\Psi}{\Phi}.
\end{align*}

\end{proposition}

\begin{proof}

The proofs of the two desired equalities are completely symmetric, so we prove only the first.
We begin by proving
\[ \snorm{\Phi}{r} \leq \max_{\snorm{\Psi}{r}^*\leq 1} \inner{\Psi}{\Phi}. \]
Let $\set{T_0,T_1}$ be an $r$-round measuring co-strategy attaining the maximum in the definition of $\snorm{\Phi}{r}$, so that \[\snorm{\Phi}{r} = \inner{T_0-T_1}{\jam{\Phi}}.\]
The desired inequality follows from the claim that $\snorm{T_0-T_1}{r}^*\leq 1$.
To verify this claim, we note that any $r$-round measuring strategy $\set{S_0,S_1}$ attaining the maximum in the definition of $\snorm{T_0-T_1}{r}^*$ has
\[ \snorm{T_0-T_1}{r}^* = \inner{S_0-S_1}{T_0-T_1} \leq \inner{S_0+S_1}{T_0+T_1} = 1. \]
That the final inner product equals one follows immediately from Theorem \ref{theorem:inner-product} (\theoreminnerproduct) and the observation that $S_0+S_1$ may be viewed as a one-outcome measuring strategy and $T_0+T_1$ may be viewed as a one-outcome measuring co-strategy.

For the reverse inequality, choose a Hermitian-preserving super-operator $\Psi$ with $\snorm{\Psi}{r}^*\leq 1$ that maximizes the inner product $\inner{\Psi}{\Phi}$.
Proposition \ref{prop:unit-ball} tells us that $\abs{\jam{\Psi}}\preceq T$ for some $r$-round non-measuring co-strategy $T$.
Let $\jam{\Psi}=R^+-R^-$ be the Jordan decomposition of $\jam{\Psi}$ and let $\set{T_0,T_1}$ be the $r$-round measuring co-strategy given by
\begin{align*}
T_0 &= R^+ + \frac{1}{2}\Pa{T - \abs{\jam{\Psi}}}, \\
T_1 &= R^- + \frac{1}{2}\Pa{T - \abs{\jam{\Psi}}}.
\end{align*}
Then
\[ \inner{\Psi}{\Phi} = \inner{T_0-T_1}{\jam{\Phi}} \leq \snorm{\Phi}{r} \]
as desired.
\end{proof}

\def\thmsep{Distinguishability of convex sets of strategies}
\subsubsection{\thmsep}

As mentioned earlier, our proof of the distinguishability of convex sets of strategies closely resembles the proof of the distinguishability of convex sets of states appearing in Ref.~\cite{GutoskiW05}.
Now that we have identified and established the relevant properties of the norm $\snorm{\cdot}{r}$ and its dual $\snorm{\cdot}{r}^*$, we are ready to translate the proof of Ref.~\cite{GutoskiW05} onto the domain of quantum strategies.

\begin{theorem}[\thmsep]
\label{thm:sep}

Let $\bS_0,\bS_1\subset\pos{\kprod{\cY}{1}{r}\ot\kprod{\cX}{1}{r}}$ be nonempty convex sets of $r$-round strategies.
There exists an $r$-round measuring co-strategy $\set{T_0,T_1}$ with the property that
\[ \Inner{T_0-T_1}{S_0-S_1} \geq \min_{R_a\in\bS_a} \Snorm{R_0-R_1}{r} \]
for all choices of $S_0\in\bS_0$ and $S_1\in\bS_1$.
A similar statement holds in terms of the dual norm $\snorm{\cdot}{r}^*$ for convex sets of co-strategies.

\end{theorem}





\begin{proof}

The proof for co-strategies is completely symmetric to the proof for strategies, so we only address strategies here.
Let $d$ denote the minimum distance between $\bS_0$ and $\bS_1$ as stated in the theorem.
If $d=0$ then the theorem is satisfied by the trivial $r$-round measuring co-strategy corresponding to a random coin flip.
(For this trivial co-strategy, both $T_0$ and $T_1$ are equal to the identity divided by $2\dim(\kprod{\cX}{1}{r})$.)
For the remainder of this proof, we shall restrict our attention to the case $d>0$.

Define
\[ \bS \defeq \bS_0-\bS_1 = \Set{S_0-S_1 : S_0\in\bS_0, S_1\in\bS_1 } \]
and let
\[ \bB \defeq \Set{B\in\her{\kprod{\cY}{1}{r}\ot\kprod{\cX}{1}{r}} : \snorm{B}{r}<d } \]
denote the open ball of radius $d$ with respect to the $\snorm{\cdot}{r}$ norm.
The sets $\bS$ and $\bB$ are nonempty disjoint sets of Hermitian operators, both are convex, and $\bB$ is open.
By the Separation Theorem (Fact \ref{fact:herm-sep}), there exists a Hermitian operator $H$ and a scalar $\alpha$ such that
\[ \inner{H}{S} \geq \alpha > \inner{H}{B} \]
for all $S\in\bS$ and $B\in\bB$.

For every choice of $B\in\bB$ we have $-B\in\bB$ as well, from which it follows that
$\abs{\inner{H}{B}}<\alpha$ for all $B\in\bB$ and hence $\alpha>0$.
Moreover, as $\bB$ is the open ball of radius $d$ in the norm $\snorm{\cdot}{r}$, it follows from Proposition \ref{prop:duality} (\propduality) that
\[\snorm{H}{r}^* \leq \alpha/d. \]
Now let $\hat{H}=\frac{d}{\alpha}H$ be the normalization of $H$ with $\snorm{\hat{H}}{r}^*\leq 1$ and let $\hat{H}=\hat{H}^+-\hat{H}^-$ be the Jordan decomposition of $\hat{H}$.
By Proposition \ref{prop:unit-ball} (\propunitball) we have $\abs{\hat{H}}\preceq T$ for some co-strategy $T$.
Let $\set{T_0,T_1}$ be the measuring co-strategy given by
\begin{align*}
T_0 &= \hat{H}^+ + \frac{1}{2}\Pa{T - \abs{\hat{H}}}, \\
T_1 &= \hat{H}^- + \frac{1}{2}\Pa{T - \abs{\hat{H}}}.
\end{align*}
It remains only to verify that $\set{T_0,T_1}$ has the desired properties:
for every choice of $S_0\in\bS_0$ and $S_1\in\bS_1$ we have
\[ \inner{T_0-T_1}{S_0-S_1} = \inner{\hat{H}}{S_0-S_1} = \frac{d}{\alpha} \inner{H}{S_0-S_1} \geq d \]
as desired.
\end{proof}

The claimed result regarding the distinguishability of convex sets of strategies now follows immediately.
To recap, let $\bS_0,\bS_1$ be convex sets of strategies and let $\set{T_0,T_1}$ denote the measuring co-strategy from Theorem \ref{thm:sep} that distinguishes elements in $\bS_0$ from elements in $\bS_1$.
Suppose Bob selects $S_0\in\bS_0$ and $S_1\in\bS_1$ arbitrarily and then selects $S\in\set{S_0,S_1}$ uniformly at random.
As derived in Section \ref{sec:dist:snorm}, if Alice acts according to $\set{T_0,T_1}$ then the probability with which she correctly guesses whether $S\in\bS_0$ or $S\in\bS_1$ is given by
\[ \frac{1}{2} + \frac{1}{4}\inner{T_0-T_1}{S_0-S_1} \geq \frac{1}{2} + \frac{1}{4}\min_{R_a\in\bS_a}\Snorm{R_0-R_1}{r}\]
as desired.

\section{The dual of the diamond norm}
\label{sec:diamond-dual}

In Proposition \ref{prop:dnorm-gen} we established that the strategy $1$-norm $\snorm{\Phi}{1}$ agrees with the diamond norm $\dnorm{\Phi}$ on Hermitian-preserving super-operators $\Phi$.
Can the same be said of the duals of these norms?
In this section we answer that question in the affirmative.
In order to do so, we also establish several basic facts concerning the dual of the diamond norm.
While none of these facts are surprising, the only proofs we can offer are nontrivial and possibly interesting in their own right.
Thus, in this section we take a ``detour'' from the strategy $r$-norm in order to study the dual of the diamond norm, returning only at the very end to answer this section's opening question.

\def\defdiamonddual{Dual of the diamond norm}
\begin{definition}[\defdiamonddual]

  As with any norm, the \emph{dual} $\dnorm{\cdot}^*$ of the diamond norm $\dnorm{\cdot}$ is defined for every super-operator $\Phi$ as
  \[ \Dnorm{\Phi}^* \defeq \max_{\dnorm{\Psi}=1} \Abs{\Inner{\Psi}{\Phi}}. \]
\end{definition}

That the dual of the dual of the diamond norm equals the diamond norm follows from the Duality Theorem, a proof of which can be found in Horn and Johnson \cite{HornJ85}.
In other words, it holds that
\[ \Dnorm{\Phi}^{**} =\Dnorm{\Phi} = \max_{\dnorm{\Psi}^*=1} \Abs{\Inner{\Psi}{\Phi}}. \]
While much is known of the diamond norm, its dual has never been studied.
In this section, we establish the following basic facts about this norm:
\begin{enumerate}


\item
If $\Phi$ is completely positive then the maximum in the definition of $\dnorm{\Phi}^*$ is achieved by a completely positive and trace-preserving super-operator.

\item
If $\Phi$ is Hermitian-preserving then the maximum in the definition of $\dnorm{\Phi}^*$ is achieved by a Hermitian-preserving super-operator.

\item 
For each of the two maxima
\[
  \dnorm{\Phi}^* = \max_{\dnorm{\Psi}=1} \Abs{\Inner{\Psi}{\Phi}},
  \quad 
  \dnorm{\Phi} = \max_{\dnorm{\Psi}^*=1} \Abs{\Inner{\Psi}{\Phi}}
\]
there exist super-operators $\Phi$ such that the maximum is not attained by any Hermitian-preserving super-operator $\Psi$.
Moreover, there exist Hermitian-preserving super-operators $\Phi$ such that the maximum is not attained by any completely positive super-operator $\Psi$.

\item
If $\Phi$ is Hermitian-preserving then $\dnorm{\Phi}^*=\snorm{\Phi}{1}^*$.

\end{enumerate}

\subsubsection{Useful lemmas involving the diamond norm}

We begin with three technical lemmas, two of which establish facts about the diamond norm.
The first lemma follows immediately from Refs.~\cite{RosgenW05,AmosovH+00a}.
We provide an alternate proof for completeness.

\begin{lemma}
\label{lemma:cp-dnorm}

For any completely positive super-operator $\Phi$ it holds that $\tnorm{\Phi}=\dnorm{\Phi}$.
Moreover, the maximum in the definition of $\tnorm{\Phi}$ is achieved by a positive semidefinite operator.

\end{lemma}

\begin{proof}

It is clear that $\dnorm{\Phi}\geq\tnorm{\Phi}$, so let us concentrate only on the reverse inequality.
Suppose $\Phi$ has the form $\Phi:\lin{\cX}\to\lin{\cY}$.
As $\Phi$ is completely positive, there exists a space $\cW$ and a density operator $\rho\in\pos{\cX\ot\cW}$ with the property that \[ \dnorm{\Phi}=\Tnorm{\Pa{\Phi\ot\idsup{\cW}}(\rho)}=\Ptr{}{\Pa{\Phi\ot\idsup{\cW}}(\rho)}. \]
By convexity, we may assume that $\rho$ is a pure state---that is, $\rho=uu^*$ for some unit vector $u\in\cX\ot\cW$.
For any orthonormal basis $\set{e_1,\dots,e_{\dim(\cW)}}$ of $\cW$ we may write
\[ u=\sum_{i=1}^{\dim(\cW)} \sqrt{q_i} x_i\ot e_i \]
for some choice of unit vectors $x_1,\dots,x_{\dim(\cW)}\in\cX$ (not necessarily orthogonal) and nonnegative real numbers $q_1,\dots,q_{\dim(\cW)}$ that sum to one.
Then
\begin{align*}
\dnorm{\Phi}
&= \sum_{i,j=1}^{\dim(\cW)} \sqrt{q_iq_j}\Ptr{}{\Phi(x_ix_j^*)\ot e_ie_j^*}
= \sum_{i=1}^{\dim(\cW)} q_i \Ptr{}{\Phi(x_ix_i^*)}
= \Ptr{}{\Phi(\sigma)}
\end{align*}
for $\sigma=\sum_{i=1}^{\dim(\cW)} q_i x_ix_i^*$.
\end{proof}

For any operator $A:\cX\to\cY$ of rank $r$ it follows from the Singular Value Theorem (Section \ref{sec:intro:linalg}) that it is possible to choose vectors $u_1,\dots,u_r\in\cY$, $v_1,\dots,v_r\in\cX$ such that \[ A = \sum_{i=1}^r u_iv_i^*. \]
Given such a choice of vectors, we define the positive semidefinite operators $A_\rL\in\pos{\cY}$, $A_\rR\in\pos{\cX}$ by
\[ A_\rL = \sum_{i=1}^r u_iu_i^*, \quad A_\rR = \sum_{i=1}^r v_iv_i^*. \]
If $A$ is Hermitian then it is clear that we may take $A_\rL=A_\rR=\abs{A}$.
As a consequence, if $A$ is positive semidefinite then we may take $A_\rL=A_\rR=A$.
The following simple lemma is a special case of Ref.~\cite[Lemma 2]{Watrous05}.

\begin{lemma}
\label{lemma:cauchy}

For any operators $A,B:\cX\to\cY$ and any decompositions $A_\rL,A_\rR$ and $B_\rL,B_\rR$ of those operators it holds that
\[ \Abs{\Inner{A}{B}}^2 \leq \Inner{A_\rL}{B_\rL} \cdot \Inner{A_\rR}{B_\rR}. \]

\end{lemma}

\begin{proof}

Let \[ A=\sum_{i=1}^r u_i v_i^*, \quad B=\sum_{j=1}^q w_j x_j^* \]
be decompositions of $A,B$ that yield $A_\rL,A_\rR$ and $B_\rL,B_\rR$.
Then
\begin{align*}
\Abs{\Inner{A}{B}}
&= \Abs{ \sum_{i=1}^r \sum_{j=1}^q \Inner{u_iv_i^*}{w_jx_j^*} }
= \Abs{ \sum_{i=1}^r \sum_{j=1}^q
  \Inner{u_i}{w_j} \cdot \Inner{x_j}{v_i} } \\
&\leq \sqrt{ \sum_{i=1}^r \sum_{j=1}^q \Abs{\Inner{u_i}{w_j}}^2 } \cdot
  \sqrt{ \sum_{i=1}^r \sum_{j=1}^q \Abs{\Inner{x_j}{v_i}}^2 } \\
&= \sqrt{ \sum_{i=1}^r \sum_{j=1}^q \Inner{u_i u_i^*}{w_j w_j^*} } \cdot
  \sqrt{ \sum_{i=1}^r \sum_{j=1}^q \Inner{v_i v_i^*}{x_j x_j^*} } \\
&= \sqrt{\Inner{A_\rL}{B_\rL}} \cdot \sqrt{\Inner{A_\rR}{B_\rR}}
\end{align*}
where the inequality follows from Cauchy-Schwarz.
\end{proof}

For any super-operator $\Phi$ and any decomposition $\jam{\Phi}_\rL,\jam{\Phi}_\rR$ of $\jam{\Phi}$ we let $\Phi_\rL,\Phi_\rR$ denote the super-operators with $\jam{\Phi_\rL}=\jam{\Phi}_\rL$ and $\jam{\Phi_\rR}=\jam{\Phi}_\rR$.
As $\jam{\Phi_\rL},\jam{\Phi_\rR}$ are positive semidefinite, it holds that $\Phi_\rL,\Phi_\rR$ are completely positive.

\begin{lemma}
\label{lm:dnorm}

For any super-operator $\Phi:\lin{\cX}\to\lin{\cY}$ there exists a decomposition $\Phi_\rL$, $\Phi_\rR$ of $\Phi$ with
\[ \Dnorm{\Phi}=\Dnorm{\Phi_\rL}=\Dnorm{\Phi_\rR}. \]

\end{lemma}

\begin{proof}

It was noted in the conclusion of Ref.~\cite{Watrous05} that
\[ \dnorm{\Phi}^2=\inf\Set{\tnorm{\Phi_\rL}\cdot\tnorm{\Phi_\rR}} \]
where the infimum is taken over all decompositions $\Phi_\rL,\Phi_\rR$ of $\Phi$.
The existence of a fixed pair $(\Phi_\rL,\Phi_\rR)$ that achieve this infimum may be argued as in Ref.~\cite{KitaevS+02}.
By rescaling, we may assume that $\tnorm{\Phi_\rL}=\tnorm{\Phi_\rR}$, from which it follows that
\[ \dnorm{\Phi} = \tnorm{\Phi_\rL} = \tnorm{\Phi_\rR}. \]
The lemma then follows from Lemma \ref{lemma:cp-dnorm}.
\end{proof}

\subsubsection{Achieving the maximum}

We are now ready to exhibit two theorems that establish two of the claims from the beginning of this section.

For any Hermitian-preserving super-operator $\Phi$, we let $\abs{\Phi}$ denote the completely positive super-operator with $\jam{\abs{\Phi}}=\abs{\jam{\Phi}}$.
Just as with operators, there is a decomposition $\Phi_\rL,\Phi_\rR$ of $\Phi$ with $\Phi_\rL=\Phi_\rR=\abs{\Phi}$.
Moreover, if $\Phi$ is completely positive then we may take $\Phi_\rL=\Phi_\rR=\Phi$.

\begin{theorem}[Achieving the maximum for completely positive super-operators]
\label{theorem:cp-dnorm}

For any Hermitian-preserving super-operator $\Phi$ it holds that
\[ \dnorm{\Phi}^*\leq\dnorm{\abs{\Phi}}^*. \]
Moreover, the maximum in the definition of $\dnorm{\abs{\Phi}}^*$ is achieved by a completely positive and trace-preserving super-operator.

As a consequence, if $\Phi$ is completely positive then
the maximum in the definition of $\dnorm{\Phi}^*$ is achieved by a completely positive and trace-preserving super-operator.

\end{theorem}

\begin{proof}

First we show that the maximum is attained by a completely positive super-operator---the trace-preserving property will be established later.

Let $\Psi$ be any super-operator with $\dnorm{\Psi}=1$ achieving the maximum in the definition of $\dnorm{\Phi}^*$, so that $\dnorm{\Phi}^*=\abs{\inner{\Psi}{\Phi}}$.
By Lemma \ref{lm:dnorm} there exists a decomposition $\Psi_\rL,\Psi_\rR$ of $\Psi$ with
\[ \dnorm{\Psi}=\dnorm{\Psi_\rL}=\dnorm{\Psi_\rR}=1. \]
Moreover, $\Psi_\rL$ and $\Psi_\rR$ are completely positive.

As $\Phi$ is Hermitian-preserving, there is a decomposition $\Phi_\rL,\Phi_\rR$ of $\Phi$ with $\Phi_\rL=\Phi_\rR=\abs{\Phi}$.
By Lemma \ref{lemma:cauchy} we have
\begin{align*}
\dnorm{\Phi}^*
&= \abs{\inner{\Psi}{\Phi}}
\leq \sqrt{\inner{\Psi_\rL}{\abs{\Phi}}}\cdot\sqrt{\inner{\Psi_\rR}{\abs{\Phi}}} \\
&\leq \max \Set{\inner{\Xi}{\abs{\Phi}} : \textrm{ $\dnorm{\Xi}=1$ and $\Xi$ is completely positive}}.
\end{align*}
Thus, $\dnorm{\Phi}^*\leq\dnorm{\abs{\Phi}}^*$ and the maximum in the definition of $\dnorm{\abs{\Phi}}^*$ is attained by a completely positive super-operator $\Xi$.

For the trace-preserving property, we note that $1=\dnorm{\Xi}=\snorm{\Xi}{1}$.
Then by Proposition \ref{prop:unit-ball} and the complete positivity of $\Xi$ it holds that $\jam{\Xi}\preceq\jam{\Xi'}$ for some completely positive and trace-preserving super-operator $\Xi'$ with $\dnorm{\Xi'}=1$.
The desired result follows from the observation that $\inner{\Xi'}{\abs{\Phi}}\geq\inner{\Xi}{\abs{\Phi}}$.
\end{proof}

\begin{theorem}[Achieving the maximum for Hermitian-preserving super-operators]
\label{thm:dnorm-dual-herm}

For any Hermitian-preserving super-operator $\Phi$
the maximum in the definition of $\dnorm{\Phi}^*$ is achieved by a Hermitian-preserving super-operator.

\end{theorem}

\begin{proof}

Let $\jam{\Phi}=T^+-T^-$ be a Jordan decomposition of $\jam{\Phi}$, so that $\jam{\abs{\Phi}}=T^++T^-$.
By Theorem \ref{theorem:cp-dnorm} it holds that $\dnorm{\Phi}^*\leq\dnorm{\abs{\Phi}}^*$ and the maximum in the definition of $\dnorm{\abs{\Phi}}^*$ is achieved by a completely positive and trace-preserving super-operator $\Psi$.

Let $\Pi^\pm$ denote the projections onto the support of $T^\pm$ and let $\Xi$ denote the Hermitian-preserving super-operator with
\[ \jam{\Xi} = \Pi^+\jam{\Psi}\Pi^+ - \Pi^-\jam{\Psi}\Pi^-. \]
It is easily verified that
\[\inner{\Psi}{\abs{\Phi}} = \inner{\Xi}{\Phi}.\]
The left side of this equality is $\dnorm{\abs{\Phi}}^*$, which we know
to be at least as large as $\dnorm{\Phi}^*$.
Hence, the desired result will follow once we establish that the Hermitian-preserving super-operator $\Xi$ has $\dnorm{\Xi}\leq 1$.

Toward that end, we note that
\[ \abs{\jam{\Xi}} = \Pi^+\jam{\Psi}\Pi^+ + \Pi^-\jam{\Psi}\Pi^- \preceq \jam{\Psi}. \]
As $\Psi$ is completely positive and trace-preserving, it holds that $\jam{\Psi}$ denotes a one-round non-measuring strategy.
Thus, by Proposition \ref{prop:unit-ball} it holds that $\snorm{\Xi}{1}\leq 1$.
The theorem follows from Proposition \ref{prop:dnorm-gen}, which tells us that $\dnorm{\Xi}=\snorm{\Xi}{1}$.
\end{proof}

\subsubsection{Not achieving the maximum}

We now show that Theorems \ref{theorem:cp-dnorm} and \ref{thm:dnorm-dual-herm} cannot be extended beyond completely positive and Hermitian-preserving super-operators, respectively.
Our counterexamples rely upon the following lemma.

\begin{lemma}[Diamond norm for 1-dimensional spaces]
\label{lemma:dnorm:1d}

The following hold for any super-operator $\Phi:\lin{\cX}\to\lin{\cY}$:
\begin{enumerate}
\item \label{item:lemma:dim1:tnorm}
Suppose $\dim(\cX)=1$ and let $A\in\lin{\cY}$ be the operator with $\Phi:\alpha\mapsto\alpha A$.
It holds that $\dnorm{\Phi} = \tnorm{\Phi} = \tnorm{A}$.

\item \label{item:lemma:dim1:norm}
Suppose $\dim(\cY)=1$ and let $A\in\lin{\cX}$ be the operator with $\Phi:X\mapsto\inner{A}{X}$.
It holds that $\dnorm{\Phi} = \tnorm{\Phi} = \norm{A}$.

\end{enumerate}

\end{lemma}

\begin{proof}

Item \ref{item:lemma:dim1:tnorm} is a simple consequence of the fact that the auxiliary space $\cW$ in the definition of the diamond norm can be assumed to have dimension no larger than $\dim(\cX)$.
This fact immediately implies $\dnorm{\Phi}=\tnorm{\Phi}$.
It follows immediately from the definition of the super-operator trace norm that $\tnorm{\Phi}=\tnorm{A}$.

For item \ref{item:lemma:dim1:norm}, choose a space $\cW$ and vectors $u,v\in\cX\ot\cW$ with
\[ \dnorm{\Phi} = \Tnorm{\Pa{\Phi\ot\idsup{\cW}}(uv^*)}. \]
Let $U\in\lin{\cW}$ be a unitary such that
\[ \Tnorm{\Pa{\Phi\ot\idsup{\cW}}(uv^*)} = \Ptr{}{U\Pa{\Phi\ot\idsup{\cW}}(uv^*)} \]
and let $w\in\cX\ot\cW$ be the vector given by $w=(I_\cX\ot U)u$, so that
\[ \Tnorm{\Pa{\Phi\ot\idsup{\cW}}(uv^*)} = \Ptr{}{\Pa{\Phi\ot\idsup{\cW}}(wv^*)}. \]
As in the proof of Lemma \ref{lemma:cp-dnorm}, choose an orthonormal basis $\set{e_i}$ of $\cW$ and write
\[
  w=\sum_{i=1}^{\dim(\cW)} \sqrt{q_i} w_i\ot e_i, \quad
  v=\sum_{i=1}^{\dim(\cW)} \sqrt{p_i} v_i\ot e_i
\]
for some choice of unit vectors $\set{w_i},\set{v_i}\subset\cX$ and probability distributions $\set{q_i},\set{p_i}$.
Then
\[
  \Ptr{}{\Pa{\Phi\ot\idsup{\cW}}(wv^*)}
  = \sum_{i=1}^{\dim(\cW)} \sqrt{p_i}\sqrt{q_i} \Inner{A}{w_iv_i}
  \leq \norm{A} \sum_{i=1}^{\dim(\cW)} \sqrt{p_i}\sqrt{q_i}
  \leq \norm{A}
\]
where the final inequality is Cauchy-Schwarz.
Equality is obtained for real numbers $p_1=q_1=1$ and unit vectors $w_1,v_1$ maximizing the inner product $\inner{A}{w_1v_1}$.
We have thus established $\dnorm{\Phi}=\norm{A}$.
As $\tnorm{\Phi(w_1v_1^*)}=\norm{A}$, we also have $\dnorm{\Phi}=\tnorm{\Phi}$.
\end{proof}

%
%

\begin{proposition}[Theorems \ref{theorem:cp-dnorm} and \ref{thm:dnorm-dual-herm} do not extend]

For each of the two maxima
\[
  \dnorm{\Phi}^* = \max_{\dnorm{\Psi}=1} \Abs{\Inner{\Psi}{\Phi}},
  \quad 
  \dnorm{\Phi} = \max_{\dnorm{\Psi}^*=1} \Abs{\Inner{\Psi}{\Phi}}
\]
there exist super-operators $\Phi$ such that the maximum is not attained by any Hermitian-preserving super-operator $\Psi$.

Moreover, there exist Hermitian-preserving super-operators $\Phi$ such that the maximum is not attained by any completely positive super-operator $\Psi$.

\end{proposition}

\begin{proof}

The counterexamples presented here are all achieved via reduction from the diamond norm and its dual to the trace norm and its dual (the operator norm).

Let $\cX$ be a space of dimension one, so that for each super-operator $\Phi:\lin{\cX}\to\lin{\cY}$ there is an operator $A\in\lin{\cY}$ with $\Phi:\alpha\mapsto\alpha A$.
For any $\Psi:\alpha\mapsto\alpha B$ it is easily verified that
\[\Abs{\Inner{\Psi}{\Phi}}=\Abs{\Inner{B}{A}}.\]
Moreover, it is clear that $\Psi$ is Hermitian-preserving if and only if $B$ is Hermitian and that $\Psi$ is completely positive if and only if $B$ is positive semidefinite.
By Lemma \ref{lemma:dnorm:1d} it holds that
\begin{equation}
  \dnorm{\Phi}^*
  = \max_{\dnorm{\Psi}=1} \Abs{\Inner{\Psi}{\Phi}}
  = \max_{\tnorm{B}=1} \Abs{\Inner{B}{A}}
  = \norm{A}.
\end{equation}
We now exhibit an operator $A$ such that the maximum is not achieved by any Hermitian operator $B$.
Consider the operator
\[ A = \left( \begin{array}{cc} 1&1\\0&1 \end{array} \right). \]
The quantity \[ \max_{\substack{\tnorm{B}=1,\\\textrm{$B$ Hermitian}}} \abs{\inner{B}{A}} = \max_{\norm{v}=1} \abs{v^*Av} \] is known as the \emph{numerical radius} $\varrho(A)$ of $A$.
It is easy to compute $\varrho(A)=\frac{3}{2}$, which is strictly smaller than $\norm{A}=\frac{1+\sqrt{5}}{2}$.

Next, let $\cY$ be a space of dimension one, so that for each super-operator $\Phi:\lin{\cX}\to\lin{\cY}$ there is an operator $A\in\lin{\cY}$ with $\Phi:X\mapsto\inner{A}{X}$.
For any $\Psi:X\mapsto\inner{B}{X}$ it is easily verified that
\[\Abs{\Inner{\Psi}{\Phi}}=\Abs{\Inner{B}{A}}.\]
Moreover, it is clear that $\Psi$ is Hermitian-preserving if and only if $B$ is Hermitian and that $\Psi$ is completely positive if and only if $B$ is positive semidefinite.
By Lemma \ref{lemma:dnorm:1d} it holds that
\begin{equation}
  \dnorm{\Phi}^*
  = \max_{\dnorm{\Psi}=1} \Abs{\Inner{\Psi}{\Phi}}
  = \max_{\norm{B}=1} \Abs{\Inner{B}{A}}
  = \tnorm{A}.
\end{equation}
Consider the operator \[ A=\Pi_0-\Pi_1 \] where $\Pi_0,\Pi_1\in\pos{\cY}$ are nonzero orthogonal projections.
It is easy to verify that the maximum is not achieved by any positive semidefinite $B$.

The proposition is now proved for the maximum
\( \dnorm{\Phi}^* = \max_{\dnorm{\Psi}=1} \Abs{\Inner{\Psi}{\Phi}}. \)
The proof for the other maximum follows along similar lines.
\end{proof}

\subsubsection{The dual of the strategy $1$-norm agrees with the dual of the diamond norm}

Finally, we have what we need to establish agreement between $\snorm{\cdot}{1}^*$ and $\dnorm{\cdot}^*$ for Hermitian-preserving super-operators.

\def\propdnormgendual{Agreement with the dual of the diamond norm}
\begin{proposition}[\propdnormgendual]
\label{prop:dnorm-gen-dual}

For every Hermitian-preserving super-operator $\Phi$ it holds that \( \dnorm{\Phi}^*=\snorm{\Phi}{1}^* \).

\end{proposition}

\begin{proof}

We have
\begin{align*}
\dnorm{\Phi}^*
&= \max \Set{ \abs{\inner{\Psi}{\Phi}} : \textrm{ $\dnorm{\Psi}=1$ and $\Psi$ is Hermitian-preserving}} \\
&= \max \Set{ \abs{\inner{\Psi}{\Phi}} : \textrm{ $\snorm{\Psi}{1}=1$ and $\Psi$ is Hermitian-preserving}} \\
&= \snorm{\Phi}{1}^*.
\end{align*}
The first equality is Theorem \ref{thm:dnorm-dual-herm}, the second Proposition \ref{prop:dnorm-gen}, and the third Proposition \ref{prop:duality}.
\end{proof}

\part{Local Operations with Shared Entanglement} \label{part:LOSE}

\chapter{Introduction to Local Operations} \label{ch:intro-LO}


In Part \ref{part:LOSE} of this thesis we are interested in various classes of ``local'' quantum operations, which are operations that can be jointly implemented
by two or more parties who act on distinct portions of the input system and who do not communicate once they receive their portions of the input.

In this introductory chapter for Part \ref{part:LOSE} we provide formal definitions and immediate observations for three such classes: local operations, local operations with shared randomness, and local operations with shared entanglement.
We also introduce convenient shorthand notations for separable operators and for the vector space spanned by the local operations---objects which are of fundamental importance to the work in this part of the thesis.

\subsubsection{Local operations and the space that they span}

The simplest examples of quantum operations that can be implemented jointly by multiple parties without communication are the product operations,
which consist of several completely independent quantum operations juxtaposed and viewed as a larger single operation.

\begin{definition}[Local operation]
  A quantum operation $\Lambda:\lin{\kprod{\cX}{1}{m}}\to\lin{\kprod{\cY}{1}{m}}$ is \emph{$m$-party local} with respect to the input partition $\cX_1,\dots,\cX_m$ and output partition $\cY_1,\dots,\cY_m$ if it can be written as a product operation of the form
  \[ \Lambda = \Psi_1\ot\cdots\ot\Psi_m \]
  where each $\Psi_i:\lin{\cX_i}\to\lin{\cY_i}$ is a quantum operation.
  Typically, the partition of the input and output spaces is implicit and clear from the context.
\end{definition}

The vector space of super-operators spanned by the local operations is of paramount interest in Part \ref{part:LOSE} of this thesis.
As such, we introduce a convenient shorthand notation for this space.

Toward that end, recall that a super-operator $\Psi_i:\lin{\cX_i}\to\lin{\cY_i}$ denotes a quantum operation if and only if $\Psi_i$ is completely positive and trace-preserving.
In terms of Choi-Jamio\l kowski representations, the trace-preserving property of $\Psi_i$ is characterized by the condition $\ptr{\cY_i}{\jam{\Psi_i}}=I_{\cX_i}$.
The set of all operators $X$ obeying the inhomogeneous linear condition $\ptr{\cY_i}{X}=I_{\cX_i}$ is not a vector space,
but this set is easily extended to a unique smallest vector space by including its closure under multiplication by real scalars.

\begin{definition}[Shorthand notation for the space spanned by local operations]
\label{def:tpspace}
  For complex Euclidean spaces $\cX_i,\cY_i$, let
  \[
    \bQ_i \defeq
    \Set{
      X\in\her{\cY_i\otimes\cX_i} : \ptr{\cY_i}{X}=\lambda I_{\cX_i}
      \textrm{ for some } \lambda\in\mathbb{R}
    }
  \]
  denote the subspace of the real vector space $\her{\cY_i\ot\cX_i}$ of Hermitian operators $X$ of the form $X=\jam{\Psi}$ for which $\Psi:\lin{\cX_i}\to\lin{\cY_i}$ is a trace-preserving super-operator, or a scalar multiple thereof.
  (Throughout this thesis, the spaces $\cX_i,\cY_i$ are implicit whenever the notation $\bQ_i$ is used.)
  
  Let $\cX_1,\dots,\cX_m$ and $\cY_1,\dots,\cY_m$ be complex Euclidean spaces.
  Employing our shorthand notation for Kronecker products, the subspace $\kprod{\bQ}{1}{m}$ of $\her{\kprod{\cY}{1}{m}\ot\kprod{\cX}{1}{m}}$ spanned by the $m$-party local operations is given by
  \[ \kprod{\bQ}{1}{m} = \spn\Set{X_1\ot\cdots\ot X_m : X_i\in\bQ_i}. \]
\end{definition}

By the end of Part \ref{part:LOSE}, we will see that the space $\kprod{\bQ}{1}{m}$ contains not only all the local operations, but also the local operations with shared randomness and with shared entanglement, as well as an even broader class of local quantum operations called ``no-signaling'' operations.
Moreover, in Chapter \ref{ch:no-sig} we will also see a converse result---that any quantum operation whose Choi-Jamio\l kowski representation lies inside $\kprod{\bQ}{1}{m}$ is necessarily a no-signaling operation.

\subsubsection{Local operations with shared randomness, separable operators}

Local (product) operations are not the only quantum operations admitted by our model.
In particular, there is nothing
to stop the parties from meeting ahead of time so as to prepare shared resources that might allow them to correlate their separate quantum operations.

Consider, for example, a convex combination $\sum_j p_j \Lambda_j$ of $m$-party local operations $\Lambda_j$.
Such a quantum operation is legal in our model because it can be implemented by $m$ parties who share prior knowledge of an integer $j$ sampled according to the probability distribution $p_j$.
In this case, the shared resource is randomness.


\begin{definition}[Local operation with shared randomness (LOSR)]
\label{def:LOSR}
  A quantum operation $\Lambda:\lin{\kprod{\cX}{1}{m}}\to\lin{\kprod{\cY}{1}{m}}$ is a \emph{$m$-party LOSR operation} if it can be written as a convex combination
  of $m$-party local operations.
\end{definition}

As a convex combination of products of positive semidefinite operators, the Choi-Jamio\l kowski representation of a LOSR operation belongs to a broader class of operators called ``separable'' operators.

\begin{definition}[Separable operator] \label{def:sep}

  Let $\bS_1,\dots,\bS_m$ be arbitrary spaces of Hermitian operators.
  An element $X$ of the product space
  $\kprod{\bS}{1}{m}$
  is said to be \emph{$(\bS_1;\dots;\bS_m)$-separable} if $X$ can be written as a convex combination of product operators of the form $P_1\otimes\cdots\otimes P_m$ where each $P_i\in\bS_i^+$ is positive semidefinite.
  The set of $(\bS_1;\dots;\bS_m)$-separable operators forms a cone inside $\Pa{\kprod{\bS}{1}{m}}^+$.
\end{definition}

Using this terminology, Definition \ref{def:LOSR} states that a quantum operation $\Lambda$ is a LOSR operation if and only if $\jam{\Lambda}$ is a $\Pa{\bQ_1;\dots;\bQ_m}$-separable operator.

Another important example of separable operators are the separable quantum states.
A state $\rho\in\pos{\kprod{\cX}{1}{m}}$ is called \emph{separable} if $\rho$ can be written as a convex combination of product states of the form $\sigma_1\ot\cdots\ot\sigma_m$ where each $\sigma_i\in\pos{\cX_i}$ is a density operator.
In other words, $\rho$ denotes a separable state if and only if $\rho$ is a $\Pa{\her{\cX_1};\dots;\her{\cX_m}}$-separable operator with trace equal to one.
(Quantum states that are not separable are called \emph{entangled}.)

\subsubsection{Local operations with shared entanglement}

LOSR operations are not the only quantum operations that can be implemented locally without communication.
In the most general case, the parties could each hold a portion of some distinguished quantum state $\sigma$;
each party produces his output by applying some quantum operation to his portions of the input system and $\sigma$.
In this case, the shared resource is $\sigma$.

\begin{definition}[Local operation with shared entanglement (LOSE)]
  A quantum operation $\Lambda:\lin{\kprod{\cX}{1}{m}}\to\lin{\kprod{\cY}{1}{m}}$ is a
  \emph{$m$-party LOSE operation with finite entanglement}
  if there exist
  spaces $\cE_1,\dots,\cE_m$,
  a quantum state $\sigma\in\pos{\kprod{\cE}{1}{m}}$,
  and quantum operations $\Psi_i:\lin{\cX_i\ot\cE_i}\to\lin{\cY_i}$ for each $i=1,\dots,m$
  such that
  \[ \Lambda : X \mapsto (\kprod{\Psi}{1}{m})(X\ot\sigma). \]
  This arrangement is depicted for the two-party case in Figure \ref{fig:LO}.
  
  The operation $\Lambda$ is a \emph{finitely approximable $m$-party LOSE operation} if it lies in the closure of the set of $m$-party LOSE operations with finite entanglement.
  The term ``LOSE operation'' is used to refer to any finitely approximable LOSE operation; the restriction to finite entanglement is made explicit whenever it is required.
\end{definition}

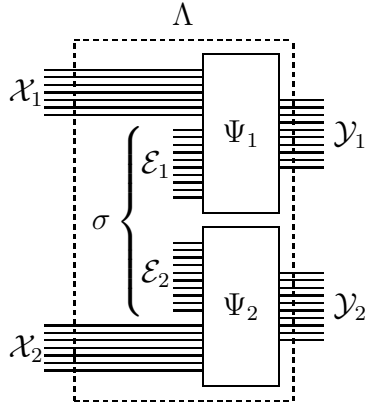
\begin{figure}[h]
  \hrulefill \vspace{2mm}
  \begin{center}
\setlength{\unitlength}{2072sp}%
\begin{picture}(4140,4617)(846,-4123)
\thinlines
\put(2881,-1051){\line( 1, 0){360}}
\put(2881,-1142){\line( 1, 0){360}}
\put(2881,-1231){\line( 1, 0){360}}
\put(2881,-1322){\line( 1, 0){360}}
\put(2881,-1411){\line( 1, 0){360}}
\put(2881,-961){\line( 1, 0){360}}
\put(2881,-871){\line( 1, 0){360}}
\put(2881,-1501){\line( 1, 0){360}}
\put(2881,-1591){\line( 1, 0){360}}
\put(2881,-1681){\line( 1, 0){360}}
\put(2881,-2401){\line( 1, 0){360}}
\put(2881,-2492){\line( 1, 0){360}}
\put(2881,-2581){\line( 1, 0){360}}
\put(2881,-2672){\line( 1, 0){360}}
\put(2881,-2761){\line( 1, 0){360}}
\put(2881,-2311){\line( 1, 0){360}}
\put(2881,-2221){\line( 1, 0){360}}
\put(2881,-2851){\line( 1, 0){360}}
\put(2881,-2941){\line( 1, 0){360}}
\put(2881,-3031){\line( 1, 0){360}}
\put(1341,-3391){\line( 1, 0){1900}}
\put(1341,-3481){\line( 1, 0){1900}}
\put(1341,-3571){\line( 1, 0){1900}}
\put(1341,-3661){\line( 1, 0){1900}}
\put(1341,-3751){\line( 1, 0){1900}}
\put(1341,-3301){\line( 1, 0){1900}}
\put(1341,-3211){\line( 1, 0){1900}}
\put(4141,-2941){\line( 1, 0){540}}
\put(4141,-3031){\line( 1, 0){540}}
\put(4141,-3121){\line( 1, 0){540}}
\put(4141,-3211){\line( 1, 0){540}}
\put(4141,-3301){\line( 1, 0){540}}
\put(4141,-2851){\line( 1, 0){540}}
\put(4141,-2761){\line( 1, 0){540}}
\put(4141,-2671){\line( 1, 0){540}}
\put(4141,-2581){\line( 1, 0){540}}
\put(4141,-3391){\line( 1, 0){540}}
\put(4141,-871){\line( 1, 0){540}}
\put(4141,-961){\line( 1, 0){540}}
\put(4141,-1051){\line( 1, 0){540}}
\put(4141,-1141){\line( 1, 0){540}}
\put(4141,-1231){\line( 1, 0){540}}
\put(4141,-781){\line( 1, 0){540}}
\put(4141,-691){\line( 1, 0){540}}
\put(4141,-601){\line( 1, 0){540}}
\put(4141,-511){\line( 1, 0){540}}
\put(4141,-1321){\line( 1, 0){540}}
\put(1341,-331){\line( 1, 0){1900}}
\put(1341,-421){\line( 1, 0){1900}}
\put(1341,-511){\line( 1, 0){1900}}
\put(1341,-601){\line( 1, 0){1900}}
\put(1341,-691){\line( 1, 0){1900}}
\put(1341,-241){\line( 1, 0){1900}}
\put(1341,-151){\line( 1, 0){1900}}
\put(3241,-3931){\framebox(900,1890){$\Psi_2$}}
\put(1701,-4111){\dashbox{57}(2620,4320){}}
\put(3241,-1861){\framebox(900,1890){$\Psi_1$}}
\put(2501,-1446){\makebox(0,0)[lb]{$\cE_1$}}
\put(2501,-2746){\makebox(0,0)[lb]{$\cE_2$}}
\put(4771,-1086){\makebox(0,0)[lb]{}}
\put(4771,-3156){\makebox(0,0)[lb]{}}
\put(911,-586){\makebox(0,0)[lb]{$\cX_1$}}
\put(911,-3646){\makebox(0,0)[lb]{$\cX_2$}}
\put(4800,-1086){\makebox(0,0)[lb]{$\cY_1$}}
\put(4800,-3146){\makebox(0,0)[lb]{$\cY_2$}}
\put(2871,389){\makebox(0,0)[lb]{$\Lambda$}}
\put(1911,-2060){$\sigma\left\{\rule[-11.25mm]{0mm}{0mm}\right.$}
\end{picture}%
  \end{center}
    \caption{A two-party local quantum operation $\Lambda$
    with shared entanglement.
    The local operations are represented by $\Psi_1,\Psi_2$;
    the shared entanglement by $\sigma$.}
    \label{fig:LO}
  \hrulefill
\end{figure}

\subsubsection{Finite and infinite shared entanglement}

The need to distinguish between LOSE operations with finite entanglement and finitely approximable LOSE operations arises from the fascinating fact that there exist LOSE operations that cannot be implemented with any finite amount of shared entanglement, yet can be approximated to arbitrary precision by LOSE operations with finite entanglement \cite{LeungT+08}.

Specifically, there exists an infinite sequence $\Lambda^{(1)},\Lambda^{(2)},\dots$ of LOSE operations with the property that the dimension of the space $\kprod{\cE^{(i)}}{1}{m}$ associated with the shared state of $\Lambda^{(i)}$ is ever increasing with $i$.
Moreover, this sequence is known to converge to a quantum operation that is not a LOSE operation with finite entanglement.

An analytic consequence of this fact is that the set of LOSE operations with finite entanglement is not a closed set.
Fortunately, the work of this thesis is unhindered by a more encompassing notion of LOSE operation that includes the closure of that set.

\subsubsection{Equivalence of shared randomness and shared separable states}

The following straightforward but tedious proposition asserts that a LOSR operation is merely a LOSE operation for which the shared state is $m$-partite separable.
It establishes that, as shared resources, a quantum state differs from randomness only when the state in question is entangled among the different parties.

Also established by this proposition is an explicit bound on the dimension of the shared state as a function of the dimensions of the input and output spaces.
Thus, while there exist LOSE operations that cannot be implemented with a finite shared state, every LOSR operation can be implemented with a finite shared state whose size scales favorably in the size of the input and output states.

\begin{proposition}[Equivalence of shared randomness and shared separable states]
\label{prop:shared-randomness}

  A quantum operation $\Lambda:\lin{\kprod{\cX}{1}{m}}\to\lin{\kprod{\cY}{1}{m}}$ is an $m$-party LOSR operation if and only if there exist
  \begin{enumerate}
  \item[(i)] spaces $\cE_1,\dots,\cE_m$ of dimension $d=\dim\Pa{\kprod{\bQ}{1}{m}}$,
  \item[(ii)] a $\Pa{\her{\cE_1};\dots;\her{\cE_m}}$-separable state $\sigma\in\pos{\kprod{\cE}{1}{m}}$, and
  \item[(iii)] quantum operations $\Psi_i:\lin{\cX_i\ot\cE_i}\to\lin{\cY_i}$
  \end{enumerate}
  such that \( \Lambda : X \mapsto (\kprod{\Psi}{1}{m})(X\ot\sigma). \)
\end{proposition}

\begin{proof}

Let $\sigma$ be a $\Pa{\her{\cE_1};\dots;\her{\cE_m}}$-separable state and let
$\Psi_i:\lin{\cX_i\ot\cE_i}\to\lin{\cY_i}$ be quantum operations such that
$\Lambda:X\mapsto(\kprod{\Psi}{1}{m})(X\ot\sigma)$.
Let
\[ \sigma=\sum_{j} p_j \sigma_{1,j}\ot\cdots\ot\sigma_{m,j} \]
be a decomposition of $\sigma$ into a convex combination of product states, where each $\sigma_{i,j}\in\pos{\cE_i}$.
For each $i$ and $j$ define a quantum operation
\[ \Phi_{i,j} : \lin{\cX_i}\to\lin{\cY_i} : X\mapsto\Psi_i\pa{X\ot\sigma_{i,j}} \]
and observe that
\[ \Lambda = \sum_{j} p_j \Phi_{1,j}\ot\cdots\ot\Phi_{m,j} \]
as desired.

Conversely, suppose that $\Lambda$ may be decomposed into a convex combination of product quantum operations as above.
By Carath\'eodory's Theorem (Fact \ref{fact:Carateodory}), this sum may be assumed to have no more than $d$ terms.
Let $\cE_1,\dots,\cE_m$ be complex Euclidean spaces of dimension $d$ and let
$\rho_{i,1},\dots,\rho_{i,d}\in\pos{\cE_i}$ be mutually orthogonal pure states for each $i$.
Let
\[ \sigma=\sum_{j=1}^d p_j \rho_{1,j} \ot\cdots\ot \rho_{m,j} \]
be a $\Pa{\her{\cE_1};\dots;\her{\cE_m}}$-separable state, let
\[
  \Psi_i:\lin{\cX_i\ot\cE_i}\to\lin{\cA_i} :
  X\ot E\mapsto\sum_{j=1}^d \inner{E}{\rho_{i,j}} \cdot \Phi_{i,j}\pa{X}
\]
be quantum operations, and observe that
\( \Lambda : X \mapsto \Pa{\kprod{\Psi}{1}{m}}\pa{X\ot\sigma} \).
\end{proof}

\subsubsection{Convexity of local operations with shared entanglement}

It is not difficult to see that the set of local operations is not a convex set.
By definition, the set of LOSR operations is convex.
That the set of LOSE operations with finite entanglement is also convex follows from a simple argument that is similar in principle to that of Proposition \ref{prop:shared-randomness}.

\begin{proposition}[Convexity of LOSE operations]

  The set of $m$-party LOSE operations with finite entanglement is convex.
  As a consequence, the set of finitely approximable $m$-party LOSE operations is also convex.

\end{proposition}

\begin{proof}

Let $\Lambda,\Lambda':\lin{\kprod{\cX}{1}{m}}\to\lin{\kprod{\cY}{1}{m}}$ be LOSE operations with finite entanglement.
Choose spaces $\cE_1,\dots,\cE_m$ of large enough dimension so that there exist states $\sigma,\sigma'\in\pos{\kprod{\cE}{1}{m}}$ and quantum operations
$\Psi_i,\Psi_i':\lin{\cX_i\ot\cE_i}\to\lin{\cY_i}$
such that
\begin{align*}
  \Lambda  &: X\mapsto(\kprod{\Psi}{1}{m})(X\ot\sigma), \\
  \Lambda' &: X\mapsto(\kprod{\Psi'}{1}{m})(X\ot\sigma').
\end{align*}
Let $\cF_1,\dots,\cF_m$ be two-dimensional spaces,
let $\rho_i,\rho_i'\in\cF_i$ be orthogonal pure states for each $i$, and
let
\begin{align*}
  \Phi_i
  &: \lin{\cX_i\ot\cE_i\ot\cF_i}\to\lin{\cY_i} \\
  &: X\ot E\ot F\mapsto \inner{F}{\rho_i} \cdot \Psi_i(X\ot E) + \inner{F}{\rho_i'} \cdot \Psi_i'(X\ot E)
\end{align*}
be quantum operations.
For positive real numbers $\alpha,\alpha'$ that sum to one, let
\[ \xi=\alpha\sigma\ot\kprod{\rho}{1}{m} + \alpha'\sigma'\ot\kprod{\rho'}{1}{m} \]
be a state in $\pos{\kprod{\cE}{1}{m}\ot\kprod{\cF}{1}{m}}$.
It is easy to verify that
\[ \alpha\Lambda + \alpha'\Lambda' : X \mapsto (\kprod{\Phi}{1}{m})(X\ot\xi), \]
implying that the convex combination $\alpha\Lambda + \alpha'\Lambda'$ is also a LOSE operation with finite entanglement.

That the set of finitely approximable LOSE operations is convex follows immediately from the fact that this set is the closure of the set of LOSE operations with finite entanglement, which we just showed to be convex.
\end{proof}




\chapter{Ball Around the Completely Noisy Channel} \label{ch:ball}

In this chapter we show that any quantum operation $\Lambda:\lin{\kprod{\cX}{1}{m}}\to\lin{\kprod{\cY}{1}{m}}$ for which $\jam{\Lambda}$ lies in the product space $\kprod{\bQ}{1}{m}$ and close enough to a distinguished ``completely noisy'' quantum operation $\tilde\noisy$ must necessarily be a LOSR operation.
In other words, there is a ball of LOSR operations surrounding $\tilde\noisy$.

For each input space $\cX$ and output space $\cY$, the \emph{completely noisy channel} $\tilde\noisy:\lin{\cX}\to\lin{\cY}$ is the unique quantum operation defined by \[ \tilde\noisy:X\mapsto\frac{\ptr{}{X}}{\dim(\cY)}I_\cY. \]
The quantum state $\frac{1}{\dim(\cY)}I_\cY$ is known as the \emph{completely mixed state} of the system associated with $\cY$.
Intuitively, this state denotes complete noise of the underlying system---it represents a uniform classical distribution over each of the $\dim(\cY)$ levels in the system.
Thus, the completely noisy channel is the quantum operation that always produces complete noise as output, regardless of the input.

The reason for interest in the completely noisy channel is that it is the unique super-operator whose Choi-Jamio\l kowski representation equals the identity.
More specifically, the completely noisy channel $\tilde\noisy:\lin{\cX}\to\lin{\cY}$ has
\[ \jam{\tilde\noisy} = \frac{1}{\dim(\cY)} I_{\cY\ot\cX}. \]
It is convenient for us to temporarily ignore the scalar multiple and deal directly with the unnormalized version $\noisy$ of $\tilde\noisy$ with \[ \jam{\noisy} = I_{\cY\ot\cX}.\]

The existence of a ball of LOSR operations around the completely noisy channel is established by proving that every operator in $\kprod{\bQ}{1}{m}$ and close enough to the identity must be a $\Pa{\bQ_1;\dots;\bQ_m}$-separable operator.
Indeed, this fact is shown to hold not only for the specific choice $\bQ_1,\dots,\bQ_m$ of subspaces, but for \emph{every} choice $\bS_1,\dots,\bS_m$ of subspaces that contain the identity.

Choosing $\bS_1,\dots,\bS_m$ to be full spaces of Hermitian operators yields an alternate (and simpler) proof of the existence of a ball of separable quantum states surrounding the completely mixed state and, consequently, of the NP-completeness of separability testing for quantum states.
However, the ball of separable states implied by the present work is not as large as that exhibited by Gurvits and Barnum \cite{GurvitsB02, GurvitsB03}.

The technical results we require are proven in Section \ref{sec:gen:sep}, as is a new bound for norms of super-operators.
The application of these results to establish the existence of the ball of LOSR operations is provided in Section \ref{sec:balls} along with some discussion of the details of this ball.

\section{General results on separable operators}
\label{sec:gen:sep}

Due to the general nature of the results in this chapter, discussion in the present section is abstract---applications to quantum information are deferred until Section \ref{sec:balls}.
The results presented herein were inspired by Chapter 2 of Bhatia~\cite{Bhatia07}.

\subsubsection{Hermitian subspaces generated by separable cones}

Let $\bS$ be any subspace of Hermitian operators that contains the identity.
The cone $\bS^+$ always \emph{generates} $\bS$, meaning that each element of
$\bS$ may be written as a difference of two elements of $\bS^+$.
As proof, choose any $X\in\bS$ and let
\[ X^\pm = \frac{ \norm{X} I \pm X }{2}. \]
It is clear that $X=X^+ - X^-$ and that $X^\pm\in\bS^+$
(using the fact that $I\in\bS$).
Moreover, it holds that $\norm{X^\pm}\leq\norm{X}$ for this particular choice of
$X^\pm$.

In light of this observation, one might wonder whether it could be extended in
product spaces to separable operators.
In particular, do the
$(\bS_1;\dots;\bS_m)$-separable operators generate the product space
$\kprod{\bS}{1}{m}$?
If so, can elements of $\kprod{\bS}{1}{m}$ be generated by
$(\bS_1;\dots;\bS_m)$-separable operators with bounded norm?
The following theorem answers these two questions in the affirmative.

\begin{theorem}[Generation via bounded separable operators] \label{thm:sep-bound}

  Let $\bS_1,\dots,\bS_m$ be subspaces of
  Hermitian operators---all of which contain the identity---and let
  $n=\dim\pa{\kprod{\bS}{1}{m}}$.
  Then every element $X\in\kprod{\bS}{1}{m}$ may be written
  $X = X^+ - X^-$ where $X^\pm$ are $(\bS_1;\dots;\bS_m)$-separable
  with \[\norm{X^\pm}\leq 2^{m-1}\sqrt{n}\fnorm{X}.\]

\end{theorem}

\begin{proof}

First, it is proven that there is an orthonormal basis $\mathbf{B}$ of
$\kprod{\bS}{1}{m}$ with the property that every element $E\in \mathbf{B}$
may be written $E = E^+ - E^-$ where $E^\pm$ are
$(\bS_1;\dots;\bS_m)$-separable
with $\norm{E^\pm}\leq 2^{m-1}$.
The proof is by straightforward induction on $m$.
The base case $m=1$ follows immediately from the earlier observation that
every element $X\in\bS_1$ is generated by some $X^\pm\in\bS_1^+$ with
$\norm{X^\pm}\leq \norm{X}$.
In particular, any element $E=E^+-E^-$ of any orthonormal basis of $\bS_1$ has
\[\norm{E^\pm}\leq\norm{E}\leq\fnorm{E}=1=2^0.\]

In the general case, the induction hypothesis states that there is
an orthonormal basis $\mathbf{B}'$ of $\kprod{\bS}{1}{m}$ with the desired
property.
Let $\mathbf{B}_{m+1}$ be any orthonormal basis of $\bS_{m+1}$.
As in the base case, each $F\in \mathbf{B}_{m+1}$ is generated by some
$F^\pm\in\bS_{m+1}^+$ with $\norm{F^\pm}\leq \norm{F} \leq 1$.
Define the orthonormal basis
$\mathbf{B}$ of $\kprod{\bS}{1}{m+1}$ to consist of all product operators
of the form $E\otimes F$ for $E\in \mathbf{B}'$ and $F\in \mathbf{B}_{m+1}$.
Define
\begin{align*}
  K^+ &\defeq \Br{E^+\otimes F^+} + \Br{E^-\otimes F^-},\\
  K^- &\defeq \Br{E^+\otimes F^-} + \Br{E^-\otimes F^+}.
\end{align*}
It is clear that $E\otimes F=K^+-K^-$ and $K^\pm$ are
$(\bS_1;\dots;\bS_{m+1})$-separable.
Moreover,
\[
  \Norm{K^+} \leq \Norm{E^+\otimes F^+} + \Norm{E^-\otimes F^-} \leq
  \Br{2^{m-1}\times 1} + \Br{2^{m-1}\times 1} = 2^m.
\]
A similar computation yields $\Norm{K^-}\leq 2^m$, which establishes the
induction.

Now, let $X\in\kprod{\bS}{1}{m}$ and let $x_j\in\mathbb{R}$ be the unique
coefficients of $X$ in the aforementioned orthonormal basis
$\mathbf{B}=\set{E_1,\dots,E_n}$.
Define
\begin{align*}
  X^+ &\defeq \sum_{j\::\:x_j>0} x_j E_j^+ - \sum_{j\::\:x_j<0} x_j E_j^-,\\
  X^- &\defeq \sum_{j\::\:x_j>0} x_j E_j^- - \sum_{j\::\:x_j<0} x_j E_j^+.
\end{align*}
It is clear that $X=X^+-X^-$ and $X^\pm$ are
$(\bS_1;\dots;\bS_m)$-separable.
Employing the triangle inequality and the vector norm inequality
$\norm{x}_1\leq\sqrt{n}\norm{x}_2$, it follows that
\[
  \Norm{X^+} \leq 2^{m-1} \sum_{j=1}^n \Abs{x_j} \leq
  2^{m-1} \sqrt{n} \sqrt{ \sum_{j=1}^n \Abs{x_j}^2 } =
  2^{m-1} \sqrt{n} \Fnorm{X}.
\]
A similar computation yields $\Norm{X^-}\leq 2^{m-1}\sqrt{n}\fnorm{X}$,
which completes the proof.
\end{proof}

\subsubsection{Separable operators as a subtraction from the identity}

At the beginning of this section it was observed that
$\norm{X}I-X$ lies in $\bS^+$ for all $X\in\bS$.
One might wonder whether more could be expected of $\norm{X}I-X$ than mere
positive semidefiniteness.
For example, under what conditions is $\norm{X}I-X$ a
$\pa{\bS_1;\dots;\bS_m}$-separable operator?
The following theorem provides three such conditions.
Moreover, the central claim of this section is established by this theorem.

\begin{theorem}[Ball of separable operators surrounding the identity] \label{thm:identity-sep}

Let $\bS_1,\dots,\bS_m$ be subspaces of Hermitian operators---all of which
contain the identity---and let $n=\dim\pa{\kprod{\bS}{1}{m}}$.
The following hold:
\begin{enumerate}

\item \label{item:identity-sep:1}
  \( \Norm{P}I - P \)
  is $\pa{\bS_1;\dots;\bS_m}$-separable whenever $P$ is a product operator
  of the form $P=P_1\otimes\cdots\otimes P_m$ where each $P_i\in\bS_i^+$.

\item \label{item:identity-sep:2}
  \( \Br{n+1}\Norm{Q}I - Q \) is
  $\pa{\bS_1;\dots;\bS_m}$-separable whenever $Q$ is
  $\pa{\bS_1;\dots;\bS_m}$-separable.

\item \label{item:identity-sep:3}
  \( 2^{m-1}\sqrt{n}\Br{n+1}\Fnorm{X} I - X \) is
  $\pa{\bS_1;\dots;\bS_m}$-separable for every
  $X\in\kprod{\bS}{1}{m}$.

\end{enumerate}

\end{theorem}

\begin{proof}
The proof of item \ref{item:identity-sep:1} is an easy but notationally
cumbersome induction on $m$.
The base case $m=1$ was noted at the beginning of Section \ref{sec:gen:sep}.
For the general case,
it is convenient to let $\kprod{I}{1}{m}$, $I_{m+1}$, and $\kprod{I}{1}{m+1}$
denote the identity elements of
$\kprod{\bS}{1}{m}$, $\bS_{m+1}$, and $\kprod{\bS}{1}{m+1}$, respectively.
The induction hypothesis states that the operator
\( S \defeq \norm{\kprod{P}{1}{m}} \kprod{I}{1}{m} - \kprod{P}{1}{m} \)
is $\pa{\bS_1;\dots;\bS_m}$-separable.
Just as in the base case, we know
\( S' \defeq \norm{P_{m+1}}I_{m+1}-P_{m+1} \)
lies in $\bS_{m+1}^+$.
Isolating the identity elements, these expressions may be rewritten
\begin{align*}
  \kprod{I}{1}{m} &=
    \frac{1}{ \Norm{\kprod{P}{1}{m}} } \Br{\kprod{P}{1}{m} + S} \\
  I_{m+1} &=
    \frac{1}{ \Norm{P_{m+1}} } \Br{ P_{m+1} + S' }.
\end{align*}
Taking the Kronecker product of these two equalities and rearranging the terms
yields
\[
  \Norm{\kprod{P}{1}{m+1}} \kprod{I}{1}{m+1} - \kprod{P}{1}{m+1} =
  \Br{ \kprod{P}{1}{m} \otimes S' } +
  \Br{ S \otimes P_{m+1} } +
  \Br{ S \otimes S' }.
\]
The right side of this expression is clearly a
$\pa{\bS_1;\dots;\bS_{m+1}}$-separable operator;
the proof by induction is complete.

Item \ref{item:identity-sep:2} is proved as follows.
By Carath\'eodory's Theorem (Fact \ref{fact:Carateodory}), every $(\bS_1;\dots;\bS_m)$-separable operator $Q$ may be written as a sum of no more than $n+1$ product operators.
In particular,
\[ Q = \sum_{j=1}^{n+1} P_{1,j}\otimes\cdots\otimes P_{m,j} \]
where each $P_{i,j}$ is an element of $\bS_i^+$.
As each term in this sum is positive semidefinite, it holds that
\( \norm{Q} \geq \norm{P_{1,j}\otimes\cdots\otimes P_{m,j}} \)
for each $j$.
Item \ref{item:identity-sep:1} implies that the sum
\[
  \sum_{j=1}^{n+1}
    \Norm{P_{1,j}\otimes\cdots\otimes P_{m,j}}I -
          P_{1,j}\otimes\cdots\otimes P_{m,j}
\]
is also $(\bS_1;\dots;\bS_m)$-separable.
Naturally, each of the identity terms
\[\Norm{P_{1,j}\otimes\cdots\otimes P_{m,j}}I\]
in this sum may be replaced by
$\norm{Q}I$ without compromising
$(\bS_1;\dots;\bS_m)$-separability,
from which it follows that
\( \br{n+1}\norm{Q}I-Q \) is also $(\bS_1;\dots;\bS_m)$-separable.

To prove item \ref{item:identity-sep:3}, apply Theorem \ref{thm:sep-bound} to obtain
$X=X^+-X^-$ where $X^\pm$ are $(\bS_1;\dots;\bS_m)$-separable with
\( \norm{X^\pm}\leq 2^{m-1}\sqrt{n}\fnorm{X}. \)
By item \ref{item:identity-sep:2}, it holds that
\( \br{n+1}\norm{X^+} I - X^+ \)
is $\pa{\bS_1;\dots;\bS_m}$-separable, implying that
\[ 2^{m-1}\sqrt{n}\Br{n+1}\Fnorm{X} I - X^+ \]
is also $\pa{\bS_1;\dots;\bS_m}$-separable.
To complete the proof, it suffices to note that adding $X^-$ to this operator
yields another $\pa{\bS_1;\dots;\bS_m}$-separable operator.
\end{proof}

\subsubsection{Norms of positive super-operators}

We are now in a position to prove new bounds on the capacity of a positive super-operator to increase the norm of its output relative to that of its input.
While these new bounds are not used later in this thesis, it is worthwhile to note them and compare them to previously known bounds of this nature.

The seminal such bound is due to Russo and Dye and states that \[ \norm{\Phi(X)}\leq \norm{\Phi(I)}\ \norm{X} \] for all $X$ whenever $\Phi$ is positive \cite{RussoD66}.
In particular, if $\Phi$ does not increase the operator norm of the identity then $\Phi$ does not increase the operator norm of \emph{any} operator.
When discussing bounds of this form, it is convenient to make the assumption that $\norm{\Phi(I)}\leq 1$, which allows us to state the bound more succinctly as \[ \norm{\Phi(X)}\leq\norm{X}. \]
Of course, the original bound, which applies to \emph{all} positive $\Phi$, can be recovered from this succinct version via suitable rescaling of $\Phi$ by a factor of $\norm{\Phi(I)}$.

Consider a subspace $\bS\subseteq\lin{\cX}$ that contains the identity.
In Bhatia \cite{Bhatia07} it is shown that \[ \norm{\Phi(X)}\leq \sqrt{2} \norm{X} \] for all $X\in\bS$ whenever $\Phi:\lin{\cX}\to\lin{\cY}$ is positive on $\bS^+$.
Essentially, the Russo-Dye bound weakens by a factor of $\sqrt{2}$ when we relax the positivity requirement to a mere subspace of the input space.
Gurvits and Barnum showed that \[ \norm{\Phi(X)}\leq \sqrt{2} \fnorm{X} \] for all $X$ whenever $\Phi:\lin{\cX}\to\lin{\cY}$ is positive on the unit ball \[ \set{ X : \fnorm{I_\cX-X}\leq 1 } \subset \lin{\cX} \]  surrounding the identity \cite{GurvitsB03}.
This bound is considerably weaker than the previous two because the operator norm $\norm{X}$ of $X$ is replaced with the Frobenius norm $\fnorm{X}$.
On the other hand, it has the advantage that $\Phi$ need not be positive on an entire subspace of $\lin{\cX}$.
Gurvits and Barnum used this additional flexibility to establish the existence of a ball of multipartite separable quantum states surrounding the completely mixed state.

To the above list of bounds we add one of our own.
While the Gurvits-Barnum bound applies only to super-operators $\Phi:\lin{\cX}\to\lin{\cY}$ that are positive on a set that spans all of $\lin{\cX}$ (namely, the unit ball around the identity), the super-operators to which our bound applies might only be positive on a set that does not span all of $\lin{\cX}$.
Alas, this flexibility comes at the cost of a multiple of the dimension of the input space.
For context, our bound is stated as an extension of
Theorem 2.6.3 of Bhatia~\cite{Bhatia07}.

\begin{theorem}[Norms of positive super-operators] \label{thm:map-bound}

  Let $\bS_1,\dots,\bS_m$ be subspaces of Hermitian operators---all of which
  contain the identity---and let $n=\dim\pa{\kprod{\bS}{1}{m}}$.
  Let $\Phi$ be a super-operator acting on $\kprod{\bS}{1}{m}$
  with the property that $\Phi$ is positive on
  $(\bS_1;\dots;\bS_m)$-separable operators.
  The following hold:
  \begin{enumerate}

  \item
    $\Norm{\Phi\pa{P}}\leq \Norm{P}\Norm{\Phi\pa{I}}$
    whenever $P$ is a product operator of the form
    $P=P_1\otimes\cdots\otimes P_m$ where each $P_i\in\bS_i^+$.

  \item
    $\Norm{\Phi\pa{Q}}\leq \Br{n+1}\Norm{Q}\Norm{\Phi\pa{I}}$
    whenever $Q$ is $\pa{\bS_1;\dots;\bS_m}$-separable.

  \item
    $\Norm{\Phi\pa{X}}\leq 2^{m-1}\sqrt{n}\Br{n+1}\Fnorm{X}\Norm{\Phi\pa{I}}$
    for every $X\in\kprod{\bS}{1}{m}$.

  \end{enumerate}

\end{theorem}

\begin{proof}

To prove item 1, we observe that
\[ \Phi\Pa{\norm{P}I-P} = \norm{P}\Phi\pa{I} - \Phi\pa{P} \]
is positive semidefinite.
As $\Phi(I)$ and $\Phi(P)$ are also positive semidefinite, it follows that
\( \Norm{\Phi\pa{P}} \leq \norm{P}\norm{\Phi\pa{I}}. \)
Item 2 is proved by the same argument with
$\br{n+1}\norm{Q}I-Q$ in place of $\norm{P}I-P$.

To prove item 3, apply Theorem \ref{thm:sep-bound} to obtain
$X=X^+-X^-$ where $X^\pm$ are
$(\bS_1;\dots;\bS_m)$-separable
with $\norm{X^\pm}\leq 2^{m-1}\sqrt{n}\fnorm{X}$.
As $\Phi\pa{X^\pm}$ are positive semidefinite, we have
\begin{align*}
  \Norm{\Phi\pa{X}} &= \Norm{\Phi\pa{X^+}-\Phi\pa{X^-}} \\
  &\leq \max\Pa{\Norm{\Phi\pa{X^+}},\Norm{\Phi\pa{X^-}}} \\
  &\leq \max\Pa{\Norm{X^+},\Norm{X^-}} \Br{n+1} \Norm{\Phi\pa{I}} \\
  &\leq 2^{m-1} \sqrt{n} \Br{n+1} \Fnorm{X} \Norm{\Phi\pa{I}}.
\end{align*}
\end{proof}

\section{Ball around the completely noisy channel} \label{sec:balls}

In Chapter \ref{ch:intro-LO} we noted that a quantum operation $\Lambda$ is a LOSR operation if and only if $\jam{\Lambda}$ is $\Pa{\bQ_1;\dots;\bQ_m}$-separable.
We then established via Theorem \ref{thm:identity-sep} of the previous section that any operator in the product space $\kprod{\bQ}{1}{m}$ and close enough to the identity is necessarily a $\Pa{\bQ_1;\dots;\bQ_m}$-separable operator, from which the existence of a ball of LOSR operations around the completely noisy channel follows.

\begin{theorem}[Ball around the completely noisy channel]
\label{thm:ball}

  Let $n=\dim(\kprod{\bQ}{1}{m})$ and let \(k = 2^{m-1}\sqrt{n}\Br{n+1}.\)
  For each operator $A\in\kprod{\bQ}{1}{m}$ with $\fnorm{A} \leq \frac{1}{k}$ there exists an unnormalized LOSR operation $\Lambda:\lin{\kprod{\cX}{1}{m}}\to\lin{\kprod{\cY}{1}{m}}$ for which
  \[\jam{\Lambda}=I-A.\]
  As a consequence, any quantum operation $\Xi:\lin{\kprod{\cX}{1}{m}}\to\lin{\kprod{\cY}{1}{m}}$ with \[ \fnorm{\jam{\Xi}-\jam{\tilde\noisy}}\leq\frac{1}{kd} \] is a LOSR operation.
  Here $\tilde\noisy$ denotes the completely noisy channel and $d=\dim\pa{\kprod{\cY}{1}{m}}$.
  
\end{theorem}




Theorem \ref{thm:ball} establishes a ball of LOSR operations (and hence also of LOSE operations) surrounding the completely noisy channel.
However, there seems to be no obvious way to obtain a bigger ball if
such a ball is allowed to contain operations that are LOSE but not LOSR.
Perhaps a more careful future investigation will uncover such a ball.

\subsubsection{The subspace containing the ball of LOSR operations}

The ball of Theorem \ref{thm:ball} is contained within the product space
$\kprod{\bQ}{1}{m}$,
which is a strict subspace of the space
spanned by all quantum operations
$\Phi:\lin{\kprod{\cX}{1}{m}} \to \lin{\kprod{\cY}{1}{m}}$.
Why was attention restricted to this subspace?
The answer is that there are no LOSE or LOSR operations $\Lambda$ for which
$\jam{\Lambda}$ lies outside $\kprod{\bQ}{1}{m}$.
In other words, $\kprod{\bQ}{1}{m}$ is the \emph{largest} possible space in which
to find a ball of LOSR operations.
We shall return to this topic in Chapter \ref{ch:no-sig} wherein it is shown that the space $\kprod{\bQ}{1}{m}$ is generated by the so-called \emph{no-signaling} quantum operations.

Of course, there exist quantum operations arbitrarily close to the
completely noisy channel that are not no-signaling operations,
much less LOSE or LOSR operations.
This fact might seem to confuse the study of, say, the effects of noise on such
operations because a completely general model of noise would allow for
extremely tiny perturbations that nonetheless turn no-signaling operations into
signaling operations.
This confusion might even be exacerbated by the fact that separable quantum
states, by contrast, are resilient to \emph{arbitrary} noise: any conceivable
physical perturbation of the completely mixed state is separable, so long as the
perturbation has small enough magnitude.

There is, of course, nothing unsettling about this picture.
In any reasonable model of noise, perturbations to a LOSE or LOSR operation
occur only on the local operations performed by the parties involved,
or perhaps on the state they share.
It is easy to see that realistic perturbations such as these always maintain the
no-signaling property of these operations.
Moreover, any noise \emph{not} of this form could, for example, bestow
faster-than-light communication upon spatially separated parties.

\chapter{Recognizing LOSE Operations is NP-hard} \label{ch:NPhard}

In this chapter the existence of the ball of LOSR operations around the completely noisy channel is employed to prove that the weak membership problem for LOSE operations is strongly $\cls{NP}$-hard.
Informally, the weak membership problem asks,
\begin{quote}
  ``Given a description of a quantum operation $\Lambda$
  and an accuracy parameter $\varepsilon$,
  is $\Lambda$ within distance $\varepsilon$ of a LOSE operation?''
\end{quote}

This result is achieved in several stages.
Section \ref{sec:app:games} reviews a relevant recent result of Kempe \emph{et al.}~pertaining to quantum games.
In Section \ref{sec:app:validity} this result is exploited in order to prove that the weak validity problem---a relative of the weak membership problem---is strongly $\cls{NP}$-hard for LOSE operations.
Finally, Section \ref{sec:app:membership} illustrates how the strong $\cls{NP}$-hardness of the weak membership problem for LOSE operations follows from a Gurvits-Gharibian-style application
of Liu's version
of the Yudin-Nemirovski\u\i{} Theorem.
It is also noted that similar $\cls{NP}$-hardness results hold trivially for LOSR
operations, due to the fact that separable quantum states arise as a special
case of LOSR operations in which the input space is empty.

\section{Co-operative quantum games with shared entanglement}
\label{sec:app:games}

Local operations with shared entanglement have been previously studied in the context of
two-player co-operative games.
In these games, a referee prepares a question for each player and the players each respond to the referee with an answer.
The referee evaluates these answers and declares that the players have jointly won or lost the game according to this evaluation.
The goal of the players, then, is to coordinate their answers so as to maximize the probability with which the referee declares them to be winners.
In a \emph{quantum} game the questions and answers are quantum states.

In order to differentiate this model from a one-player game, the players are not permitted to communicate with each other after the referee has sent his questions.
The players can, however, meet prior to the commencement of the game in order to agree on a strategy.
In a quantum game the players might also prepare a shared entangled quantum state so as to enhance the coordination of their answers to the referee.

More formally, a \emph{quantum game} $G=(q,\pi,\mathbf{R},\mathbf{V})$ is specified by:
\begin{itemize}
  \item
  A positive integer $q$ denoting the number of distinct questions.
  \item
  A probability distribution $\pi$ on the question indices $\{1,\dots,q\}$, according to which 
  the referee selects his questions.
  \item
  Complex Euclidean spaces $\cV,\cX_1,\cX_2,\cA_1,\cA_2$ corresponding to the different quantum systems used by the referee and players.
  \item
  A set $\mathbf{R}$ of quantum states
  $\mathbf{R} = \{ \rho_i \}_{i=1}^q \subset \pos{\cV\otimes\cX_1\otimes\cX_2}$.
  These states correspond to questions and are selected by the referee according to $\pi$.
  \item
  A set $\mathbf{V}$ of unitary operators $\mathbf{V} = \{V_i \}_{i=1}^q \subset
  \lin{\cV\otimes\cA_1\otimes\cA_2}$.
  These unitaries are used by the referee to evaluate the players' answers.
\end{itemize}
For convenience, the two players are called \emph{Alice} and \emph{Bob}.
The game is played as follows.
The referee samples $i$ according to $\pi$ and prepares the state
$\rho_i\in\mathbf{R}$,
which is placed in the three quantum registers corresponding to
$\cV\otimes\cX_1\otimes\cX_2$.
This state contains the questions to be sent to the players:
the portion of $\rho_i$ corresponding to $\cX_1$ is sent to Alice,
the portion of $\rho_i$ corresponding to $\cX_2$ is sent to Bob, and
the portion of $\rho_i$ corresponding to $\cV$ is kept by the referee as a
private workspace.
In reply, Alice sends a quantum register corresponding to $\cA_1$ to the referee,
as does Bob to $\cA_2$.
The referee then applies the unitary operation $V_i\in \mathbf{V}$
to the three quantum registers corresponding to $\cV\otimes\cA_1\otimes\cA_2$,
followed by a standard measurement
$\{\Pi_\mathrm{accept},\Pi_\mathrm{reject}\}$
that dictates the result of the game.

As mentioned at the beginning of this subsection,
Alice and Bob may not communicate once the game commences.
But they may meet prior to the commencement of the game in order prepare a
shared entangled quantum state $\sigma$.
Upon receiving the question register corresponding to $\cX_1$ from the referee,
Alice may perform any physically realizable quantum operation upon that register
and upon her portion of $\sigma$.
The result of this operation shall be contained in the quantum register
corresponding to $\cA_1$---this is the answer that Alice sends to the referee.
Bob follows a similar procedure to obtain his own answer register corresponding
to $\cA_2$.

For any game $G$, the \emph{value} $\omega(G)$ of $G$ is the supremum of the
probability with which the referee can be made to accept taken over all
strategies of Alice and Bob.

\begin{theorem}[Kempe \emph{et al.}~\cite{KempeK+07}]
  There is a fixed polynomial $p$ such that the following promise problem is
  $\cls{NP}$-hard under mapping (Karp) reductions:
  \begin{description}
  \item[Input.]
    A quantum game $G=(q,\pi,\mathbf{R},\mathbf{V})$.
    The distribution $\pi$ and the sets $\mathbf{R},\mathbf{V}$ are each given
    explicitly:
    for each $i=1,\dots,q$, the probability $\pi(i)$ is given in binary,
    as are the real and complex parts of each entry of the matrices
    $\rho_i$ and $V_i$.
  \item[Yes.] The value $\omega(G)$ of the game $G$ is 1.
  \item[No.]  The value $\omega(G)$ of the game $G$ is less than
    $1-\frac{1}{p(q)}$.
  \end{description}
\end{theorem}

\section{Strategies and weak validity}
\label{sec:app:validity}

Viewing the two players as a single entity, a quantum game may be seen as a two-message quantum interaction between the referee and the players---a message from the referee to the players, followed by a reply from the players to the referee.
The actions of the referee during such an interaction are completely specified by the parameters of the game.

In the language of Part \ref{part:strategies}, the game specifies a one-round measuring co-strategy for the referee represented by some positive semidefinite operators
\[
  R_\mathrm{accept},R_\mathrm{reject} \in
  \pos{\kprod{\cA}{1}{2}\ot\kprod{\cX}{1}{2}},
\]
which are easily computed given the parameters of the game.

In these games, the players implement a one-round non-measuring strategy compatible with $\{R_\mathrm{accept},R_\mathrm{reject}\}$ whose representation is given by a positive semidefinite operator
\[
  P \in
  \pos{\kprod{\cA}{1}{2}\ot\kprod{\cX}{1}{2}}.
\]
For any fixed strategy $P$ for the players, Theorem \ref{theorem:inner-product} tells us that the probability with which the players cause the referee to accept is given by the inner product
\[
  \Pr[\textrm{Players win with strategy $P$}] = \inner{R_\mathrm{accept}}{P}.
\]

In any game, the players combine to implement some physical operation
\( \Lambda:\lin{\kprod{\cX}{1}{2}}\to\lin{\kprod{\cA}{1}{2}}. \)
It is clear that a given super-operator $\Lambda$ denotes a legal strategy for the
players if and only if $\Lambda$ is a LOSE operation.
As the players implement a one-round non-measuring strategy, the representation $P$ of their strategy is given by $P=\jam{\Lambda}.$

Thus, the problem studied by Kempe \emph{et al.}~\cite{KempeK+07} of deciding
whether $\omega\pa{G}=1$ can be reduced via the formalism of strategies to an
optimization problem over the set of LOSE operations:
\[
  \omega\pa{G} = \sup_{\textrm{$\Lambda\in$ LOSE}}
    \inner{R_\mathrm{accept}}{\jam{\Lambda}}
  .
\]
The following theorem is thus proved.

\begin{theorem}
\label{thm:wval}

  The weak validity problem for the set of LOSE [LOSR] operations is
  strongly $\cls{NP}$-hard [$\cls{NP}$-complete] under mapping (Karp) reductions:
  \begin{description}
  \item[Input.]
    A Hermitian matrix $R$, a real number $\gamma$, and a positive
    real number $\varepsilon>0$.
    The number $\gamma$ is given explicitly in binary, as are the
    real and complex parts of each entry of $R$.
    The number $\varepsilon$ is given in unary, where $1^s$ denotes $\varepsilon=1/s$.
  \item[Yes.]
    There exists a LOSE [LOSR] operation $\Lambda$ such that
    $\inner{R}{\jam{\Lambda}} \geq \gamma + \varepsilon$.
  \item[No.]
    For every LOSE [LOSR] operation $\Lambda$ we have
    $\inner{R}{\jam{\Lambda}} \leq \gamma - \varepsilon$.
  \end{description}

\end{theorem}

\begin{remark}
\label{rem:sep-reduction}

The hardness result for LOSR operations follows from a simple reduction
from separable quantum states to LOSR operations:
every separable state may be written as a LOSR operation in which the input space has dimension one.
That weak validity for LOSR operations is in $\cls{NP}$
(and is therefore $\cls{NP}$-complete)
follows from the fact that all LOSR operations may be implemented with
polynomially-bounded shared randomness (Proposition \ref{prop:shared-randomness}).

\end{remark}

\section{The Yudin-Nemirovski\u\i{} Theorem and weak membership}
\label{sec:app:membership}

Having established that the weak \emph{validity} problem for LOSE operations
is strongly $\cls{NP}$-hard, the next step is to follow the leads of
Gurvits and Gharibian~\cite{Gurvits02,Gharibian08}
and apply Liu's version~\cite{Liu07} of the
Yudin-Nemirovski\u\i{} Theorem~\cite{YudinN76, GrotschelL+88}
in order to prove that the weak \emph{membership} problem for LOSE operations
is also strongly $\cls{NP}$-hard.

The Yudin-Nemirovski\u\i{} Theorem establishes an oracle-polynomial-time
reduction from the weak validity problem to the weak membership problem for any
convex set $C$ that satisfies certain basic conditions.
One consequence of this theorem is that if the weak validity problem for $C$ is
$\cls{NP}$-hard then the associated weak membership problem for $C$ is also
$\cls{NP}$-hard.
Although hardness under mapping reductions is preferred, any hardness result
derived from the Yudin-Nemirovski\u\i{} Theorem in this way is only guaranteed
to hold under more inclusive oracle reductions.

The basic conditions that must be met by the set $C$ in order for the
Yudin-Nemirovski\u\i{} Theorem to apply are
\begin{enumerateroman}
\item $C$ is bounded,
\item $C$ contains a ball, and
\item the size of the bound and the ball are polynomially related to the
dimension of the vector space containing $C$.
\end{enumerateroman}
It is simple to check these criteria against the set of LOSE operations.
For condition (i), an explicit bound is established by Proposition \ref{prop:explicit-bound}.
The remaining conditions then follow from Theorem \ref{thm:ball}.

The following theorem is now proved.
(As in Remark \ref{rem:sep-reduction},
the analogous result for LOSR operations follows from a straightforward
reduction from separable quantum states.)

\begin{theorem} \label{thm:wmem}

  The weak membership problem for the set of LOSE [LOSR] operations is
  strongly $\cls{NP}$-hard [$\cls{NP}$-complete] under oracle (Cook) reductions:
  \begin{description}

  \item[Input.]
    A Hermitian matrix $X\in\kprod{\bQ}{1}{m}$
    and a positive real number $\varepsilon>0$.
    The real and complex parts of each entry of $X$ are given explicitly in binary.
    The number $\varepsilon$ is given in unary, where $1^s$ denotes $\varepsilon=1/s$.

  \item[Yes.]
    $X=\jam{\Lambda}$ for some LOSE [LOSR] operation $\Lambda$.

  \item[No.]
    \( \norm{X-\jam{\Lambda}}\geq \varepsilon \)
    for every LOSE [LOSR] operation $\Lambda$.
    
  \end{description}

\end{theorem}

\chapter{Characterizations of Local Operations} \label{ch:char}

Characterizations of LOSE and LOSR operations are presented in this chapter.
These characterizations are reminiscent of the well-known characterizations of
bipartite and multipartite separable quantum states due to
Horodecki \emph{et al.}~\cite{HorodeckiH+96,HorodeckiH+01}.

Specifically, it is proven in Section \ref{sec:char:LOSR} that $\Lambda$ is a LOSR
operation if and only if $\varphi\pa{\jam{\Lambda}}\geq 0$ for every linear functional $\varphi$ that is positive on the cone of $\Pa{\bQ_1;\dots;\bQ_m}$-separable operators.
(This cone is defined in Definitions \ref{def:sep} and \ref{def:tpspace}.)
This characterization of LOSR operations is proved by a straightforward application of the
fundamental Separation Theorems of convex analysis (Facts \ref{fact:separation} and \ref{fact:herm-sep}).

More interesting is the characterization of LOSE operations presented in Section
\ref{sec:char:LOSE}:
$\Lambda$ is a LOSE operation if and only if
$\varphi\pa{\jam{\Lambda}}\geq 0$ for every linear functional $\varphi$ that is \emph{completely} positive on that \emph{same cone} of $\Pa{\bQ_1;\dots;\bQ_m}$-separable operators.
Prior to the present work, the notion of complete positivity was only ever
considered in the context wherein the underlying cone is the positive
semidefinite cone.
Indeed, what it even \emph{means} for a super-operator or functional to be
``completely'' positive on some cone other than the positive semidefinite cone
must be clarified before any nontrivial discussion can occur.


\section{Characterization of local operations with shared randomness}
\label{sec:char:LOSR}

The characterization of LOSR operations presented herein is an immediate
corollary of the following simple proposition.

\begin{proposition} \label{prop:LOSR}

Let $K\subset\mathbb{R}^n$ be any closed convex cone.
A vector $x\in\mathbb{R}^n$ is an element of $K$ if and only if
$\varphi\pa{x}\geq 0$ for every linear functional
$\varphi : \mathbb{R}^n \to \mathbb{R}$
that is positive on $K$.

\end{proposition}

\begin{proof}

The ``only if'' part of the proposition is immediate:
as $x$ is in $K$, any linear functional positive on $K$ must also be positive on
$x$.
For the ``if'' part of the proposition, suppose that
$x$ is not an element of $K$.
The Separation Theorem (Fact \ref{fact:separation}) implies that
there exists a vector $h\in\mathbb{R}^n$
such that
$\inner{h}{y} \geq 0$ for all $y\in K$, yet
$\inner{h}{x} < 0$.
(To apply Fact \ref{fact:separation}, consider a small open ball that is disjoint from $K$ and contains $x$.)
Let $\varphi : \mathbb{R}^n \to \mathbb{R}$
be the linear functional given by
$\varphi : z \mapsto \inner{h}{z}$.
It is clear that $\varphi$ is positive on $K$, yet $\varphi\pa{x}<0$.
\end{proof}

\begin{corollary}[Characterization of LOSR operations] \label{thm:LOSR}

  A quantum operation
  $\Lambda : \lin{\kprod{\cX}{1}{m}} \to \lin{\kprod{\cA}{1}{m}}$
  is an $m$-party LOSR operation if and only if
  $\varphi\pa{\jam{\Lambda}} \geq 0$ for every linear functional
  $\varphi : \lin{\kprod{\cA}{1}{m}\otimes\kprod{\cX}{1}{m}} \to \mathbb{C}$
  that is positive on the cone of $\Pa{\bQ_1;\dots;\bQ_m}$-separable operators.

\end{corollary}

\begin{proof}

In order to apply Proposition \ref{prop:LOSR}, it suffices to note the following:
\begin{itemize}
\item
The space
$\her{\kprod{\cA}{1}{m}\ot\kprod{\cX}{1}{m}}$
is isomorphic to $\mathbb{R}^n$ for
$n=\dim\pa{\kprod{\cA}{1}{m}\ot\kprod{\cX}{1}{m}}^2$.
\item
The $\Pa{\bQ_1;\dots;\bQ_m}$-separable operators form a closed convex cone within the space $\her{\kprod{\cA}{1}{m}\ot\kprod{\cX}{1}{m}}$.
\end{itemize}

While Proposition \ref{prop:LOSR} only gives the desired result for real
linear functionals
$\varphi: \her{\kprod{\cA}{1}{m}\ot\kprod{\cX}{1}{m}} \to \mathbb{R}$,
it is trivial to construct a complex extension functional
$\varphi':\lin{\kprod{\cA}{1}{m}\otimes\kprod{\cX}{1}{m}} \to \mathbb{C}$
that agrees with $\varphi$ on $\her{\kprod{\cA}{1}{m}\ot\kprod{\cX}{1}{m}}$.
\end{proof}

\section{Characterization of local operations with shared entanglement}
\label{sec:char:LOSE}

In order to state our characterization of LOSE operations, it is convenient to parameterize the shorthand notation $\bQ_i$ of Definition \ref{def:tpspace}, which denotes the subspace of Hermitian operators $X$ with $X=\jam{\Psi}$ for some trace-preserving super-operator $\Psi$, or a scalar multiple thereof.

\begin{definition}[Parameterization $\bQ_i(\cE_i)$ of the shorthand notation $\bQ_i$]
  For each $i=1,\dots,m$ and each complex Euclidean space $\cE_i$,
  let \[\bQ_i\pa{\cE_i}\subset\her{\cA_i\otimes\cX_i\otimes\cE_i}\]
  denote the subspace of operators $\jam{\Psi}$ for which
  \[\Psi:\lin{\cX_i\otimes\cE_i}\to\lin{\cA_i}\]
  is a trace-preserving super-operator, or a scalar multiple thereof.
  In particular, the parameter space $\cE_i$ is always tensored with the \emph{input} space of $\Psi$, as opposed to the output space.
\end{definition}

\begin{theorem}[Characterization of LOSE operations]
\label{thm:LOSE}

  A quantum operation
  \[\Lambda : \lin{\kprod{\cX}{1}{m}} \to \lin{\kprod{\cA}{1}{m}}\]
  is an $m$-party LOSE operation if and only if
  $\varphi\pa{\jam{\Lambda}} \geq 0$ for every linear functional
  \[\varphi : \lin{\kprod{\cA}{1}{m}\otimes\kprod{\cX}{1}{m}} \to \mathbb{C}\]
  with the property that the super-operator
  $\Pa{\varphi \otimes \idsup{\kprod{\cE}{1}{m}}}$
  is positive on the cone of
  $\Pa{\bQ_1\pa{\cE_1};\dots;\bQ_m\pa{\cE_m}}$-separable operators
  for all choices of complex Euclidean spaces $\cE_1,\dots,\cE_m$.

\end{theorem}

\begin{remark}
\label{rem:LOSE:cp}

The positivity condition of Theorem \ref{thm:LOSE} bears striking resemblance to
the familiar notion of complete positivity of a super-operator.
With this resemblance in mind, a linear functional $\varphi$ obeying
the positivity condition of Theorem \ref{thm:LOSE} is said to be
\emph{completely positive on the
$\Pa{\bQ_1\pa{\cE_1};\dots;\bQ_m\pa{\cE_m}}$-separable family of cones}.
In this sense, Theorem \ref{thm:LOSE} represents what seems,
to the knowledge of the author,
to be the first application of the notion of complete positivity to a family of
cones other than the positive semidefinite family.

Moreover, for any fixed choice of complex Euclidean spaces $\cE_1,\dots,\cE_m$ there exists a linear functional $\varphi$ for which the super-operator $\Pa{\varphi \otimes \idsup{\kprod{\cE}{1}{m}}}$ is positive on the cone of $\Pa{\bQ_1\pa{\cE_1};\dots;\bQ_m\pa{\cE_m}}$-separable operators,
and yet $\varphi$ is \emph{not} completely positive on this family of
cones.
This curious property is a consequence of the fact the set of LOSE
operations with finite entanglement is not a closed set.
By contrast, complete positivity (in the traditional sense) of a super-operator
$\Phi:\lin{\cX}\to\lin{\cA}$ is assured whenever $\Pa{\Phi\otimes\idsup{\cZ}}$ is
positive for a space $\cZ$ with dimension at least that of $\cX$.
(See, for example, Bhatia \cite{Bhatia07} for a proof of this fact.)

\end{remark}

The proof of Theorem \ref{thm:LOSE} employs the following helpful identity
involving the Choi-Jamio\l kowski representation for super-operators.
This identity is proven by straightforward calculation.

\begin{proposition}
\label{prop:identities}

  Let $\Psi:\lin{\cX\otimes\cE} \to \lin{\cA}$
  and let $Z\in\lin{\cE}$.
  Then the super-operator
  \(\Lambda : \lin{\cX} \to \lin{\cA} \)
  defined by
  \( \Lambda\pa{X} = \Psi \pa{X\otimes Z} \)
  for all $X$ satisfies
  \(
  \jam{\Lambda} =
  \Ptr{\cE}{ \Br{ I_{\cA\otimes\cX} \otimes Z^\trans } \jam{\Psi} }.
  \)

\end{proposition}

\begin{proof}

Let
$\set{e_1,\dots,e_{\dim\pa{\cX}}}$ and
$\set{f_1,\dots,f_{\dim\pa{\cE}}}$
denote the standard bases of $\cX$ and $\cE$, respectively,
and let $z_{k,l}\in\mathbb{C}$ be the coefficients of $Z$ in the basis
$\set{f_k f_l^* : k,l=1,\dots,\dim\pa{\cE}}$
of $\lin{\cE}$.
We have
\begin{align*}
  \Ptr{\cE}{ \Br{ I_{\cA\otimes\cX} \otimes Z^\trans } \jam{\Psi} } &=
  \sum_{i,j=1}^{\dim\pa{\cX}} \sum_{k,l=1}^{\dim\pa{\cE}}
  \Psi \Pa{e_i e_j^* \otimes f_k f_l^* } \otimes e_i e_j^* \cdot
  \underbrace{ \Tr{ Z^\trans f_k f_l^* } }_{\textrm{\normalsize $z_{k,l}$}} \\
  &=
  \sum_{i,j=1}^{\dim\pa{\cX}} \Psi
  \Pa{ e_i e_j^* \otimes \sum_{k,l=1}^{\dim\pa{\cE}} z_{k,l} f_k f_l^* }
  \otimes e_i e_j^* = \jam{\Lambda}.
\end{align*}
\end{proof}

The following technical lemma is also employed in the proof of
Theorem \ref{thm:LOSE}.

\begin{lemma} \label{lm:LOSE:ip}

  Let \( \Psi : \lin{\cX\otimes\cE} \to  \lin{\cA} \),
  let \( Z \in \lin{\cE}$,
  and
  let \( \varphi : \lin{ \cA\otimes\cX } \to \mathbb{C} \).
  Then the super-operator  \( \Lambda : \lin{\cX} \to \lin{\cA} \) defined by
  \( \Lambda\pa{X} = \Psi \Pa{ X\otimes Z } \) for all $X$ satisfies
  \[
    \varphi\pa{\jam{\Lambda}} = \Inner{ \overline{Z} }
    { \Pa{ \varphi \otimes \idsup{\cE} } \Pa{ \jam{ \Psi } } }.
  \]

\end{lemma}

\begin{proof}

Let $H$ be the unique operator satisfying
$\varphi\pa{X}=\inner{H}{X}$ for all $X$
and note that the adjoint
$\varphi^*:\mathbb{C}\to\lin{ \cA\otimes\cX }$ satisfies
$\varphi^*(1)=H$.
Then
\begin{align*}
  \Inner{ \overline{Z} }{
    \Pa{ \varphi \otimes \idsup{\cE} }
    \Pa{ \jam{ \Psi } }
  }
  &= \Inner{ \varphi^*\pa{1} \otimes \overline{Z} }{ \jam{ \Psi } }
  = \Inner{ H \otimes \overline{Z} }{ \jam{ \Psi } } \\
  &= \Inner{ H }{
      \Ptr{\cE}
        { \Br{ I_{ \cA\otimes\cX } \otimes Z^\trans } \jam{\Psi} }
    }
  = \varphi\pa{\jam{\Lambda}}.
\end{align*}
\end{proof}


\begin{proof}[Proof of Theorem \ref{thm:LOSE}]

For the ``only if'' part of the theorem,
let $\Lambda$ be any LOSE operation with finite entanglement and let
$\Psi_1,\dots,\Psi_m,\sigma$ be such that
\[\Lambda : X \mapsto \Pa{\kprod{\Psi}{1}{m}}\pa{X\otimes \sigma}.\]
Let $\varphi$ be any linear functional on
$\lin{\kprod{\cA}{1}{m}\otimes\kprod{\cX}{1}{m}}$
that satisfies the stated positivity condition.
Lemma \ref{lm:LOSE:ip} implies
\[
  \varphi\pa{\jam{\Lambda}} = \Inner{ \overline{\sigma} }{
    \Pa{\varphi \otimes \idsup{\kprod{\cE}{1}{m}} }
    \Pa{ \jam{\kprod{\Psi}{1}{m}} }
  } \geq 0.
\]
A standard continuity argument establishes the desired implication when
$\Lambda$ is a finitely approximable LOSE operation.

For the ``if'' part of the theorem, suppose that
$\Xi : \lin{\kprod{\cX}{1}{m}} \to \lin{\kprod{\cA}{1}{m}}$ is a
quantum operation that is \emph{not} a LOSE operation.
The Separation Theorem (Fact \ref{fact:herm-sep}) implies that there
is a Hermitian operator
$H\in\her{\kprod{\cA}{1}{m}\otimes\kprod{\cX}{1}{m}}$ such that
$\inner{H}{\jam{\Lambda}} \geq 0$ for all LOSE operations $\Lambda$, yet
$\inner{H}{\jam{\Xi}} < 0$.
Let \[\varphi : \lin{\kprod{\cA}{1}{m}\otimes\kprod{\cX}{1}{m}}\to\mathbb{C}\]
be the linear functional given by
$\varphi : X \mapsto \inner{H}{X}$.

It remains to verify that $\varphi$ satisfies the desired positivity condition.
Toward that end, let $\cE_1,\dots,\cE_m$ be arbitrary complex Euclidean spaces.
By convexity, it suffices to consider only those
$\Pa{\bQ_1\pa{\cE_1};\dots;\bQ_m\pa{\cE_m}}$-separable operators that are product
operators.
Choose any such operator and note that, up to a scalar multiple,
it has the form
$\jam{ \kprod{\Psi}{1}{m} }$ where each
$\Psi_i : \lin{\cX_i\otimes\cE_i} \to \lin{\cA_i}$ is a quantum operation.
The operator
\[
  \Pa{ \varphi \otimes \idsup{\kprod{\cE}{1}{m}} }
  \Pa{ \jam{ \kprod{\Psi}{1}{m} } }
\]
is positive semidefinite if and only if it has a nonnegative inner product with
every density operator in $\pos{\kprod{\cE}{1}{m}}$.
As $\sigma$ ranges over all such operators, so does $\overline{\sigma}$.
Moreover, every such $\sigma$---together with $\Psi_1,\dots,\Psi_m$---induces a
LOSE operation $\Lambda$ defined by
$\Lambda : X \mapsto \Pa{\kprod{\Psi}{1}{m}}\pa{X\otimes \sigma}$.
Lemma \ref{lm:LOSE:ip} and the choice of $\varphi$ imply
\[
0 \leq \varphi \pa{\jam{\Lambda}} =
\Inner{ \overline{\sigma} }{
  \Pa{\varphi \otimes \idsup{\kprod{\cE}{1}{m}} }
  \Pa{ \jam{\kprod{\Psi}{1}{m}} }
},
\]
and so $\varphi$ satisfies the desired positivity condition.
\end{proof}

\chapter{No-Signaling Operations} \label{ch:no-sig}

At the end of Chapter \ref{ch:ball} it was claimed that the product space $\kprod{\bQ}{1}{m}$ is spanned by Choi-Jamio\l kowski representations of no-signaling operations.
It appears as though this fact has yet to be noted explicitly in the literature,
so a proof is offered in this chapter.

More accurately, two characterizations of no-signaling operations are presented
in Section \ref{sec:no-sig:char}, each of which is expressed as a condition on
Choi-Jamio\l kowski representations of super-operators.
It then follows immediately that $\kprod{\bQ}{1}{m}$ is spanned by Choi-Jamio\l kowski representations of no-signaling operations.

Finally, Section \ref{sec:no-sig:counter-example} provides
an example of a so-called \emph{separable} no-signaling operation that is
not a LOSE operation, thus ruling out an easy ``short cut'' to the ball of LOSR
operations revealed in Theorem \ref{thm:ball}.

\section{Formal definition of a no-signaling operation}

Intuitively, a quantum operation $\Lambda$ is no-signaling if it cannot be used by
spatially separated parties to violate relativistic causality.
Put another way, an operation $\Lambda$ jointly implemented by several parties is
no-signaling if those parties cannot use $\Lambda$ as a ``black box'' to
communicate with one another.

In order to facilitate a formal definition for no-signaling operations,
the shorthand notation for Kronecker products from Chapter \ref{ch:intro} must be
extended:
if $K\subseteq\set{1,\dots,m}$ is an arbitrary index set with
$K=\set{k_1,\dots,k_n}$ then we write
\[ \cX_K \defeq \cX_{k_1} \otimes \cdots \otimes \cX_{k_n} \]
with the convention that $\cX_\emptyset = \mathbb{C}$.
As with the original shorthand, a similar notation also applies to operators, sets of operators, and
super-operators.
The notation $\overline{K}$ refers to the set of indices \emph{not} in $K$, so that $K$,$\overline{K}$ is a partition of $\set{1,\dots,m}$.

\begin{definition}[No-signaling operation] \label{def:no-sig}

A quantum operation
$\Lambda:\lin{\kprod{\cX}{1}{m}}\to\lin{\kprod{\cA}{1}{m}}$ is an
\emph{$m$-party no-signaling operation} if for each index set
$K\subseteq\set{1,\dots,m}$ we have
\[ \ptr{\cA_K}{\Lambda\pa{\rho}} = \ptr{\cA_K}{\Lambda\pa{\sigma}} \]
whenever
\[ \ptr{\cX_K}{\rho} = \ptr{\cX_K}{\sigma}. \]
\end{definition}

What follows is a brief argument that Definition \ref{def:no-sig} captures the meaning of a
no-signaling operation---a more detailed discussion of this condition can be
found in Beckman \emph{et al.}~\cite{BeckmanG+01}.
If $\Lambda$ is no-signaling and $\rho,\sigma$ are locally indistinguishable to a
coalition $K$ of parties
(for example, when \( \ptr{\cX_{\overline{K}}}{\rho} = \ptr{\cX_{\overline{K}}}{\sigma} \))
then clearly the members of $K$ cannot perform a measurement on their portion of
the output that might allow them to distinguish $\Lambda\pa{\rho}$ from
$\Lambda\pa{\sigma}$
(that is, \( \ptr{\cA_{\overline{K}}}{\Lambda\pa{\rho}} = \ptr{\cA_{\overline{K}}}{\Lambda\pa{\sigma}} \)).
For otherwise, the coalition $K$ would have extracted
information---a signal---from the other parties that would allow it to distinguish
$\rho$ from $\sigma$.

Conversely, if there exist input states $\rho,\sigma$ such that
\( \ptr{\cX_{\overline{K}}}{\rho} = \ptr{\cX_{\overline{K}}}{\sigma} \) and yet
\( \ptr{\cA_{\overline{K}}}{\Lambda\pa{\rho}} \neq \ptr{\cA_{\overline{K}}}{\Lambda\pa{\sigma}} \)
then there exists a measurement that allows the coalition $K$ to distinguish
these two output states with nonzero bias, which implies that signaling must
have occurred.

It is not hard to see that every LOSE operation is also a no-signaling
operation.
Conversely, much has been made of the fact that there exist no-signaling
operations that are not LOSE operations---this is so-called ``super-strong''
nonlocality, exemplified by the popular ``nonlocal box'' discussed in Section
\ref{sec:no-sig:counter-example}.

\section{Two characterizations of no-signaling operations}
\label{sec:no-sig:char}

In this section it is shown that the product space $\kprod{\bQ}{1}{m}$ is spanned
by Choi-Jamio\l kowski representations of no-signaling operations.
(Recall from Definition \ref{def:tpspace} that each $\bQ_i\subset\her{\cA_i\otimes\cX_i}$ denotes the subspace of Hermitian operators $\jam{\Phi}$ for which $\Phi:\lin{\cX_i}\to\lin{\cA_i}$ is a trace-preserving super-operator, or a scalar multiple thereof.)
Indeed, that fact is a corollary of the following characterizations of
no-signaling operations.

\begin{theorem}[Two characterizations of no-signaling operations]
\label{thm:char:no-signal}

Let
$\Lambda:\lin{\kprod{\cX}{1}{m}}\to\lin{\kprod{\cA}{1}{m}}$
be a quantum operation.
The following are equivalent:
\begin{enumerate}

\item \label{item:no-signal}
$\Lambda$ is a no-signaling operation.

\item \label{item:prod-space}
$\jam{\Lambda}$ is an element of $\kprod{\bQ}{1}{m}$.

\item \label{item:constraints}
For each index set $K\subseteq\set{1,\dots,m}$
there exists
an operator
$Q\in\pos{\cA_{\overline{K}}\ot\cX_{\overline{K}}}$
with
\( \ptr{\cA_K}{\jam{\Lambda}} = Q\otimes I_{\cX_K}. \)

\end{enumerate}

\end{theorem}

\begin{remark}

Membership in the set of no-signaling operations may be verified in
polynomial time by checking the linear constraints in
Item \ref{item:constraints} of Theorem \ref{thm:char:no-signal}.
While the number of such constraints grows exponentially with $m$, 
this exponential growth is not a problem because the number of parties
$m$ is always $O\pa{\log n}$ for $n=\dim\pa{\kprod{\bQ}{1}{m}}$.
(This logarithmic bound follows from the fact that each space $\bQ_i$ has dimension at
least two and the total dimension $n$ is the product of the dimensions of each
of the $m$ different spaces.)

\end{remark}

The partial trace condition of Item \ref{item:constraints} of
Theorem \ref{thm:char:no-signal} is quite plainly
suggested by Theorem \ref{thm:char} (\thmchar).
Moreover, essential components of the proofs presented for two of the three
implications claimed in Theorem \ref{thm:char:no-signal}
appear in a 2001 paper of
Beckman \emph{et al.}~\cite{BeckmanG+01}.
The following theorem, however, establishes the third implication and appears to
be new.

\begin{theorem} \label{thm:TP-constraints}

A Hermitian operator $X\in\her{\kprod{\cA}{1}{m}\otimes\kprod{\cX}{1}{m}}$ is in
$\kprod{\bQ}{1}{m}$ if and only if for each index set
$K\subseteq\set{1,\dots,m}$ there exists a Hermitian operator
$Q\in\her{\cA_{\overline{K}}\ot\cX_{\overline{K}}}$
with
\( \ptr{\cA_K}{X} = Q \otimes I_{\cX_K}. \)

\end{theorem}

\begin{proof}
The ``only if'' portion of the theorem is straightforward---only the
``if'' portion is proven here.
The proof proceeds by induction on $m$.
The base case $m=1$ is trivial.
Proceeding directly to the general case,
let $s=\dim\pa{\her{\cA_{m+1}\otimes\cX_{m+1}}}$ and
let $\set{E_1,\dots,E_s}$ be a basis of
$\her{\cA_{m+1}\otimes\cX_{m+1}}$.
Let $X_1,\dots,X_s\in\her{\kprod{\cA}{1}{m}\otimes\kprod{\cX}{1}{m}}$
be the unique operators satisfying
\[ X = \sum_{j=1}^s X_j \otimes E_j. \]
It shall be proven that $X_1,\dots,X_s\in\kprod{\bQ}{1}{m}$.
The intuitive idea is to exploit the linear independence of $E_1,\dots,E_s$ in
order to ``peel off'' individual product terms in the decomposition of $X$.

Toward that end, for each fixed index $j\in\set{1,\dots,s}$ let
$H_j$ be a Hermitian operator
for which the real number $\inner{H_j}{E_i}$ is nonzero only when
$i=j$.
Define a linear functional
$\varphi_j:E\mapsto\inner{H_j}{E}$
and note that
\[
\Pa{ \idsup{\kprod{\cA}{1}{m}\otimes\kprod{\cX}{1}{m} } \ot \varphi_j}
  \pa{X} =
\sum_{i=1}^s \varphi_j\pa{E_i} X_i = \varphi_j\pa{E_j} X_j.
\]
Fix an arbitrary partition $K,\overline{K}$ of the index set $\set{1,\dots,m}$.
By assumption,
\( \ptr{\cA_K}{X} = Q\otimes I_{\cX_K} \)
for some Hermitian operator $Q$.
Apply
$\trace_{\cA_K}$
to both sides of the above identity, then use the fact that
$\trace_{\cA_K}$
and $\varphi_j$ act upon different spaces
to obtain
\begin{align*}
& \varphi_j\pa{E_j} \ptr{\cA_K}{X_j} \\
={}& \Ptr{\cA_K}{
  \Pa{\idsup{ \kprod{\cA}{1}{m}\otimes\kprod{\cX}{1}{m} } \ot \varphi_j}
  \Pa{X}
} \\
={}& \Pa{ 
  \idsup{\cA_{\overline{K}} \otimes \kprod{\cX}{1}{n}}
  \ot \varphi_j
} \Pa{ \Ptr{\cA_K}{X} } \\
={}& \Pa{
  \idsup{\cA_{\overline{K}}\otimes \cX_{\overline{K}} }
  \ot \varphi_j
} \pa{ Q }
\otimes I_{\cX_K},
\end{align*}
from which it follows that $\ptr{\cA_K}{X_j}$ is a product operator of the form
$R \otimes I_{\cX_K}$ for some Hermitian operator $R$.
As this identity holds for all index sets $K$,
it follows from the induction hypothesis that
$X_j\in\kprod{\bQ}{1}{m}$ as desired.

Now, choose a maximal linearly independent subset
$\set{X_1,\dots,X_t}$ of $\set{X_1,\dots,X_s}$ and let
$Y_1,\dots,Y_t$ be the unique Hermitian
operators satisfying
\[ X = \sum_{i=1}^t X_i \otimes Y_i. \]
A similar argument shows $Y_1,\dots,Y_t\in\bQ_{m+1}$,
which completes the induction.
\end{proof}


\begin{proof}[Proof of Theorem \ref{thm:char:no-signal}]

\begin{description}

\item[Item \ref{item:constraints} implies item \ref{item:prod-space}.]

This implication follows immediately from Theorem \ref{thm:TP-constraints}.

\item[Item \ref{item:prod-space} implies item \ref{item:no-signal}.]

The proof of this implication borrows heavily from the proof of Theorem 2 in
Beckman \emph{et al.}~\cite{BeckmanG+01}.

Fix any partition $K,\overline{K}$ of the index set $\set{1,\dots,m}$.
Let $s=\dim\pa{\lin{\cX_{\overline{K}}}}$ and $t=\dim\pa{\lin{\cX_K}}$
and let $\set{\rho_1,\dots,\rho_s}$ and $\set{\sigma_1,\dots,\sigma_t}$
be bases of $\lin{\cX_{\overline{K}}}$ and $\lin{\cX_K}$, respectively,
that consist entirely of density operators.
Given any two operators $X,Y\in\lin{\kprod{\cX}{1}{m}}$ let
$x_{j,k},y_{j,k}\in\mathbb{C}$ be the unique coefficients of $X$ and $Y$
respectively in the product basis $\set{\rho_j\otimes \sigma_k}$.
Then
\( \ptr{\cX_K}{X} = \ptr{\cX_K}{Y} \)
implies
\[ \sum_{k=1}^t x_{j,k} = \sum_{k=1}^t y_{j,k} \]
for each fixed index $j=1,\dots,s$.

As $\jam{\Lambda}\in\kprod{\bQ}{1}{m}$, it is possible to write
\[
\jam{\Lambda} = \sum_{l=1}^n \jam{\Phi_{1,l}}\ot\cdots\ot\jam{\Phi_{m,l}}
\]
where $n$ is a positive integer and
$\Phi_{i,l}:\lin{\cX_i}\to\lin{\cA_i}$ satisfies
$\jam{\Phi_{i,l}}\in\bQ_i$ for each of the indices
$i=1,\dots,m$ and $l=1,\dots,n$.
In particular, as each $\Phi_{i,l}$ is (a scalar multiple of) a trace-preserving
super-operator, it holds that for each index $l=1,\dots,n$ there exists
$a_l\in\mathbb{R}$ with
$\tr{\Phi_{K,l}\pa{\sigma}}=a_l$ for all density operators $\sigma$.
Then
\begin{align*}
  \ptr{\cA_K}{\Lambda\pa{X}}
  &= \sum_{l=1}^n a_l \cdot \sum_{j=1}^s \Br{\sum_{k=1}^t x_{j,k}} \cdot
      \Phi_{\overline{K},l}\pa{\rho_j} \\
  &= \sum_{l=1}^n a_l \cdot \sum_{j=1}^s \Br{\sum_{k=1}^t y_{j,k}} \cdot
      \Phi_{\overline{K},l}\pa{\rho_j}
  = \ptr{\cA_K}{\Lambda\pa{Y}}
\end{align*}
as desired.

\item[Item \ref{item:no-signal} implies item \ref{item:constraints}.]

This implication is essentially a multi-party generalization of Theorem 8 in
Beckman \emph{et al.}~\cite{BeckmanG+01}.
The proof presented here differs from theirs in some interesting but
non-critical details.

Fix any partition $K,\overline{K}$ of the index set $\set{1,\dots,m}$.
To begin,  observe that
\[
\ptr{\cX_K}{X} = \ptr{\cX_K}{Y}
\implies
\ptr{\cA_K}{\Lambda\pa{X}} = \ptr{\cA_K}{\Lambda\pa{Y}}
\]
for \emph{all} operators
$X,Y\in\lin{\kprod{\cX}{1}{m}}$---not just density operators.
(This observation follows from the fact that $\lin{\kprod{\cX}{1}{m}}$ is
spanned by the density operators---a fact used in the above proof that
item \ref{item:prod-space} implies item \ref{item:no-signal}.)

Now, let $s=\dim\pa{\cX_{\overline{K}}}$ and $t=\dim\pa{\cX_K}$
and let
$\set{e_1,\dots,e_s}$ and $\set{f_1,\dots,f_t}$
be the standard bases of $\cX_{\overline{K}}$ and $\cX_K$ respectively.
If $c$ and $d$ are distinct indices in $\set{1,\dots,t}$
and $Z\in\lin{\cX_{\overline{K}}}$ is any operator then
\[
  \ptr{ \cX_K }{ Z \otimes f_cf_d^* } =
  Z \otimes \tr{ f_cf_d^* } =
  0_{\cX_{\overline{K}}} = \ptr{ \cX_K }{ 0_{\kprod{\cX}{1}{m}} }
\]
and hence
\[
  \Ptr{ \cA_K }{ \Lambda\Pa{ Z \otimes f_cf_d^*   } } =
  \Ptr{ \cA_K }{ \Lambda\Pa{ 0_{\kprod{\cX}{1}{m}} } } =
  \Ptr{ \cA_K }{ 0_{\kprod{\cA}{1}{m}} } = 0_{\cA_{\overline{K}}}.
\]
(Here a natural notation for the zero operator was used implicitly.)
Similarly, if $\rho\in\lin{\cX_K}$ is any density operator then
\[
\Ptr{ \cA_K }{ \Lambda\Pa{ Z \otimes f_cf_c^* } } =
\Ptr{ \cA_K }{ \Lambda\Pa{ Z \otimes \rho } }
\]
for each fixed index $c=1,\dots,t$.
Employing these two identities, one obtains
\begin{align*}
  \Ptr{ \cA_K }{ \jam{\Lambda} }
  &= \sum_{a,b=1}^s \sum_{c=1}^t \Ptr{ \cA_K }{
    \Lambda\Pa{e_ae_b^* \otimes f_cf_c^*}} \otimes \Br{e_ae_b^* \otimes f_cf_c^*} \\
  &= \sum_{a,b=1}^s  \Ptr{ \cA_K }{
    \Lambda\Pa{e_ae_b^* \otimes \rho}} \otimes
    e_ae_b^* \otimes \Br{\sum_{c=1}^tf_cf_c^*}
   = \jam{\Psi} \otimes I_{ \cX_K }
\end{align*}
where the quantum operation $\Psi$ is defined by
$\Psi:X\mapsto\ptr{ \cA_K }{\Lambda\Pa{X\otimes\rho}}$.
As $\jam{\Psi} \otimes I_{ \cX_K }$ is a product operator of the desired form,
the proof that item \ref{item:no-signal} implies item
\ref{item:constraints} is complete.

\end{description}
\end{proof}

\section{A separable no-signaling operation that is not a LOSE operation}
\label{sec:no-sig:counter-example}

Theorem \ref{thm:ball} establishes a ball of LOSR operations around the completely noisy channel.
Was the work of Chapter \ref{ch:ball} necessary, or might there be a simpler way to establish the same thing?
For example, suppose $\Lambda:\lin{\cX_1\otimes\cX_2}\to\lin{\cA_1\otimes\cA_2}$
is a quantum operation for which the operator $\jam{\Lambda}$
is $\Pa{\her{\cA_1\otimes\cX_1};\her{\cA_2\otimes\cX_2}}$-separable.
Quantum operations with separable Choi-Jamio\l kowski representations such as
this are called \emph{separable operations} \cite{Rains97}.
If an operation is both separable and no-signaling then must it always be a
LOSE operation, or even a LOSR operation?
An affirmative answer to this question would yield a trivial proof of Theorem \ref{thm:ball} that leverages existing knowledge of separable balls around the identity operator.

Alas, such a short cut is not to be had:
there exist no-signaling operations $\Lambda$ that are not LOSE operations, yet
$\jam{\Lambda}$ is separable.
One example of such an operation is the so-called ``nonlocal box'' discovered in
1994 by Popescu and Rohrlich~\cite{PopescuR94}.
This nonlocal box is easily formalized as a two-party no-signaling quantum
operation $\Lambda$, as in Ref.~\cite[Section V.B]{BeckmanG+01}.
That formalization is reproduced here.

Let \( \cX_1 = \cX_2 = \cA_1 = \cA_2 = \mathbb{C}^2 \),
let $\set{e_0,e_1}$ denote the standard bases of both $\cX_1$ and $\cA_1$, and
let $\set{f_0,f_1}$ denote the standard bases of both $\cX_2$ and $\cA_2$.
Write
\[ \rho_{ab} \defeq e_ae_a^*\otimes f_bf_b^* \]
for $a,b\in\set{0,1}$.
The nonlocal box $\Lambda:\lin{\kprod{\cX}{1}{2}}\to\lin{\kprod{\cA}{1}{2}}$
is defined by
\begin{align*}
\Set{ \rho_{00},\rho_{01},\rho_{10} }
&\stackrel{\Lambda}{\longmapsto}
\frac{1}{2} \Br{ \rho_{00} + \rho_{11} } \\
\rho_{11}
&\stackrel{\Lambda}{\longmapsto}
\frac{1}{2} \Br{ \rho_{01} + \rho_{10} }.
\end{align*}
Operators not in $\spn\set{\rho_{00},\rho_{01},\rho_{10},\rho_{11}}$ are
annihilated by $\Lambda$.
It is routine to verify that $\Lambda$ is a no-signaling operation, and
this operation $\Lambda$ is known not to be a LOSE operation~\cite{PopescuR94}.
To see that $\jam{\Lambda}$ is separable, write
\begin{align*}
  E_{a\to b} &\defeq e_be_a^* \\
  F_{a\to b} &\defeq f_bf_a^*
\end{align*}
for $a,b\in\set{0,1}$.
Then for all $X\in\lin{\kprod{\cX}{1}{2}}$ it holds that
\begin{align*}
\Lambda\pa{X}
&= \frac{1}{2}
\Big[ E_{0\to 0} \otimes F_{0\to 0} \Big] X
\Big[ E_{0\to 0} \otimes F_{0\to 0} \Big]^*
+  \frac{1}{2}
\Big[ E_{0\to 1} \otimes F_{0\to 1} \Big] X
\Big[ E_{0\to 1} \otimes F_{0\to 1} \Big]^* \\
&+ \frac{1}{2}
\Big[ E_{0\to 0} \otimes F_{1\to 0} \Big] X
\Big[ E_{0\to 0} \otimes F_{1\to 0} \Big]^*
+ \frac{1}{2}
\Big[ E_{0\to 1} \otimes F_{1\to 1} \Big] X
\Big[ E_{0\to 1} \otimes F_{1\to 1} \Big]^* \\
&+ \frac{1}{2}
\Big[ E_{1\to 0} \otimes F_{0\to 0} \Big] X
\Big[ E_{1\to 0} \otimes F_{0\to 0} \Big]^*
+ \frac{1}{2}
\Big[ E_{1\to 1} \otimes F_{0\to 1} \Big] X
\Big[ E_{1\to 1} \otimes F_{0\to 1} \Big]^* \\
&+ \frac{1}{2}
\Big[ E_{1\to 0} \otimes F_{1\to 1} \Big] X
\Big[ E_{1\to 0} \otimes F_{1\to 1} \Big]^*
+ \frac{1}{2}
\Big[ E_{1\to 1} \otimes F_{1\to 0} \Big] X
\Big[ E_{1\to 1} \otimes F_{1\to 0} \Big]^*,
\end{align*}
from which the
$\Pa{\her{\cA_1\otimes\cX_1};\her{\cA_2\otimes\cX_2}}$-separability of $\jam{\Lambda}$
follows.
It is interesting to note that the nonlocal box is
the smallest possible nontrivial example of such an operation---the number of
parties $m=2$ and the input spaces
$\cX_1,\cX_2$ and output spaces $\cA_1,\cA_2$ all have dimension two.

\chapter{Conclusion} \label{ch:conclusion}

In this thesis we discussed two distinct topics.
In Part \ref{part:strategies} we initiated the study of quantum strategies, which are complete descriptions of one party's actions in an interaction involving the exchange of multiple quantum messages among multiple parties.
We saw proofs of three important properties of strategies, and we saw applications of these properties to zero-sum quantum games, complexity theory, and quantum coin-flipping.
We also introduced a new norm for super-operators and argued that this norm, which generalizes the familiar diamond norm, captures the operational distinguishability of two quantum strategies in the same sense that the trace norm captures the distinguishability of two quantum states, or the diamond norm captures the distinguishability of two quantum operations.

In Part \ref{part:LOSE} we established several properties of local quantum operations, the implementation of which might be assisted by shared entanglement.
Specifically, we showed that every quantum operation sufficiently close to the completely noisy operation can be implemented locally using only shared randomness.
This fact was used to prove strong $\cls{NP}$-hardness of the weak membership problem for local operations with shared entanglement.
We then provided algebraic characterizations of the sets of local operations with shared randomness and entanglement in terms of linear functionals that are positive and ``completely'' positive, respectively, on a certain cone of separable Hermitian operators.
Finally, we made explicit for the first time two fundamental characterizations of no-signaling operations, establishing that the spaces spanned by the local operations and by the larger class of no-signaling operations are in fact equal.

We conclude the thesis with some pointers for future research and open problems pertaining to the topics covered.

\section{Quantum strategies}

\begin{description}

\item[New properties.]

Three important properties of quantum strategies were established in Theorems
\ref{theorem:inner-product} (\theoreminnerproduct),
\ref{thm:char} (\thmchar), and
\ref{thm:max-prob} (\thmmaxprob).
Other simpler and more basic properties were noted in Propositions \ref{prop:measure-sum}, \ref{prop:convexity}, \ref{prop:uniqueness}, and \ref{prop:distributive}.
We also saw several properties of a new distance measure for quantum strategies in Chapter \ref{ch:norms}.
What other properties are held by our representation for quantum strategies?

\item[Simplifying proofs.]

In Section \ref{sec:coin-flip} we employed the properties of quantum strategies to provide a simplified proof of Kitaev's bound for strong quantum coin-flipping.
We also noted in Section \ref{subsec:intro:strategies} that Chiribella \emph{et al.}~have provided a short proof by quantum strategies of the impossibility of quantum bit commitment \cite{ChiribellaD+09b}.

Given these examples, one is tempted to believe that the properties of quantum strategies encapsulate many of the critical elements of various proofs involving the exchange of quantum information.
It is reasonable to expect, for example, that our formalism could lead to new or alternate security proofs for various quantum cryptographic protocols, or possibly even to proofs that certain classical protocols are secure against quantum attacks.

For an unsolved example, consider the protocol for weak quantum coin-flipping with arbitrarily small bias given in Ref.~\cite{Mochon07}.
The proof that the exhibited protocol has arbitrarily small bias is very complicated.
Could this proof be simplified by quantum strategies?

\item[New applications.]

Of course, the formalism of quantum strategies is by no means limited to simplifications of existing proofs, as illustrated by the new results established in Sections \ref{sec:game-theory}, \ref{sec:interactive-proofs}, and \ref{sec:app:validity}, and by Chiribella, D'Ariano, Perinotti, and other authors as discussed in Section \ref{subsec:intro:strategies}.
Surely, there is more to add to this list of new applications.

\end{description}

\section{Local operations with shared entanglement}

\begin{description}

\item[Bigger ball of LOSE or LOSR operations.]

The size of the ball of (unnormalized) LOSR operations established in Theorem
\ref{thm:ball} scales as $\Omega\Pa{2^{-m} n^{-3/2}}$.
Given that this quantity already includes a factor of the dimension $n$, is it possible to eliminate the explicit dependence on the number of parties $m$?
(By contrast, for the case of multipartite separable quantum states the dependence on $m$ seems unavoidable \cite{GurvitsB03}.)
Can the exponent on the dimension $n$ be improved?

As mentioned in Section \ref{sec:balls}, it is not clear that there is a
ball of LOSE operations that strictly contains any ball of LOSR operations.
Does such a larger ball exist?

\item[Completely positive super-operators.]

As noted in Remark \ref{rem:LOSE:cp},
the characterization of LOSE operations is interesting because it involves
linear functionals that are not just positive, but
``completely'' positive on the family of
$\Pa{\bQ_1\pa{\cE_1};\dots;\bQ_m\pa{\cE_m}}$-separable cones.

Apparently, the study of completely positive super-operators has until now been
strictly limited to the context of positive semidefinite input cones.
Any question that may be asked of conventional completely positive
super-operators might also be asked of this new class of completely positive
super-operators.

\item[Entanglement required for approximating LOSE operations.]

It was mentioned in Chapters \ref{ch:intro} and \ref{ch:intro-LO} that there exist LOSE operations that
cannot be implemented with any finite amount of shared entanglement
\cite{LeungT+08}.
The natural question, then, is how much entanglement is necessary to achieve an
arbitrarily close approximation to such an operation?

The present author conjectures that for every two-party LOSE operation $\Lambda$ there exists an
$\varepsilon$-approximation $\Lambda'$ of $\Lambda$
in which the dimension of the shared entangled state scales as
$O\pa{2^{\varepsilon^{-a}} n^b}$ for some positive constants $a$ and $b$
and some appropriate notion of $\varepsilon$-approximation.
Here $n=\dim\Pa{\kprod{\bQ}{1}{m}}$ is the dimension of the space in which
$\jam{\Lambda}$ lies.

Evidence pertaining to this conjecture can be found in Refs.~\cite{CleveH+04,KempeR+07,LeungT+08}.
Moreover, the example in Ref.~\cite{LeungT+08} strongly suggests that the exponential dependence on $1/\varepsilon$ is unavoidable; the pressing open question pertains to the dependence upon $n$.
At the moment, \emph{no upper bound at all} is known for this general class of two-party LOSE operations.

\end{description}


\renewcommand{\bibname}{References}
\addcontentsline{toc}{chapter}{\textbf{References}}


\newcommand{\etalchar}[1]{$^{#1}$}

\end{document}